\documentclass[10pt,conference]{IEEEtran}

\usepackage[english]{babel}
\usepackage{amssymb}
\usepackage{amsmath}
\usepackage{amsthm}
\usepackage{enumerate}

\usepackage{graphicx}
\usepackage{semantic}
\usepackage{mathtools}

\theoremstyle{definition}

\newtheorem{definition}{Definition}
\newtheorem{lemma}{Lemma}
\newtheorem{proposition}[lemma]{Proposition}
\newtheorem{corollary}[lemma]{Corollary}
\newtheorem{theorem}{Theorem}

\theoremstyle{remark}
\newtheorem{remark}{Remark}

\theoremstyle{plain}
\newtheorem{example}{Example}

\begin{document}

\title{Branching Bisimilarity on Normed BPA Is EXPTIME-complete}


\author{\IEEEauthorblockN{Chaodong He \ \ \ \  \ \ \   Mingzhang Huang}
\IEEEauthorblockA{BASICS, Department of Computer Science and Engineering,  Shanghai Jiao Tong University
}
}

\maketitle

\begin{abstract}
We put forward an exponential-time algorithm for deciding branching bisimilarity on normed BPA (Bacis Process Algebra) systems.  The decidability of branching (or weak)  bisimilarity on normed BPA was once a long standing open problem which was closed by Yuxi Fu in~\cite{DBLP:conf/icalp/Fu13}. The $\mathrm{EXPTIME}$-hardness is an inference of a slight modification of the reduction presented by Richard Mayr~\cite{Mayr2005}. Our result claims that this problem is $\mathrm{EXPTIME}$-complete.
\end{abstract}

\IEEEpeerreviewmaketitle

\section{Introduction}

Basic process algebra (BPA)~\cite{Baeten:1991:PA:103272} is a fundamental model of infinite state systems, with its famous counterpart in the theory of formal languages: context free grammars in Greibach normal forms, which generate the entire context free languages.
In 1987, Baeten, Bergstra and Klop~\cite{BaetenBergstraKlop1987} proved a
surprising result at the time that strong bisimilarity on normed BPA is decidable.  This result is in sharp contrast to the classical fact that language equivalence is undecidable for context free grammar~\cite{Hopcroft:1990:IAT:574901}.
After this remarkable discovery, decidability and complexity issues of bisimilarity checking on infinite state systems have been intensively investigated. See~\cite{DBLP:conf/concur/JancarM99,Burkart00verificationon,DBLP:journals/iandc/MollerSS04,Srba2004,Kucera2006}
for a number of surveys.

As regards strong bisimilarity on normed BPA,
H\"uttel and Stirling~\cite{DBLP:conf/lics/HuttelS91} improved
the result of Baeten, Bergstra and Klop using a more simplified proof by relating
the strong bisimilarity of two normed BPA processes to the existence of a successful
tableau system.  Later, Huynh and Tian~\cite{DBLP:journals/tcs/HuynhT94} showed
that the problem is in $\Sigma_{2}^{\mathrm{P}}$, the second level of the polynomial hierarchy.
Before long, another significant discovery was made by Hirshfeld,
Jerrum and Moller~\cite{DBLP:journals/tcs/HirshfeldJM96} who showed that the problem can even be decided in  polynomial
time.  Improvements on running time was made later in~\cite{DBLP:conf/mfcs/LasotaR06,DBLP:conf/fsttcs/CzerwinskiL10,CzerwinskiPhD}.

The decidability of strong bisimilarity on general BPA is affirmed by Christensen,
H\"{u}ttel and Stirling~\cite{ChristensenHuttelStirling1992}. 2-$\mathrm{EXPTIME}$ is claimed to be an upper bound by Burkart, Caucal and Steffen~\cite{DBLP:conf/mfcs/BurkartCS95} and is explicitly proven recently by Jan\v{c}ar~\cite{DBLP:journals/corr/abs-1207-2479}.  As to the lower bound, Kiefer~\cite{DBLP:journals/ipl/Kiefer13} achieves $\mathrm{EXPTIME}$-hardness, which is an improvement of the previous $\mathrm{PSPACE}$-hardness obtained by Srba~\cite{DBLP:conf/icalp/Srba02}.

In the presence of silent actions, however, the picture is less clear.  The decidability for both weak bisimilarity and branching bisimilarity on normed BPA was once long standing open problems. For  weak bisimilarity~\cite{Milner1989}, the problem is still open, while for  branching bisimilarity~\cite{GlabbeekW89,GlabbeekW96}, a remarkable discovery is made by Fu~\cite{DBLP:conf/icalp/Fu13} recently that the problem is decidable.   Very recently, using the key property developed in~\cite{DBLP:conf/icalp/Fu13},  Czerwi\'{n}ski and Jan\v{c}ar shows that there exists an exponentially large bisimulation base for branching bisimilarity on normed BPA, and by guessing the base, they show that the complexity of this problem is in $\mathrm{NEXPTIME}$~\cite{DBLP:journals/corr/CzerwinskiJ14}.  The current best lowerbound for weak bisimilarity  is the $\mathrm{EXPTIME}$-hardness established by Mayr~\cite{Mayr2005}, whose proof can be slightly modified to show
the $\mathrm{EXPTIME}$-hardness for branching bisimilarity as well. As to the general BPA, decidability of branching bisimilarity is still unknown.

In this paper, we confirm that an exponential time algorithm exists for checking branching bisimilarity on normed BPA.  Comparing with the known $\mathrm{EXPTIME}$-hardness result, we get the result of $\mathrm{EXPTIME}$-completeness. Thus the complexity class of branching bisimilarity on normed BPA is completely determined.

Basically, we introduce a family of relative bisimilarities parameterized by the reference sets, which can be represented by a decomposition base defined in this paper.  The branching bisimilarity is exactly the relative bisimilarity whose reference set is the empty set.   We show that this base can be approximated.  The approximation procedure starts from an initial base, which is relatively trivial, and is carried on by repeatedly refining the current base.  In order to define the approximation procedure and to ensure that the  family of relative bisimilarities  is achieved at last, a lot of technical difficulties need overcoming. Some of them are listed here:
\begin{itemize}
\item
Despite the seeming resemblance, the
relative bisimilarities (Section~\ref{sec:Relativized_bisimilarity}) defined in this paper is significantly superior to the corresponding concepts in~\cite{DBLP:journals/corr/CzerwinskiJ14}. The relative bisimilarities in this paper is {\em suffix independent}. This property is extremely crucial for our algorithm.  The correctness of definition is characterized in Theorem~\ref{thm:relative_bis_str}.

\item
We show that a generalized unique decomposition property holds for the family of relative bisimilarities (Theorem~\ref{thm:QUDP_RBisimularity}). In the decompositions, bisimilarities with different reference sets depend and impact on each other.  The notion of decomposition bases (Section~\ref{sec:base}) provides an effective representation of an arbitrary family of process equivalences that satisfies the unique decomposition property.

\item
In an iteration of refinement operation, a new decomposition base is constructed from the old (Section~\ref{sec:naive-algorithm}). That is, a new family of equivalences is obtained from the old one. Besides, comparing with all the previous algorithms~\cite{DBLP:journals/mscs/HirshfeldJM96,DBLP:conf/fsttcs/CzerwinskiL10,DBLP:journals/corr/He14a} which take partition refinement approach,  our refinement procedure possesses several hallmarks:
\begin{itemize}
\item
The new base is constructed via a {\em globally} greedy strategy, which means that
all the relevant equivalences with different reference sets are dealt with as a whole.

\item
The refinement operation in previous works heavily depends on predefined notions of norms and decreasing transitions. These notions can be determined from the normed BPA definition immediately. Such a method does not work at present. Our solution is to define norms in a semantic way (Section~\ref{sec:norms}). Norms, relying on the relevant equivalence relations, together with decreasing transitions, can change dynamically in every iteration.   When we start to construct a new base, no information on norms is available. Thus at this time we cannot determine whether a transition is decreasing.  Our solution is to incorporate the task of computing norms into the global iteration procedure via the greedy strategy.

\item
In previous works the order of process constants can be determined in advance.  Every time a new base is constructed from the old, the constants are treated in the same order.   There is no such predefined order in our algorithm. The treating order is dynamically determined in every iteration.
\end{itemize}
\end{itemize}

Equivalence checking on normed BPA is significantly harder than the related problem on {\em totally} normed BPA.    For totally normed BPA, branching bisimilarity is recently shown polynomial-time decidable~\cite{DBLP:journals/corr/He14a}.
What is obtained in this paper is significantly stronger than previous results~\cite{DBLP:conf/cav/Huttel91,DBLP:journals/iandc/CaucalHT95,DBLP:journals/corr/He14a}.

\section{Preliminaries}\label{sec:preliminaries}

\subsection{Normed Basic Process Algebra}
A {\em basic process algebra} (BPA) system $\Gamma$ is a triple $(\mathbf{C}, \mathcal{A}, \Delta)$, where $\mathbf{C}$ is a finite set of process constants  ranged over by $X,Y,Z,U,V,W$, $\mathcal{A}$ is a finite set of actions, and $\Delta$ is a finite set of transition rules.
The {\em processes}, ranged over by $\alpha,\beta,\gamma,\delta,\zeta,\eta$,  are generated by the following grammar:
\[
\alpha \  \Coloneqq  \  \epsilon  \  \mid \  X \  \mid \  \alpha_1 . \alpha_2.
\]
The syntactic equality is denoted by $=$.
We assume that the sequential composition $\alpha_1.\alpha_2$ is associative up to $=$ and $\epsilon . \alpha = \alpha . \epsilon = \alpha$. Sometimes $\alpha . \beta$ is shortened as $\alpha\beta$.  The set of processes is exactly $\mathbf{C}^{*}$, the finite strings over $\mathbf{C}$.
There is a special symbol $\tau$ in $\mathcal{A}$ for silent transition.   $\ell$ is invariably used to denote an arbitrary action, while $a$ is used to denote a visible (i.e. non-silent) action.
The transition rules in $\Delta$ are of the form $X \stackrel{\ell}{\longrightarrow} \alpha$.
The operational semantics of the processes are defined by  the following labelled transition rules.
\[
\begin{array}{c}
    \cfrac{ ( X\stackrel{\ell}{\longrightarrow} \alpha ) \in\Delta}{X\stackrel{\ell}{\longrightarrow}\alpha} \
                    \qquad \cfrac{ \alpha\stackrel{\ell}{\longrightarrow}\alpha'}{\alpha.\beta \stackrel{\ell}{\longrightarrow} \alpha' . \beta}
 \end{array}
\]
A central dot `$\cdot$' is often used to indicate an arbitrary process. For example, we write $\alpha \stackrel{\ell_1}{\longrightarrow} \cdot \stackrel{\ell_2}{\longrightarrow} \beta$, or even $\alpha \stackrel{\ell_1}{\longrightarrow} \stackrel{\ell_2}{\longrightarrow} \beta$, to mean that there exists some $\gamma$ such that $\alpha \stackrel{\ell_1}{\longrightarrow} \gamma$ and $\gamma \stackrel{\ell_2}{\longrightarrow} \beta$.

If $\asymp$ is an equivalence relation on processes, then
we will use $\alpha\stackrel{\asymp}{\longrightarrow}\alpha'$ to denote the fact $\alpha\stackrel{\tau}{\longrightarrow}\alpha'$ and $\alpha \asymp \alpha'$,  and use
$\alpha\stackrel{\not\asymp}{\longrightarrow}\alpha'$ to denote the fact $\alpha\stackrel{\tau}{\longrightarrow}\alpha'$ and $\alpha \not\asymp \alpha'$. We write $\Longrightarrow $ for the reflexive transitive closure of $\stackrel{\tau}{\longrightarrow}$, and  $\Longleftrightarrow$ for the symmetric closure of $\Longrightarrow$ (i.e.~${\Longleftrightarrow} \stackrel{\mathrm{def}}{=}  {\Longrightarrow} \cup {\Longrightarrow}^{-1}$). Accordingly,
$\stackrel{\asymp}{\Longrightarrow}$ is understood as the  reflexive transitive closure of
$\stackrel{\asymp}{\longrightarrow}$.  That is, $\alpha \stackrel{\asymp}{\Longrightarrow}\alpha'$ if and only if
$\alpha \stackrel{\asymp}{\longrightarrow}\cdot \ldots \cdot  \stackrel{\asymp}{\longrightarrow} \alpha'$.

\begin{remark}
$\alpha \stackrel{\asymp}{\Longrightarrow}\alpha'$ is slightly different from $\alpha \Longrightarrow \alpha' \asymp \alpha$.  If Computation Lemma (Lemma~\ref{lem:computation_lemma}) holds for $\asymp$, then $\alpha \stackrel{\asymp}{\Longrightarrow}\alpha'$ if and only if $\alpha \Longrightarrow \alpha' \asymp \alpha$.

\end{remark}


A process $\alpha$ is {\em normed} if $\alpha \stackrel{\ell_1}{\longrightarrow} \cdot \ldots   \cdot \stackrel{\ell_n}{\longrightarrow} \epsilon$ for some $\ell_1, \ldots, \ell_n$.
A BPA system $\Gamma = (\mathbf{C}, \mathcal{A}, \Delta)$ is  normed if all the processes defined in $\Gamma$ are  normed. In other words, $\Gamma$ is normed if $X$ is normed for every $X \in \mathbf{C}$.
In the rest of the paper, we will invariably use $\Gamma = (\mathbf{C}, \mathcal{A}, \Delta)$ to indicate the concerned normed BPA system.  A BPA system  $\Gamma$ is called {\em realtime} if for every $(X \stackrel{\ell}{\longrightarrow} \alpha) \in \Delta$, we have $\ell \neq \tau$.

A process $\alpha$ is called a {\em ground} process if $\alpha \Longrightarrow \epsilon$.  The set of ground constants is denoted by $\mathbf{C}_{\mathrm{G}}$. Apparently $\mathbf{C}_{\mathrm{G}} \subseteq \mathbf{C}$ and $\alpha$ is ground if and only if $\alpha \in \mathbf{C}_{\mathrm{G}}^{*}$.

\begin{remark}
A BPA system $\Gamma = (\mathbf{C}, \mathcal{A}, \Delta)$ is  {\em totally normed} if and only if rules of the form $X  \stackrel{\tau}{\longrightarrow} \epsilon$ are forbidden.  $\Gamma = (\mathbf{C}, \mathcal{A}, \Delta)$ is totally normed if and only if $\mathbf{C}_{\mathrm{G}} = \emptyset$.
\end{remark}

\subsection{Bisimulation and Bisimilarity}
In the presence of silent actions, branching bisimilarity of van Glabbeek and Weijland~\cite{GlabbeekW89,GlabbeekW96} is well-known.
\begin{definition}\label{def:beq}
Let $\asymp$ be an equivalence relation on processes. $\asymp$ is called a {\em branching bisimulation},  if the following {\em bisimulation property} hold:  whenever $\alpha \asymp \beta$,
\begin{itemize}
\item
If $\alpha \stackrel{a}{\longrightarrow} \alpha'$, then
    $\beta \stackrel{\asymp} \Longrightarrow \cdot \stackrel{a}{\longrightarrow} \beta'$ for some $\beta'$ such that  $\alpha'\asymp\beta'$.
\item
If $\alpha \stackrel{\not\asymp}{\longrightarrow} \alpha'$, then
   $\beta \stackrel{\asymp}{\Longrightarrow} \cdot \stackrel{\not\asymp}{\longrightarrow} \beta'$ for some $\beta'$ such that  $\alpha'\asymp\beta'$.
\end{itemize}
The {\em branching bisimilarity}  $\simeq$ is the largest branching bisimulation.
\end{definition}

\begin{remark}
In this paper, branching bisimulations in Definition~\ref{def:beq} and other bisimulation-like relations in later chapters are forced to be  equivalence relations.  This technical convention does not affect the notion of branching bisimilarity.
\end{remark}

The branching bisimilarity is a congruence relation, and it
satisfies the following famous  lemma.
\begin{lemma}[Computation Lemma~\cite{GlabbeekW89}]\label{lem:computation_lemma}
If $\alpha \Longrightarrow \alpha'\Longrightarrow  \alpha'' \simeq \alpha$, then $\alpha'\simeq \alpha$.
\end{lemma}

If $\Gamma$ is realtime,  the branching bisimilarity is the same as the {\em strong bisimilarity}.  In this paper, branching bisimilarity will be abbreviated as {\em bisimilarity}. For realtime systems, the term bisimilarity will also be used to indicate strong bisimilarity.

\section{Relativized Bisimilarities on Normed BPA}\label{sec:Relativized_bisimilarity}

\subsection{Retrospection}\label{subsec:retrospection}
In~\cite{DBLP:conf/icalp/Fu13}, Yuxi Fu creates the notion of redundant processes, and discover the following Proposition~\ref{prop:redundantset_characterize_bisimilarity}, which is crucial to the proof of decidability of bisimilarity for normed BPA.

\begin{definition}\label{def:redundant_process}
A process $\alpha$ is a {\em $\simeq$-redundant over $\gamma$} if $\alpha\gamma \simeq \gamma$.
\end{definition}
We use $\mathsf{Rd}(\gamma) = \{ X \,|\, X \gamma \simeq \gamma \}$ to indicate the set of all constants that is $\simeq$-redundant over $\gamma$. Clearly, $\mathsf{Rd}(\gamma) \subseteq \mathbf{C}_{\mathrm{G}}$.

The following lemma confirms that the redundant processes over $\gamma$ are completely determined by the redundant constants.
\begin{lemma}\label{lem:redundant_constants}
$\alpha\gamma \simeq \gamma$ if and only if $\alpha \in (\mathsf{Rd}(\gamma))^{*}$.
\end{lemma}

The crucial observation in~\cite{DBLP:conf/icalp/Fu13} is the following fact.
\begin{proposition}\label{prop:redundantset_characterize_bisimilarity}
Assume that $\mathsf{Rd}(\gamma_1) = \mathsf{Rd}(\gamma_2)$, then  $\alpha\gamma_1 \simeq \beta\gamma_1$ if and only if $\alpha\gamma_2 \simeq \beta\gamma_2$.
\end{proposition}
Proposition~\ref{prop:redundantset_characterize_bisimilarity} inspires us to define a relativized version of bisimilarity $\simeq_{R}$ for   a given suitable {\em reference set $R$},  which will satisfy the following theorem.
\begin{theorem}\label{thm:relative_bis}
Let $\gamma$ be a process satisfying $\mathsf{Rd}(\gamma) = R$.  Then $\alpha \simeq_{R} \beta$ if and only if $\alpha \gamma \simeq \beta\gamma$.
\end{theorem}
Proposition~\ref{prop:redundantset_characterize_bisimilarity} confirms that $\simeq_{R}$ does not depend on the special choice of $\gamma$ under the assumption of the existence of $\gamma$ such that $R = \mathsf{Rd}(\gamma)$.  However, it is much  wiser not to take Theorem~\ref{thm:relative_bis} as the definition of $\simeq_{R}$ from a computational point of view. Here are the reasons.
\begin{itemize}
\item
We cannot tell beforehand (except when we can decide $\simeq$) whether, for a given $R$, there exists $\gamma$ such that $R = \mathsf{Rd}(\gamma)$, nor can we tell whether $R = \mathsf{Rd}(\gamma)$ even if both $R$ and $\gamma$ are given.

\item
The algorithm developed in this paper takes the refinement approach. Imagine that $\asymp$ is an approximation of  $\simeq$, we can define, for example, the $\asymp$-redundant constants $\mathsf{Rd}^{\asymp}(\gamma)$ accordingly. It is quite possible to run into the situation where, for a specific $R$, there is no  $\delta$ such that $R = \mathsf{Rd}^{\asymp}(\delta)$ even if $R = \mathsf{Rd}(\gamma)$ for some $\gamma$.
\end{itemize}
Therefore it is advisable to make $\simeq_R$ well-defined for every $R$ satisfying $R\subseteq \mathbf{C}_{G}$. Importantly,
$\simeq_{R}$ should  be defined without the knowledge of the existence of $\gamma$.

\begin{remark}
In~\cite{DBLP:journals/corr/CzerwinskiJ14},  Czerwi\'{n}ski and Jan\v{c}ar also define a relativized version of bisimilarities. The difference is that they directly take Theorem~\ref{thm:relative_bis} as the definition. After that, they establish a {\em weaker} version of  unique decomposition property.
In~\cite{DBLP:journals/corr/CzerwinskiJ14},  $\simeq_R$ is defined only for those $R$'s such that $R = \mathsf{Rd}(\gamma)$ for some $\gamma$.  Using Theorem~\ref{thm:relative_bis}, a property of $\simeq_R$ can be proved by a translation of a property of $\simeq$.  Though seemingly similar,  the properties which will be developed in this section are much stronger than those properties in~\cite{DBLP:journals/corr/CzerwinskiJ14}.
\end{remark}

\subsection{Definition of $R$-Bisimilarities}

Now we elaborate on the definition of $\simeq_R$.
Some auxiliary notations are introduced  to make things clear.
\begin{definition}\label{def:R_equal}
Let $R \subseteq \mathbf{C}_{\mathrm{G}}$.
Two processes $\alpha$ and $\beta$ are {\em $R$-equal}, denoted by $\alpha =_{R} \beta$ if there exist  $\zeta, \alpha', \beta'$ such that $\alpha = \zeta\alpha'$, $\beta = \zeta\beta'$, and $ \alpha', \beta' \in R^{*}$.
\end{definition}
Two processes are $R$-equal if they differ only in suffixes in $R^{*}$.
$R$-equality is an equivalence relation.   Eliminating a suffix in $R^{*}$ from a process does not change the $=_{R}$-class.
\begin{lemma}\label{lem:R_equal_elimination}
\begin{enumerate}
\item
$\alpha =_{R}  \alpha\gamma$ if and only if $\gamma \in R^{*}$.

\item
$\alpha =_{R} \epsilon$ if and only if $\alpha \in R^{*}$.
\end{enumerate}
\end{lemma}

\begin{definition}\label{def:R_nf}
Let $R \subseteq \mathbf{C}_{\mathrm{G}}$.
$\alpha$ is in {\em $R$-normal-form} ({\em $R$-nf}) if
\begin{enumerate}
\item
either $\alpha = \epsilon$,

\item
or there exist  $\alpha'$ and $X$ such that $\alpha = \alpha' X$ and $X\not\in R$.
\end{enumerate}
If $\alpha =_{R} \alpha'$ and $\alpha'$ is in $R$-nf, then $\alpha'$ is called an {\em $R$-nf} of $\alpha$. The (unique) $R$-nf of $\alpha$ is denoted by $\alpha_R$.
\end{definition}
From Definition~\ref{def:R_nf}, taking the $R$-nf of $\alpha$ is nothing but removing any suffix of $\alpha$ in $R$.
$R$-equality is the syntactic equality on $R$-nf's. In particular, $\emptyset$-equality is exactly the ordinary syntactic equality.

\begin{lemma}
$\alpha =_R \beta$ if and only if $\alpha_R = \beta_R$. In particular,
$\alpha =_{\emptyset} \beta$ if and only if $\alpha = \beta$.
\end{lemma}

The transition relations can be relativized as follows.
\begin{definition}\label{def:R_transition}
The {\em $R$-transition relations} between $R$-nf's are defined as follows:
We write $\zeta \stackrel{\ell}{\longrightarrow}_{R} \eta$ if there exists $\alpha$ and $\beta$ such that
$\zeta = \alpha_R$, $\eta = \beta_R$, and $\alpha \stackrel{\ell}{\longrightarrow} \beta$.
\end{definition}
According to Definition~\ref{def:R_transition}, the relation $\stackrel{\ell}{\longrightarrow}_R$ is defined only on the set of processes in $R$-nf.  When we write $\alpha \stackrel{\ell}{\longrightarrow}_R \beta$, $\alpha$ and $\beta$ are implicitly supposed to be $R$-nf's.
\begin{lemma}\label{lem:R_transition}
$\alpha_R \stackrel{\ell}{\longrightarrow}_R  \beta_R$  if and only if $\alpha  =_{R} \cdot \stackrel{\ell}{\longrightarrow}  \cdot =_{R} \beta$.
\end{lemma}
Let $\alpha \stackrel{\ell}{\longrightarrow}_{R} \beta$. Intuitively, if $\alpha \neq \epsilon$, then  $\alpha \stackrel{\ell}{\longrightarrow}_{R} \beta$ is induced by $\alpha$; if $\alpha = \epsilon$, then $\alpha \stackrel{\ell}{\longrightarrow}_{R} \beta$ is induced by one of the constants in $R$. This important fact is formalized in the following lemma.
\begin{lemma}\label{lem:char_R_arrow}
$\alpha_R \stackrel{\ell}{\longrightarrow}_{R} \beta_R$ if and only if
\begin{enumerate}
\item
either $\alpha_R  = \epsilon$ and $X \stackrel{\ell}{\longrightarrow} \beta'$ for some $\beta'$ and $X \in R$ such that $\beta' =_R \beta$.

\item
or  $\alpha_R  \neq \epsilon$ and $\alpha \stackrel{\ell}{\longrightarrow} \beta'$ for some $\beta'$  such that $\beta' =_R \beta$.
\end{enumerate}
\end{lemma}

As usual, we write $\Longrightarrow_R$ for the reflexive transitive closure of $\stackrel{\tau}{\longrightarrow}_R$, and $\Longleftrightarrow_R$ for the symmetric closure of $\Longrightarrow_R$ (i.e.~${\Longleftrightarrow_R} \stackrel{\mathrm{def}}{=}  {\Longrightarrow_R} \cup {\Longrightarrow_R}^{-1}$). Accordingly,
$\alpha \stackrel{\asymp}{\longrightarrow}_R \alpha'$ is understood as
$\alpha \stackrel{\tau}{\longrightarrow}_R \alpha'\asymp \alpha$, and $\stackrel{\asymp}{\Longrightarrow}_R $ is the reflexive transitive closure of $\stackrel{\asymp}{\longrightarrow}_R$.

The ground processes are robust under relativization.
\begin{lemma}\label{lem:groud_preserve}
$\alpha \Longrightarrow \epsilon$ if and only if $\alpha_R \Longrightarrow_R \epsilon$.
\end{lemma}

Now it is time for defining $R$-bisimilarity.

\begin{definition}\label{def:R_beq}
Let $R \subseteq \mathbf{C}_{\mathrm{G}}$ and let $\asymp$ be an equivalence relation such that
${=_{R}} \subseteq {\asymp}$. We say $\asymp$ is an {\em $R$-bisimulation},  if the following conditions are satisfied whenever $\alpha \asymp \beta$:
\begin{enumerate}
\item
{\em ground preservation}:  If $\alpha \Longrightarrow \epsilon$,  then $\beta \Longrightarrow \epsilon$.

\item
If $\alpha \stackrel{\not\asymp}{\longrightarrow} \alpha'$, then
$\beta_R  \stackrel{\asymp}{\Longrightarrow}_{R} \cdot \stackrel{\not\asymp}{\longrightarrow}_{R} \beta'$ for some $\beta'$ such that $\alpha' \asymp \beta'$.

\item
If  $\alpha \stackrel{a}{\longrightarrow} \alpha'$, then
$\beta_R \stackrel{\asymp}{\Longrightarrow}_{R} \cdot \stackrel{a}{\longrightarrow}_{R} \beta'$ for some $\beta'$ such that $\alpha' \asymp \beta'$.
\end{enumerate}
The {\em $R$-bisimilarity} $\simeq_{R}$ is the largest $R$-bisimulation.
\end{definition}

\begin{remark}
$R$-bisimilarity $\simeq_R$ is well-defined, based on the following observations:
\begin{itemize}
\item
$=_R$ is an $R$-bisimulation.

\item
If $\asymp_1$ and $\asymp_2$ are both $R$-bisimilations, then $({\asymp_1} \cup {\asymp_2})^{*}$ is an $R$-bisimulation.
\end{itemize}
\end{remark}

If $R = \emptyset$, then $\simeq_{\emptyset}$ is exactly the ordinary bisimilarity $\simeq$.

$R$-bisimulations  can actually be understood as the bisimulations on $R$-nf's under $R$-transitions, as is stated below.
\begin{proposition}\label{prop:R_bisimilarity}
Let $R \subseteq \mathbf{C}_{\mathrm{G}}$ and let $\asymp$ be an equivalence relation such that
${=_{R}} \subseteq {\asymp}$. Then $\asymp$ is an {\em $R$-bisimulation} if and only if whenever $\alpha \asymp \beta$,
\begin{enumerate}
\item
if  $\alpha_R \Longrightarrow_R \epsilon$,  then $\beta_R \Longrightarrow_R \epsilon$;

\item
if $ \alpha_R \stackrel{\not\asymp}{\longrightarrow}_R \alpha'_R$, then
$\beta_R  \stackrel{\asymp}{\Longrightarrow}_{R} \cdot \stackrel{\not\asymp}{\longrightarrow}_{R} \beta'_R$ for some $\beta'$ such that $\alpha' \asymp \beta'$;

\item
if $\alpha_R \stackrel{a}{\longrightarrow}_R \alpha'_R$, then
$\beta_R \stackrel{\asymp}{\Longrightarrow}_{R} \cdot \stackrel{a}{\longrightarrow}_{R} \beta'_R$ for some $\beta'$ such that $\alpha' \asymp \beta'$.
\end{enumerate}
\end{proposition}
Comparing with the definition of bisimulation (Definition~\ref{def:beq}),  Definition~\ref{def:R_beq} and Proposition~\ref{prop:R_bisimilarity} contains an extra {\em ground preservation} condition which guarantees that a ground process cannot be related to a non-ground process in an $R$-bisimilation.
In the definition of bisimulation, this condition is also satisfied, for it can be derived from other bisimulation conditions. As to $R$-bisimulation, this is not always the case, as is illustrated in the following example.
\begin{example}
Consider the following normed BPA $(\mathbf{C}, \mathcal{A}, \Delta)$:
\begin{itemize}
\item
 $\mathbf{C} = \{A_0, A_1\}$;

\item
$\mathcal{A} = \{a,  \tau\}$;

\item
$\Delta$ is the set containing the following rules:
\begin{center}
$A_0 \stackrel{a}{\longrightarrow} A_1$, \quad $A_1 \stackrel{a}{\longrightarrow} A_1$, \quad $A_1 \stackrel{\tau}{\longrightarrow} \epsilon$
\end{center}
\end{itemize}
Let $R = \{A_1\}$, and let $\asymp$ be the equivalence relation which relates every processes defined in $\Gamma$ to  $\epsilon$. Clearly $A_0 \not\Longrightarrow_R \epsilon$. However, we can show that $(A_0, \epsilon)$ satisfies $R$-bisimulation conditions except for the ground preserving condition:

Considering that the $R$-transitions of $\epsilon$ can be trivially matched by $A_0$, it remains to show that  $\epsilon$ can match the $R$-transitions of $A_0$.  The unique $R$-transition of $A_0$ is $A_0  \stackrel{a}{\longrightarrow}_R \epsilon$, which can be matched by $\epsilon \stackrel{a}{\longrightarrow}_R \epsilon$ since $A_1 \stackrel{a}{\longrightarrow} A_1$ and $(A_1)_R = \epsilon$.
\end{example}
The relative bisimilarity $\simeq_R$ is not a congruence in general. For example, we may not have $\alpha\gamma \simeq_R \beta\gamma$ even if $\alpha \simeq_R \beta$.  However, we have the following result.
\begin{lemma}\label{lem:congruence_R_bis}
If $\gamma \simeq_R \delta$ and $\alpha \simeq \beta$, then $\alpha\gamma \simeq_R \beta\delta$.  In particular, If $\gamma \simeq_R \delta$, then $\alpha \gamma \simeq_R \alpha\delta$.
\end{lemma}
The computation lemma also holds for $\simeq_R$.
\begin{lemma}[Computation Lemma for $\simeq_R$]\label{lem:computation_lemma_R}
If $\alpha \Longrightarrow_R \alpha'\Longrightarrow_R  \alpha'' \simeq_R \alpha$ then $\alpha'\simeq_R \alpha$.
\end{lemma}

\subsection{$R$-identities and Admissible Reference Sets}
  Clearly, $R$-bisimilarity has  the following basic property.
\begin{lemma}\label{lem:Rbisimilarity_R}
Let $R \subseteq \mathbf{C}_{\mathrm{G}}$. If $X \in R$, then $X \simeq_R \epsilon$.
\end{lemma}
Be aware that the converse of Lemma~\ref{lem:Rbisimilarity_R} does not hold in general. That is, if $X \simeq_R \epsilon$, there is no guarantee that $X \in R$. This basic observation leads to further discussion.
\begin{definition}\label{def:Id_Rd_R}
Let $R \subseteq \mathbf{C}_{\mathrm{G}}$.
A process $\alpha$ is  called a  {\em $\simeq_{R}$-identity} if $\alpha  \simeq_{R} \epsilon$. We use $\mathsf{Id}_{R}$ to denote $\{ X \;|\;  X  \simeq_{R} \epsilon \}$.
\end{definition}
By Lemma~\ref{lem:Rbisimilarity_R} and Definition~\ref{def:R_beq},
$R \subseteq \mathsf{Id}_{R} \subseteq \mathbf{C}_{\mathrm{G}}$.  Moreover,
\begin{lemma}\label{lem:Id_def}
$\alpha  \simeq_{R} \epsilon$ if and only if $\alpha \in (\mathsf{Id}_{R})^{*}$.
\end{lemma}

Below we will demonstrate that, as a reference set, $\mathsf{Id}_{R}$ plays an important role. At first we state a useful proposition for relative bisimilarities. It says that $\simeq_R$ is monotone.
\begin{proposition}\label{prop:R_monotone}
Let $R_1 \subseteq R_2 \subseteq \mathbf{C}_{\mathrm{G}}$.  If $\alpha \simeq_{R_1} \beta$, then $\alpha \simeq_{R_2} \beta$.
\end{proposition}
\begin{corollary}
Let $R_1 \subseteq R_2 \subseteq \mathbf{C}_{\mathrm{G}}$.    Then,  $\mathsf{Id}_{R_1} \subseteq \mathsf{Id}_{R_2}$.
\end{corollary}
Intuitively, $\simeq_R$ is the relative bisimilarity which is induced by regarding the constants in $R$ as $\epsilon$ purposely. It is reasonable to expect that $X \simeq_{ \mathsf{Id}_{R}} \epsilon$ if and only if $X \in \mathsf{Id}_{R}$. This intuition is confirmed by Proposition~\ref{prop:R_bis_vs_IDR} and its corollaries.
\begin{proposition}\label{prop:R_bis_vs_IDR}
$\alpha \simeq_R \beta$ if and only if $\alpha \simeq_{\mathsf{Id}_{R}} \beta$.
\end{proposition}
\begin{corollary}\label{coro:monotone_1}
Let $R_1 \subseteq \mathbf{C}_{\mathrm{G}}$ and $R_2 \subseteq \mathbf{C}_{\mathrm{G}}$. If  $\mathsf{Id}_{R_1} = \mathsf{Id}_{R_2}$, then $\alpha \simeq_{R_1} \beta$ if and only if $\alpha \simeq_{R_2} \beta$.
\end{corollary}
\begin{corollary}\label{coro:monotone_2}
 Let $R,S \subseteq \mathbf{C}_{\mathrm{G}}$ such that $R \subseteq S \subseteq \mathsf{Id}_{R}$, then $\mathsf{Id}_{S} =  \mathsf{Id}_{R}$.
\end{corollary}
A direct inference of Corollary~\ref{coro:monotone_2} is the following fact.
\begin{lemma}
$X \simeq_{ \mathsf{Id}_{R}} \epsilon$ if and only if $X \in \mathsf{Id}_{R}$. In other words,
$\mathsf{Id}_{\mathsf{Id}_{R}} = \mathsf{Id}_{R}$.
\end{lemma}
The above discussions lead to the following definition.
\begin{definition}\label{def:admissible}
An $R \subseteq \mathbf{C}_{\mathrm{G}}$ is  called {\em admissible} if $R = \mathsf{Id}_{R}$.
\end{definition}
The significance of Proposition~\ref{prop:R_monotone}, Proposition~\ref{prop:R_bis_vs_IDR}, and their corollaries is the revelation of the following fact: The set $\{\simeq_R\}_{R  \subseteq \mathbf{C}_{\mathrm{G}}}$ of all relative bisimilarities is completely determined by those  $\simeq_R$'s in which $R$ is admissible.
\begin{lemma}
For every $R \subseteq \mathbf{C}_{\mathrm{G}}$,  $\mathsf{Id}_{R}$ is admissible.
$\mathsf{Id}_{R}$ is the smallest admissible set which contains $R$.
\end{lemma}

\subsection{$R$-redundant Constants}

The properties of $\simeq$-redundant processes (Definition~\ref{def:redundant_process} and Proposition~\ref{prop:redundantset_characterize_bisimilarity}) in Section~\ref{subsec:retrospection}  can now be generalized for the relative bisimilarity $\simeq_R$.

\begin{definition}\label{def:Id_Rd_R}
Let $R \subseteq \mathbf{C}_{\mathrm{G}}$.
A process $\alpha$ is {\em $\simeq_{R}$-redundant over $\gamma$} if $\alpha\gamma  \simeq_{R} \gamma$.   We use $\mathsf{Rd}_{R}(\gamma)$ to denote $\{ X \;|\;  X \gamma \simeq_{R} \gamma \}$.
\end{definition}
Note that $\mathsf{Rd}(\gamma)$ defined in Section~\ref{subsec:retrospection} is exactly $\mathsf{Rd}_{\emptyset}(\gamma)$. Also note that $\mathsf{Id}_{R}$ is the same as $\mathsf{Rd}_{R}(\epsilon)$.

\begin{lemma}\label{lem:Rd_R_basic_facts}
If $\gamma \simeq_R \delta$, then $\mathsf{Rd}_{R}(\gamma) = \mathsf{Rd}_{R}(\delta)$.
\end{lemma}
\begin{lemma}\label{lem:redundant_constants_relative}
\begin{enumerate}
\item
$\alpha  \simeq_{R} \epsilon$ if and only if $\alpha \in (\mathsf{Id}_{R})^{*}$.

\item
$\alpha\gamma  \simeq_{R} \gamma$ if and only if $\alpha \in (\mathsf{Rd}_{R}(\gamma))^{*}$.
\end{enumerate}
\end{lemma}
Lemma~\ref{lem:Rd_R_basic_facts} is a direct inference of Lemma~\ref{lem:congruence_R_bis}. Lemma~\ref{lem:redundant_constants_relative} is the strengthened version of  Lemma~\ref{lem:redundant_constants}.

Now we can state the fundamental theorem for $\simeq_{R}$.
\begin{theorem}\label{thm:relative_bis_str}
Let $R' = \mathsf{Rd}_{R}(\gamma)$, then $\alpha \simeq_{R'} \beta$ if and only if $\alpha \gamma \simeq_{R} \beta\gamma$.
\end{theorem}
\begin{proposition}\label{prop:redundantset_characterize_bisimilarity_relative}
Assume that $\mathsf{Rd}_{R}(\gamma_1) = \mathsf{Rd}_R(\gamma_2)$, then  $\alpha\gamma_1 \simeq_R \beta\gamma_1$ if and only if $\alpha\gamma_2 \simeq_R \beta\gamma_2$.
\end{proposition}
\begin{proposition}\label{prop:Rd_congruence}
Suppose that $\gamma \simeq_R \delta$ and let  $R' = \mathsf{Rd}_{R}(\gamma) = \mathsf{Rd}_{R}(\delta)$.  Then $\alpha\gamma \simeq_R \beta\delta$ if and only if $\alpha \simeq_{R'} \beta$.
\end{proposition}
Theorem~\ref{thm:relative_bis_str} and Proposition~\ref{prop:redundantset_characterize_bisimilarity_relative} are the strengthened versions of
Theorem~\ref{thm:relative_bis} and  Proposition~\ref{prop:redundantset_characterize_bisimilarity}. Proposition~\ref{prop:Rd_congruence} is an inference of   Lemma~\ref{lem:congruence_R_bis} and Theorem~\ref{thm:relative_bis_str}. Theorem~\ref{thm:relative_bis_str} and Proposition~\ref{prop:Rd_congruence} act as the relativized version of the congruence property and the cancellation law.

The following lemma is an inference of Theorem~\ref{thm:relative_bis_str}.
\begin{lemma}\label{lem:redundant_chain}
$\mathsf{Rd}_{\mathsf{Rd}_{R}(\delta)}(\gamma) = \mathsf{Rd}_{R}(\gamma \delta)$.
\end{lemma}

In the following we discuss the significance of the admissible reference sets. First it is easy to see the following fact according to Proposition~\ref{prop:R_bis_vs_IDR}.
\begin{lemma}\label{lem:Rd_id}
$\mathsf{Rd}_{R}(\gamma) = \mathsf{Rd}_{\mathsf{Id}_R}(\gamma)$ for every $\gamma$ and $R$.
\end{lemma}
The following lemma ensures that the admissible set is preserved under the `redundant' operation.
\begin{lemma}\label{lem:Rd_admissible}
If  $R = \mathsf{Rd}_{R'}(\gamma)$ for some $\gamma$, then $R$ is admissible.
\end{lemma}

\begin{remark}
Even if $R$ is admissible, it is not guaranteed that $R = \mathsf{Rd}_{R'}(\gamma)$ for some $\gamma$ and $R'$.  This fact indicates that, even if $\simeq_R$ is only attractive for only admissible $R$'s, our notion of $R$-bisimilarities strictly generalizes the ones in~\cite{DBLP:journals/corr/CzerwinskiJ14}.
\end{remark}

\subsection{Unique Decomposition Property for $R$-bisimilarities}\label{subsec:decomposition}

When $\Gamma$ is realtime, the set $\mathbf{C}$ of process constants can be divided into two disjoint sets: {\em primes} $\mathsf{Pr}$ and {\em composites} $\mathsf{Cm}$.   Every process $\alpha$ is bisimilar to a sequential composition of prime constants $P_1 . \ldots . P_r$, and moreover, the prime decomposition is unique (up to bisimilarity). That is, if $P_1 . \ldots . P_r \simeq Q_1 . \ldots . Q_s$, then $r = s$ and $P_i \simeq Q_i$ for every $1 \leq i \leq r$. This property is called {\em unique decomposition property}, which is first established by Hirshfeld \textit{et al.} in~\cite{DBLP:journals/tcs/HirshfeldJM96}.
When $\Gamma$ is totally normed, the unique decomposition property still holds~\cite{DBLP:journals/corr/He14a}.

If $\Gamma$ is not totally normed, the unique decomposition property in the above sense does not hold due to the existence of redundant processes.   However,  we expound that, apart from the existence of redundant constants, the relative bisimilarities $\{\simeq_R\}$ enjoys a `weakened' version of unique decomposition property (Theorem~\ref{thm:QUDP_RBisimularity}), which is still called {\em unique decomposition property} in this paper.

\begin{definition}\label{def:R_prime}
Let $R \subseteq \mathbf{C}_{\mathrm{G}}$, and $X \in \mathbf{C}$.
\begin{itemize}
\item
$X$ is a {\em $\simeq_R$-composite} if $X \simeq_{R} \alpha X'$ for some $X'$ and $\alpha$ such that $X'  \not\in \mathsf{Id}_R $ and $\alpha  \not\in \mathsf{Rd}_R (X')$.

\item
$X$ is a {\em $\simeq_R$-prime} if  $X$ is neither a ${\simeq}_R$-identity nor a $\simeq_{R} $-composite.
\end{itemize}
\end{definition}
According to Definition~\ref{def:R_prime}, a constant $X \in \mathbf{C}$ must act as one of the three different roles:  ${\simeq}_R$-identity, ${\simeq}_R$-composite, or ${\simeq}_R$-prime. We will use $\mathsf{Pr}_R$ and
$\mathsf{Cm}_R$ to indicate the set of ${\simeq}_R$-primes and ${\simeq}_R$-composites, respectively. According to Proposition~\ref{prop:R_bis_vs_IDR}, $\mathsf{Pr}_R = \mathsf{Pr}_{\mathsf{Id}_{R}}$ and $\mathsf{Cm}_R = \mathsf{Cm}_{\mathsf{Id}_{R}}$.

\begin{definition}\label{def:R_prime_decomposition}
We call $P_r. P_{r-1}.  \ldots  .  P_{1}$ a {\em $\simeq_{R}$-prime decomposition} of $\alpha$, if $\alpha \simeq_{R} P_r.  P_{r-1}.  \ldots  .  P_{1}$, and $P_i$ is a $\simeq_{R_i}$-prime for $1 \leq i \leq r$, if $R_1 = \mathsf{Id}_{R}$ and $R_{i+1} = \mathsf{Rd}_{R_{i}}(P_{i})$ for $1 \leq i \leq r-1$.
\end{definition}
Note that according to Lemma~\ref{lem:Rd_admissible} every $R_i$ for $1 \leq i \leq r$ is admissible.

The following `relativized prime process property' is crucial to the unique decomposition property (Theorem~\ref{thm:QUDP_RBisimularity}).

\begin{lemma}\label{lem:R_prime_property}
Suppose that $X$, $Y$ are $\simeq_{R}$-primes and $\alpha X \simeq_{R} \beta Y$. Then $X \simeq_{R} Y$.
\end{lemma}
\begin{theorem}[Unique Decomposition Property for $R$-bisimilari\-ties]\label{thm:QUDP_RBisimularity}
Let $P_r  .  P_{r-1}  .  \ldots  .  P_{1}$ and $Q_s  .  Q_{s-1}  .  \ldots  .  Q_{1}$ be $\simeq_R$-prime decompositions. Let $R_1, S_1 =  \mathsf{Id}_{R}$ and let $R_{i+1} = \mathsf{Rd}_{R_{i}}(P_i)$ for $1 \leq i < r$ and $S_{j +1} = \mathsf{Rd}_{S_{j}}(Q_j)$ for $1 \leq j < s$.  Then, $r = s$, $R_i = S_i$ and $P_i \simeq_{R_{i}} Q_i$ for $1 \leq i \leq r$.
\end{theorem}

\section{Norms and Decreasing Bisimulations}\label{sec:norms}

\subsection{Syntactic Norms vs. Semantic Norms}
When $\Gamma$ is realtime, a natural number called {\em norm} is assigned to every process. The {\em norm} of $\alpha$ is the least number $k$ such that $\alpha   \stackrel{a_1}{\longrightarrow} \cdot \stackrel{a_2}{\longrightarrow}  \ldots   \stackrel{a_k} \longrightarrow \epsilon$ for some $a_1, a_2,\ldots, a_k$.

The norm for realtime systems is both syntactic (static) and semantic (dynamic).  It is syntactic because its definition does not rely on bisimilarity, and it can be efficiently calculated via greedy strategy merely with the knowledge of rules in $\Delta$. It is semantic, because
the norm of a realtime process $\alpha$ is the least number $k$ such that $\alpha   \simeq  \cdot \stackrel{a_1}{\longrightarrow} \cdot \simeq \cdot \stackrel{a_2}{\longrightarrow} \cdot
   \simeq \ldots  \simeq \cdot \stackrel{a_k}\longrightarrow \cdot \simeq \epsilon$ for some $a_1, a_2,\ldots, a_k$.

Therefore, we get the coincidence of the {\em syntactic norm} and {\em semantic norm} for realtime systems.
For non-realtime systems, however, the syntactic norms and the semantic ones do not coincide any more. They must be studied separately.

\subsection{Strong Norms and Weak Norms}

We define two syntactic norms for non-realtime systems. The strong norm takes silent actions into account while the weak norm neglects the contribution of  silent actions.
\begin{definition}
The {\em strong norm} of $\alpha$, denoted by $|\alpha|_{\mathrm{st}}$, is the least number $k$ such that $\alpha   \stackrel{\ell_1}{\longrightarrow} \cdot \stackrel{\ell_2}{\longrightarrow}  \ldots   \stackrel{\ell_k} \longrightarrow \epsilon$ for some $\ell_1, \ell_2,\ldots, \ell_k$.

The {\em weak norm} of $\alpha$, denoted by $|\alpha|_{\mathrm{wk}}$, is the least number $k$ such that $\alpha   \stackrel{a_1}{\Longrightarrow} \cdot \stackrel{a_2}{\Longrightarrow}  \ldots   \stackrel{a_k} \Longrightarrow \epsilon$ for some $a_1, a_2,\ldots, a_k$.
\end{definition}
%




\begin{lemma}
\begin{enumerate}
\item
$| \epsilon |_{\mathrm{st}} = | \epsilon |_{\mathrm{wk}} = 0$;

\item
$| \alpha\beta |_{\mathrm{st}} = | \alpha |_{\mathrm{st}} +  | \beta |_{\mathrm{st}}$;
$| \alpha\beta|_{\mathrm{wk}} = | \alpha |_{\mathrm{wk}} +  | \beta |_{\mathrm{wk}}$.
\end{enumerate}
\end{lemma}
\begin{lemma}
\begin{enumerate}
\item
$| \alpha |_{\mathrm{st}} = 0$ if and only if $\alpha = \epsilon$.

\item
$| \alpha |_{\mathrm{wk}} = 0$ if and only if $\alpha \in \mathbf{C}_{\mathrm{G}}^{*}$ (i.e.~$\alpha \Longrightarrow \epsilon$).
\end{enumerate}
\end{lemma}

\begin{lemma}
If $\alpha \simeq_R \beta$, then $ | \alpha |_{\mathrm{wk}} = | \beta |_{\mathrm{wk}}$.
\end{lemma}

\begin{lemma}\label{exponential_bound_st}
$| X |_{\mathrm{st}}$ is exponentially bounded for every $X$.
\end{lemma}

\subsection{The Semantic Norms}
The semantic norms play an important role in our algorithm. They depend on the involved semantic equivalence.
Let $\asymp$ be a process equivalence. A transition $\alpha\stackrel{\ell}{\longrightarrow}\alpha'$ is called {\em $\asymp$-preserving}  if $\alpha' \asymp \alpha$.

\begin{definition}\label{def:semantic_norm}
Let $\asymp$ be a process equivalence. The {\em $\asymp$-norm} of $\alpha$, denoted by $\|\alpha\|_{\asymp}$, is the least number $k$, such that
\[
 \alpha   \asymp  \cdot \stackrel{\ell_1}{\longrightarrow} \cdot \asymp \cdot \stackrel{\ell_2}{\longrightarrow} \cdot
   \asymp \ldots  \asymp \cdot \stackrel{\ell_k}\longrightarrow \cdot \asymp \epsilon.
\]
for some $\ell_1, \ell_2, \ldots, \ell_k$.
If $\|\alpha\|_{\asymp} = k$, then any transition sequence of the form
\begin{equation}\label{eqn:witness_path}
\alpha \stackrel{\asymp}{\Longrightarrow} \cdot \stackrel{\ell_1}{\longrightarrow}  \cdot \stackrel{\asymp}{\Longrightarrow} \cdot \stackrel{\ell_2}{\longrightarrow} \cdot \stackrel{\asymp}{\Longrightarrow}\ldots \stackrel{\asymp}{\Longrightarrow} \cdot \stackrel{\ell_k}{\longrightarrow} \cdot \stackrel{\asymp}{\Longrightarrow} \epsilon
\end{equation}
is called a {\em witness path} of $\asymp$-norm for $\alpha$. The {\em length} of the witness path is $k$.
\end{definition}
Clearly, the $\asymp$-norms have the following basic fact:
\begin{lemma}
If $\alpha \asymp \beta$, then $\|\alpha\|_{\asymp} = \|\beta\|_{\asymp}$.
\end{lemma}
If $\asymp$ is an arbitrary equivalence relation, the witness path does not always exist, because it is not always the case $\alpha \stackrel{\asymp}{\Longrightarrow} \beta$ whenever $\alpha \asymp \beta$.  This is one of the motivations of the forthcoming notion of {\em decreasing bisimulation} (Definition~\ref{def:decreasing_bisimulation}). For the moment, we introduce the $\asymp$-decreasing transitions.
\begin{definition}
A transition $\alpha \stackrel{\ell}{\longrightarrow} \alpha'$ is {\em $\asymp$-decreasing} if $\|\alpha'\|_{\asymp} < \|\alpha\|_{\asymp}$.
\end{definition}
According to Definition~\ref{def:semantic_norm},  $\|\alpha'\|_{\asymp} = \|\alpha\|_{\asymp} - 1$ if $\alpha \stackrel{\ell}{\longrightarrow} \alpha'$ is a $\asymp$-decreasing transition.  In witness path~(\ref{eqn:witness_path}), every transition $\stackrel{\ell_i}{\longrightarrow}$ must be $\asymp$-decreasing for $1 \leq i \leq k$.

\begin{definition}\label{def:decreasing_bisimulation}
A process equivalence $\asymp$ is a {\em decreasing bisimulation},  if the following conditions are satisfied:
\begin{enumerate}
\item
If $\alpha \asymp \epsilon$, then $\alpha \Longrightarrow \epsilon$.

\item
If $\alpha  \asymp  \beta$ and
$\alpha \stackrel{\ell}{\longrightarrow} \alpha'$ is a $\asymp$-decreasing transition, then there exist $\beta''$ and $\beta'$ such that
$\beta  \stackrel{\asymp}{\Longrightarrow} \beta'' \stackrel{\ell}{\longrightarrow} \beta'$ and $\alpha' \asymp \beta'$.
\end{enumerate}
\end{definition}
Decreasing bisimulation is a weaker version of bisimulation. The difference lies in that only decreasing transitions need to be matched. Be aware that the transition $\beta'' \stackrel{\ell}{\longrightarrow} \beta'$ in Definition~\ref{def:decreasing_bisimulation}  is forced to be $\asymp$-decreasing.

Let $\asymp$ be a decreasing bisimulation.  Then any $\asymp$-decreasing transition of $\alpha$ can be extended to a witness path of $\asymp$-norm of $\alpha$. The norm $\|\alpha\|_{\asymp}$ is equal to the least number of decreasing transitions from $\alpha$ to $\epsilon$.

Nearly all  equivalences appearing in this paper are decreasing bisimulation. For example:
\begin{proposition}
 $\simeq_R$ is a decreasing bisimulation for every $R \subseteq \mathbf{C}_{\mathrm{G}}$.
\end{proposition}
There is no need to define the so-called {\em $R$-decreasing bisimulation}. The following lemma confirms that, for decreasing transitions, $\stackrel{\ell}{\longrightarrow}_{R}$ and $\stackrel{\ell}{\longrightarrow}$ are essentially the same.
\begin{lemma}
If $\alpha \stackrel{\ell}{\longrightarrow}_{R} \beta$ is $\simeq_R$-decreasing, then $\alpha \stackrel{\ell}{\longrightarrow} \beta'$ for some $\beta'$  such that $\beta_R' = \beta$.
\end{lemma}
The following lemma provides a bound for semantic norms.
\begin{lemma}\label{lem:norm_size}
If $\asymp$ is a decreasing bisimulation, then $\|\alpha\|_{\asymp} \leq |\alpha|_{\mathrm{st}}$ for every $\alpha$.
\end{lemma}

\begin{remark}
The labelled transition graph defined by a normed BPA $\Gamma$ can be perceived as a {\em directed graph} $\mathcal{G}$ (with infinite number of nodes) whose nodes are the processes and whose edges are the labelled transitions. $\mathcal{G}$ can be extended to a {\em weighted direct graph} in different ways. Let $\mathcal{G}_{\mathrm{st}}$ be the weighted extension of $\mathcal{G}$ in which every edge of $\mathcal{G}$ has weight \textit{one}. Let $\mathcal{G}_{\mathrm{wk}}$ be the one which is the same as $\mathcal{G}_\mathrm{st}$ except that the weight of every silent transition is \textit{zero}.   The strong (resp.~weak) norm  of a process $\alpha$ is the length of the shortest path from $\alpha$ to $\epsilon$ in $\mathcal{G}_\mathrm{st}$ (resp.~ $\mathcal{G}_\mathrm{wk}$).
Let $\asymp$ be an equivalence relation on processes. We can define the graph $\mathcal{G}_{\asymp}$ which is the same as $\mathcal{G}_{\mathrm{st}}$ except that the weight of $\gamma \stackrel{\ell}{\longrightarrow} \gamma'$ is set \textit{zero} if $\gamma \asymp \gamma'$.   A witness path of $\asymp$-norm of $\alpha$ corresponds to a shortest path from $\alpha$ to $\epsilon$ in $\mathcal{G}_{\asymp}$.

If $\simeq_R$ is known for every $R$, then $\simeq_R$-norm of a BPA process (or constant) can be calculated via  the greedy strategy in an efficient way.  It depends on the following property:
\begin{enumerate}
\item
$\|\alpha\|_{\simeq_R} = 0$ if and only if $\alpha \simeq_R \epsilon$.

\item
$\|\alpha \beta\|_{\simeq_R}  =   \|\alpha\|_{\simeq_{R'}}   + \|\beta\|_{\simeq_R} $ in which $R' = {\mathsf{Rd}^{\simeq}_{R}(\beta)}$.
\end{enumerate}
It works like a generalization of Breadth-first search, or a variant of Dijkstra's algorithm.   The detail of the calculation is omitted, but this idea will be used to calculate the semantic norms $\|\cdot\|_{\stackrel{\mathcal{B}}{=}_{R}}$ later (Section~\ref{sec:naive-algorithm}) in the refinement procedure when constructing new base from the old.
\end{remark}

\subsection{Decreasing Bisimulation with $R$-Expansion of $\bumpeq$}\label{subsec:decreasing_expansion}

Based on decreasing transitions,  we can define a special notion  called {\em decreasing bisimulation with $R$-expansion of $\bumpeq$}, which will be taken as the refinement operation in our algorithm. This notion is crucial to the correctness of the refinement operation.
 The readers are suggested  to review Definition~\ref{def:R_beq} and Definition~\ref{def:decreasing_bisimulation} before going on.

\begin{definition}\label{def:_decreasing_bisimulation_with_expansion}
Let $\asymp$ and $\bumpeq$ be two equivalences on processes such that
${=_{R}} \subseteq {\asymp}  \subseteq {\bumpeq}$.
We say that $\asymp$ is an {\em $R$-expansion} of $\bumpeq$  if the following conditions hold whenever $\alpha \asymp \beta$:
\begin{enumerate}
\item
$\alpha \Longrightarrow \epsilon$ if and only if $\beta \Longrightarrow \epsilon$.

\item
If $\alpha \stackrel{\not\bumpeq}{\longrightarrow} \alpha'$, then either
$\beta_R  \stackrel{\asymp}{\Longrightarrow}_{R} \cdot \stackrel{\not\bumpeq}{\longrightarrow}_{R} \beta'$ for some $\beta'$ such that $\alpha' \bumpeq \beta'$.

\item
If $\alpha \stackrel{a}{\longrightarrow} \alpha'$, then
$\beta_R \stackrel{\asymp}{\Longrightarrow}_{R} \cdot \stackrel{a}{\longrightarrow}_{R} \beta'$ for some $\beta'$ such that $\alpha' \bumpeq \beta'$.

\end{enumerate}
We say that  $\asymp$ is a {\em decreasing bisimulation with $R$-expansion of $\bumpeq$} if
$\asymp$ is both a decreasing bisimulation and an $R$-expansion of $\bumpeq$.
\end{definition}
The following lemma provides another characterization of the decreasing bisimulation with $R$-expansion of $\bumpeq$.
\begin{lemma}\label{lem:expansion}
Assume that ${=_{R}} \subseteq {\asymp}  \subseteq {\bumpeq}$.
$\asymp$ is an decreasing bisimulation with $R$-expansion of $\bumpeq$  if and only if following conditions hold whenever $\alpha \asymp \beta$ and $\alpha, \beta$ are in $R$-nf:
\begin{enumerate}
\item
if  $\alpha\Longrightarrow_R \epsilon$,  then $\beta \Longrightarrow_R \epsilon$;

\item
if $ \alpha \stackrel{\not\asymp}{\longrightarrow}_R \alpha'$, being $\asymp$-decreasing, then
$\beta  \stackrel{\asymp}{\Longrightarrow}_{R} \cdot \stackrel{\not\asymp}{\longrightarrow}_{R} \beta'$ for some $\beta'$ such that $\alpha' \asymp \beta'$;

\item
if $ \alpha \stackrel{\not\bumpeq}{\longrightarrow}_R \alpha'$, not  being $\asymp$-decreasing, then
$\beta  \stackrel{\asymp}{\Longrightarrow}_{R} \cdot \stackrel{\not\bumpeq}{\longrightarrow}_{R} \beta'$ for some $\beta'$ such that $\alpha' \bumpeq \beta'$;

\item
if $\alpha \stackrel{a}{\longrightarrow}_R \alpha'$,   being $\asymp$-decreasing,  then
$\beta \stackrel{\asymp}{\Longrightarrow}_{R} \cdot \stackrel{a}{\longrightarrow}_{R} \beta'$ for some $\beta'$ such that $\alpha' \asymp \beta'$.

\item
if $\alpha \stackrel{a}{\longrightarrow}_R \alpha'$,   not  being $\asymp$-decreasing,  then
$\beta \stackrel{\asymp}{\Longrightarrow}_{R} \cdot \stackrel{a}{\longrightarrow}_{R} \beta'$ for some $\beta'$ such that $\alpha' \bumpeq \beta'$.
\end{enumerate}
\end{lemma}

\begin{remark}
The style of the definition of `decreasing bisimulation with expansion' also appears in~\cite{DBLP:journals/corr/He14a}. The main difference is that in this paper, the `decreasing transitions' are semantic, while in ~\cite{DBLP:journals/corr/He14a}, the `decreasing transitions' are syntactic.

The notion of `decreasing bisimulation with expansion' is a better understanding of the previous refinement operations on totally normed BPA and BPP~\cite{DBLP:conf/fsttcs/CzerwinskiL10,DBLP:journals/mscs/HirshfeldJM96}. Moreover, this notion is crucial to the development of a polynomial time algorithm for branching bisimilarity on totally normed BPA~\cite{DBLP:journals/corr/He14a}.
\end{remark}

\section{Decomposition Bases}\label{sec:base}

In this section, we define a way for finitely representing a family of equivalences which satisfies unique decomposition property in the sense of Theorem~\ref{thm:QUDP_RBisimularity}. Such family of equivalences include  $\{\simeq_R\}_R$ and all the intermediate families of equivalences constructed during the iterations. This finite representation is named decomposition base.

\subsection{$R$-blocks and $R$-orders}\label{subsec:blocks_orders}

To make our algorithm easy to formulate, we need some technical preparations.
The reason will be clear later.

\begin{definition}\label{def:R-blocks}
Let $R \subseteq \mathbf{C}_{\mathrm{G}}$.   We call that $\alpha$ is {\em $R$-associate} to $\beta$ if
$\alpha \Longleftrightarrow_{R}   \beta$.
Let $X\in \mathbf{C} \setminus R$. The {\em $R$-block} related to $X$, denoted by $[X]_R$ is the set of all the constants which is $R$-associate to $X$. Namely, $[X]_R \stackrel{\mathrm{def}}{=} \{Y \;|\;  X \Longleftrightarrow_{R}   Y \}$. We use the term {\em block} to specify any $R$-block for $R \subseteq \mathbf{C}_{\mathrm{G}}$.
\end{definition}
Clearly, two $R$-blocks coincide when they overlap. Thus $R$-blocks compose a partition of $\mathbf{C} \setminus R$. The partition is denoted by $\mathbf{C}_R \stackrel{\mathrm{def}}{=} \{[X]_R \;|\; X \in \mathbf{C} \setminus R\}$.

We will use the convention that the members of $[X]_R$ for different $R$'s are taken from different copy of $\mathbf{C}$. In other words, if $R_1 \neq R_2$, then $[X]_{R_1}$ and $[X]_{R_2}$ are always disjoint, and they are regarded as different objects, even if they indicate the same set.

\begin{remark}
The intuition of $R$-blocks is obvious. According to the  Computation Lemma (Lemma~\ref{lem:computation_lemma_R}), The $R$-associate constants are $R$-bisimilar to each other, thus they can be {\em contracted} into a single one.  In the work~\cite{DBLP:journals/corr/He14a} for totally normed BPA,  we prevent the occurrence of $X \Longleftrightarrow   Y$ via a preprocess in which mutually associate constants are contracted into a singe one. For normed BPA discussed in this paper, we take the same idea but the difficulty is that the contracting operation cannot be performed uniformly, for it depends on the reference set $R$. The only way we can take is to introduce the $R$-association and to contract $R$-associate constants into $R$-blocks for individual $R$'s. The members in an $R$-block are interchangeable.
\end{remark}

Be aware that it is possible that $X \simeq_R \epsilon$ even if $X \not\in R$.  In this case we must have $[X]_R \subseteq \mathsf{Id}_R$ by the Computation Lemma (Lemma~\ref{lem:computation_lemma_R}). Also note that it is possible that
$X  \Longleftrightarrow_{R}   \epsilon$ for some $R$. But these kinds of $R$ is uninteresting because by putting such $X$ into $R$ we can get a larger reference set $R'$ such that ${\simeq_{R'}} = {\simeq_{R}}$.

We call a reference set $R$ {\em qualified} if $X \Longleftrightarrow_{R}   \epsilon$ cannot happen for every $X \not\in R$. The unqualified $R$'s can be pre-determined. They are useless  from now to the end of this paper. From now on we assume that every reference set $R$ is qualified. For example when we write `for every $R \subseteq \mathbf{C}_{\mathrm{G}}$', we refer to every {\em qualified} $R$ which is a subset of $\mathbf{C}_{\mathrm{G}}$. In particular, every admissible set is qualified.
\begin{lemma}
All constants in a block $[X]_{R}$ are $R$-bisimilar.
\end{lemma}
\begin{lemma}\label{lem:R_order}
If $[X]_R \neq [Y]_R$ and $X \Longrightarrow_{R} Y$, then $Y \not\Longrightarrow_{R} X$.
\end{lemma}

The behaviours of $[X]_R$ can be more than the total behaviours of its member constants.  All the processes associate to $X$ should be taken into account.   It is possible that $X \Longleftrightarrow_R \zeta X'$ for some ground process $\zeta$.  For instance we can have
$X \Longrightarrow_R Z \Longrightarrow_R   \zeta Y \Longrightarrow_R Y \Longrightarrow_R X$. In this example, $X,Y,Z,\zeta X,\zeta Y, \zeta Z$ are mutually $R$-associate.
Thus the behaviour of $\zeta$ should also be taken into account.
\begin{definition}\label{def:R_propagating}
$Y$ is an {\em $R$-propagating} of $X$ (or of $[X]_R$) if $X \Longleftrightarrow_R Y\zeta X'$ for some $\zeta$ and $X'$. (In this case we must have $X' \Longleftrightarrow_R X$, and $Y \zeta$ is ground.)
\end{definition}
\begin{lemma}
$Y \in \mathsf{Rd}_R(X)$ if $Y$ is an $R$-propagating of $X$.
\end{lemma}
\begin{lemma}\label{lem:char_propagating}
Suppose  $X \Longleftrightarrow_R \zeta X' \stackrel{\ell}{\longrightarrow}_R \zeta' X'$ such that  $\zeta' X' \not\Longrightarrow_R X$. Then $X' \in [X]_R$, and $\zeta = Y\gamma$ for some $Y$ and $\gamma$ such that
\begin{itemize}
\item
$Y$ is an $R$-propagating of $[X]_R$.

\item
$Y  \stackrel{\ell}{\longrightarrow} \alpha$ and $\zeta' = \alpha\gamma$.

\item
$X \simeq_R \gamma X$. (i.e.~$\gamma \in (\mathsf{Rd}_R(X))^{*}$)

\item
$Y.X \stackrel{\ell}{\longrightarrow}_R \alpha X$ with $Y.X \simeq_R \zeta X'  \simeq_R X$ and $\alpha X  \simeq_R \zeta' X \simeq_R \zeta' X'$.
\end{itemize}
\end{lemma}
Lemma~\ref{lem:char_propagating} shows that the behaviours of $[X]_R$ are completely determined by the associate constants and the propagating constants of $X$, which leads to the following definition.
\begin{definition}\label{def:R_derived_transition}
The {\em $R$-derived transition} $\stackrel{\ell}{\longmapsto}_{R}$ is defined as follows:
\begin{enumerate}
\item
Let $\widehat{X} \in [X]_R$ and $\widehat{X} \stackrel{\ell}{\longrightarrow}_{R} \alpha$.  If
either $\ell \neq \tau$, or $\ell = \tau$ and $\alpha \not \Longrightarrow_R X$,
 then $[X]_R \stackrel{\ell}{\longmapsto}_{R} \alpha$.

\item
Let $Y$ be an $R$-propagating of $[X]_R$ and $Y \stackrel{\ell}{\longrightarrow}_{R} \alpha$. If
 either $\ell \neq \tau$,
 or $\ell = \tau$ and $\alpha \not \Longrightarrow_R \epsilon$,
then $[X]_R \stackrel{\ell}{\longmapsto}_{R} \alpha.X$.
\end{enumerate}
\end{definition}
\begin{lemma}
Suppose  $X  \Longleftrightarrow_R \cdot \stackrel{\ell}{\longrightarrow}_{R} \alpha$. If
 either $\ell \neq \tau$,
 or $\ell = \tau$ and $\alpha \not \Longrightarrow_R X$,
then
$[X]_R \stackrel{\ell}{\longmapsto}_{R} \cdot \simeq_R \alpha$.
\end{lemma}
It is technically convenient to treat the $R$-blocks as the basic objects in the algorithm, because of the following lemma.
\begin{lemma}
If $[X]_R \stackrel{\tau}{\longmapsto}_{R} \cdot \Longrightarrow_R Y$, then $[Y]_R \neq [X]_R$.
\end{lemma}

Finally
we can define an order on $R$-blocks based on Lemma~\ref{lem:R_order}. For every $R$,
we fix a linear order $<_R$ such that whenever $[X]_R <_R [Y]_R$, we have $X \not \Longrightarrow_R Y$.

\begin{lemma}
If $[X]_R \stackrel{\tau}{\longmapsto}_{R} \cdot \Longrightarrow_R Y$, then $[Y]_R <_R [X]_R$.
\end{lemma}

\begin{example}
The example illustrates why we have to introduce possibly different orders $<_R$ for different $R$'s.

In a normed BPA system $\Gamma = (\mathbf{C}, \mathcal{A}, \Delta)$, we can have the following fragment of definition: Let $A_1, A_2, B_1, B_2$ be constants in $\mathbf{C}_{\mathrm{G}}$. We have transition rules:
\begin{center}
$A_1 \stackrel{\tau}{\longrightarrow} A_2 B_1$, \quad $A_2 \stackrel{\tau}{\longrightarrow} A_1 B_2$.
\end{center}
There can be other transitions related to there constants which is of no importance. Now, take notice of the following facts:
\begin{itemize}
\item
In the case of $R_1 = \{B_1\}$, we have $A_1 \stackrel{\tau}{\longrightarrow}_{R_1} A_2$. Thus we have $[A_2]_{R_1} <_{R_1} [A_1]_{R_1}$.  Or, in short, $A_2 <_{R_1} A_1$.

\item
In the case of $R_2 = \{B_2\}$, we have $A_2 \stackrel{\tau}{\longrightarrow}_{R_2} A_1$. Thus we have $[A_1]_{R_2} <_{R_2} [A_2]_{R_2}$.  Or, in short, $A_1 <_{R_2} A_2$.
\end{itemize}
These two orders $<_{R_1}$ and $<_{R_2}$ are clearly not consistent.  This feature reflects a big difference between normed BPA and  totally normed BPA.
\end{example}

\begin{remark}
There is also a big difference between normed BPA and normed BPP.
In the case of normed BPP~\cite{Stirling2001,DBLP:conf/concur/CzerwinskiHL11}, let us say that $X$ {\em generates} $Y$ if $X \Longleftrightarrow Y \parallel X$, in while `$\parallel$' is the operator of {\em parallel composition}.  Thus, if $X$ generates $Y$, then $X \Longleftrightarrow  Y^{n} \parallel X$ hence $X \simeq  Y^{n} \parallel X$ for every $n \in \mathbb{N}$. Suppose that $X \stackrel{a}{\longrightarrow} \epsilon$, we have
\[
X \Longleftrightarrow  \underbrace{Y \parallel \ldots \parallel Y }_{n\textrm{ times}}  \parallel X \stackrel{a}{\longrightarrow} \underbrace{Y \parallel \ldots \parallel Y }_{n\textrm{ times}} 
\]
for every $n \in \mathbb{N}$.
Now, if all the 
\[
\underbrace{Y \parallel \ldots \parallel Y }_{n\textrm{ times}} \parallel X 
\]
for every $n \in \mathbb{N}$ are contracted into a block $[X]$, then we have 
\[
[X] \stackrel{\ell}{\longmapsto} \underbrace{Y \parallel \ldots \parallel Y }_{n\textrm{ times}}
\]
for every $n \in \mathbb{N}$, as is done in the same way as Definition~\ref{def:R_derived_transition}.  This example shows that the behaviour of $[X]$ is infinite branching.
Note that in the case of normed BPA, this situation is not possible. If $X \Longleftrightarrow  Y \zeta X$, then actions in $X$ can only be activated after $Y\zeta$ is consumed completely.  This nice property of normed BPA simplifies the situation greatly.
\end{remark}


\subsection{Decomposition Bases}\label{subsec:decomposition_base}

 A {\em decomposition base} $\mathcal{B}$ is a family of $\{\mathcal{B}_{R}\}_{R \subseteq{\mathbf{C}_{\mathrm{G}}}}$ in which every $ \mathcal{B}_R$ is a quintuple $(\mathbf{Id}^{\mathcal{B}}_R, \mathbf{Pr}^{\mathcal{B}}_R, \mathbf{Cm}^{\mathcal{B}}_R ,\mathbf{Dc}^{\mathcal{B}}_R, \mathbf{Rd}^{\mathcal{B}}_R)$.
\begin{itemize}
\item
$\mathbf{Id}^{\mathcal{B}}_R$ is a subset of ground constants called {\em $\mathcal{B}_{R}$-identities}.

\item
$\mathbf{Cm}^{\mathcal{B}}_R$ specifies the set of {\em $\mathcal{B}_{R}$-composites}. A $\mathcal{B}_{R}$-composite is an $\mathbf{Id}^{\mathcal{B}}_R$-block.

\item
$\mathbf{Pr}^{\mathcal{B}}_R$ specifies the set of {\em $\mathcal{B}_{R}$-primes}. A $\mathcal{B}_{R}$-prime is an $\mathbf{Id}^{\mathcal{B}}_R$-block.

\item
$\mathbf{Rd}^{\mathcal{B}}_R$ is a function whose domain is $\mathbf{Pr}^{\mathcal{B}}_R$.  Let $[X]_{\mathbf{Id}^{\mathcal{B}}_R}$ be a  $\mathcal{B}_{R}$-prime.   The value $\mathbf{Rd}^{\mathcal{B}}_R([X]_{\mathbf{Id}^{\mathcal{B}}_R})$ is a set of ground constants which are called {\em $\mathcal{B}_{R}$-redundant} over $[X]_{\mathbf{Id}^{\mathcal{B}}_R}$.

\item
$\mathbf{Dc}^{\mathcal{B}}_R$ is a function whose domain is $\mathbf{Cm}^{\mathcal{B}}_R$. Let $[X]_{\mathbf{Id}^{\mathcal{B}}_R}$ be a  $\mathcal{B}_{R}$-composite.  The value $\mathbf{Dc}^{\mathcal{B}}_R([X]_{\mathbf{Id}^{\mathcal{B}}_R})$ is called the {\em $\mathcal{B}_{R}$-decomposition} of $[X]_{\mathbf{Id}^{\mathcal{B}}_R}$, which is a string of blocks $[X_r]_{R_{r}} [X_{r-1}]_{R_{r-1}}\ldots [X_2]_{R_{2}} [X_1]_{R_{1}}$ with $r \geq 1$, $R_1 = \mathbf{Id}^{\mathcal{B}}_R$, $[X_i]_{R_{i}} \in \mathbf{Pr}^{\mathcal{B}}_{R_{i}}$ and $R_{i+1} = \mathbf{Rd}^{\mathcal{B}}_{R_{i}}([X_i]_{R_{i}})$ for every $1 \leq i <r$.
\end{itemize}
To make a decomposition base $\mathcal{B}$ work properly, we need the following constraints:
\begin{enumerate}
\item
$R \subseteq \mathbf{Id}^{\mathcal{B}}_R \subseteq \mathbf{C}_{\mathrm{G}}$.

\item
If $R \subseteq S$, then $\mathbf{Id}^{\mathcal{B}}_R  \subseteq \mathbf{Id}^{\mathcal{B}}_S$.   If  $R \subseteq S \subseteq \mathbf{Id}^{\mathcal{B}}_R$, then $\mathbf{Id}^{\mathcal{B}}_R  = \mathbf{Id}^{\mathcal{B}}_S$.  In particular,  $\mathbf{Id}^{\mathcal{B}}_{\mathbf{Id}^{\mathcal{B}}_R} = \mathbf{Id}^{\mathcal{B}}_R$.

\item
$\mathcal{B}_R = \mathcal{B}_{\mathbf{Id}^{\mathcal{B}}_R}$ for every $R$. When $R = \mathbf{Id}^{\mathcal{B}}_R$, $R$ is called {\em $\mathcal{B}$-admissible}.   $\mathcal{B}$ is completely determined by those $\mathcal{B}_R$ in which $R$ is $\mathcal{B}$-admissible.

\item
If $R$ is  $\mathcal{B}$-admissible, then $\mathbf{Cm}^{\mathcal{B}}_R$ and $\mathbf{Pr}^{\mathcal{B}}_R$ are a partition of $R$-blocks:
$\mathbf{Cm}^{\mathcal{B}}_R \cup \mathbf{Pr}^{\mathcal{B}}_R = \mathbf{C}_{R}$ and  $\mathbf{Cm}^{\mathcal{B}}_R \cap \mathbf{Pr}^{\mathcal{B}}_R = \emptyset$.

\item
$\mathbf{Rd}^{\mathcal{B}}_R([X]_{\mathbf{Id}^{\mathcal{B}}_R})$ is  $\mathcal{B}$-admissible provided that $[X]_{\mathbf{Id}^{\mathcal{B}}_R}$ is a  $\mathcal{B}_{R}$-prime.  Thus $\mathbf{Dc}^{\mathcal{B}}_R$ is well-defined.
\end{enumerate}

A decomposition base $\mathcal{B}$ defines a family of string rewriting system $\{\stackrel{\mathcal{B}}{\rightarrow}_R\}_{R \subseteq \mathbf{C}_{\mathrm{G}}}$. The family of {\em $\mathcal{B}_R$-reduction} relations are defined according to the following structural rules.
\begin {displaymath}
    \begin{array}{cccccc}
        \cfrac{X \in \mathbf{Id}^{\mathcal{B}}_R}{ X  \stackrel{\mathcal{B}}{\rightarrow}_R \epsilon} \qquad
        \cfrac{X \not\in \mathbf{Id}^{\mathcal{B}}_R}{  X \stackrel{\mathcal{B}}{\rightarrow}_R [X]_{\mathbf{Id}^{\mathcal{B}}_R}} \qquad
        \cfrac{\mathbf{Dc}^{\mathcal{B}}_R([X]_{\mathbf{Id}^{\mathcal{B}}_R}) = \alpha}{ [X]_{\mathbf{Id}^{\mathcal{B}}_R} \stackrel{\mathcal{B}}{\rightarrow}_R \alpha}  \\
        \cfrac{[X]_{\mathbf{Id}^{\mathcal{B}}_R} \in \mathbf{Pr}^{\mathcal{B}}_R \qquad  \alpha \stackrel{\mathcal{B}}{\rightarrow}_{\mathbf{Rd}^{\mathcal{B}}_R([X]_{\mathbf{Id}^{\mathcal{B}}_R})} \alpha'}{\alpha. [X]_{\mathbf{Id}^{\mathcal{B}}_R} \stackrel{\mathcal{B}}{\rightarrow}_R \alpha'.[X]_{\mathbf{Id}^{\mathcal{B}}_R}}
\qquad
        \cfrac{\beta \stackrel{\mathcal{B}}{\rightarrow}_R \beta'}{\alpha.\beta \stackrel{\mathcal{B}}{\rightarrow}_R \alpha.\beta'}
     \end{array}
\end {displaymath}
$\mathcal{B}_R$-reduction relations are deterministic. Thus for any process $\alpha$, the {\em $\mathcal{B}_R$-normal-form} (in the sense of string rewriting systems) is unique, and it is called the {\em $\mathcal{B}_R$-decomposition} of $\alpha$.  We use the notation $\mathtt{dcmp}_R^{\mathcal{B}}(\alpha)$ to indicate the $\mathcal{B}_R$-decomposition of $\alpha$.
Processes $\alpha$ and $\beta$ are {\em $\mathcal{B}_R$-equivalent}, notation $\alpha \stackrel{\mathcal{B}}{=}_{R} \beta$, if they have the same $\mathcal{B}_R$-decomposition.

\begin{lemma}
$\alpha \stackrel{\mathcal{B}}{=}_{R} \beta$ if and only if   $\mathtt{dcmp}_R^{\mathcal{B}}(\alpha) = \mathtt{dcmp}_R^{\mathcal{B}}(\beta)$.
\end{lemma}
According to $\mathcal{B}_R$-reduction rules, we have
the following characterization of  $\mathtt{dcmp}_R^{\mathcal{B}}(\alpha)$.
\begin{lemma}
If $R$ is $\mathcal{B}$-admissible, then
\begin{itemize}
\item
$\mathtt{dcmp}_R^{\mathcal{B}}(\epsilon) = \epsilon$.

\item
If $X \in R$, then $\mathtt{dcmp}_R^{\mathcal{B}}(\gamma X) = \mathtt{dcmp}_R^{\mathcal{B}}(\gamma)$.

\item
If $X \not\in R$, then $\mathtt{dcmp}_R^{\mathcal{B}}(\gamma X)
= \mathtt{dcmp}_R^{\mathcal{B}}(\gamma. [X]_R )$.

\item
If $[X]_R \in \mathbf{Cm}_R$, then
\begin{center}
$\mathtt{dcmp}_R^{\mathcal{B}}(\gamma. [X]_R) = \mathtt{dcmp}_{R}^{\mathcal{B}}(\gamma. \mathbf{Dc}_R([X]_R))$.
\end{center}

\item
If $[X]_R  \in \mathbf{Pr}_R$, then
\begin{center}
$\mathtt{dcmp}_R^{\mathcal{B}}(\gamma. [X]_{R} ) = (\mathtt{dcmp}_{\mathbf{Rd}^{\mathcal{B}}_R([X]_R)}^{\mathcal{B}}(\gamma)). [X]_R $.
\end{center}
\end{itemize}
If $R$ is not $\mathcal{B}$-admissible, then $\mathtt{dcmp}_R^{\mathcal{B}}(\alpha) = \mathtt{dcmp}_{\mathbf{Id}^{\mathcal{B}}_R}^{\mathcal{B}}(\alpha)$.
\end{lemma}
We list some basic facts.
\begin{lemma}
$\alpha \stackrel{\mathcal{B}}{=}_{R} \epsilon$ if and only if $\alpha \in \mathbf{Id}^{\mathcal{B}}_R$.  When $R$ is $\mathcal{B}$-admissible, $\alpha \stackrel{\mathcal{B}}{=}_{R} \epsilon$ if and only if $\alpha \in R$.
\end{lemma}
\begin{lemma}
 If $X_1, X_2 \in [X]_R$, then $X_1 \stackrel{\mathcal{B}}{=}_{R} X_2$ and $\| X_1 \|_{\stackrel{\mathcal{B}}{=}_{R}}
  = \| X_2\|_{\stackrel{\mathcal{B}}{=}_{R}}$.
\end{lemma}
We can write $\|[X]_R\|_{\stackrel{\mathcal{B}}{=}_{R}} = \|\widehat{X}\|_{\stackrel{\mathcal{B}}{=}_{R}}$ for any $\widehat{X} \in [X]_R$.
\begin{lemma}
$\|X\|_{\stackrel{\mathcal{B}}{=}_{R}} \geq 1$ if $R$ is $\mathcal{B}$-admissible and $X \not\in R$.
\end{lemma}
\begin{lemma}\label{lem:bound_dc}
If $\stackrel{\mathcal{B}}{=}_{R}$ is a decreasing bisimulation, then the size of $\mathbf{Dc}^{\mathcal{B}}_R([X]_R)$ is exponentially bounded.
\end{lemma}
In the following the superscript $\mathcal{B}$ will often be omitted if $\mathcal{B}$ is clear  from the context. For example sometimes we write $\mathbf{Pr}_R$  for $\mathbf{Pr}^{\mathcal{B}}_R$.

\subsection{Representing $\simeq_R$ via Decomposition Base}\label{subsec:representing_simeq}

We define a decomposition base $\widehat{\mathcal{B}}$ which can represent $\simeq_R$ for every $R$.  That is,  $\alpha \stackrel{\widehat{\mathcal{B}}}{=}_R \beta$ if and only if $\alpha \simeq_R \beta$.
Theorem~\ref{thm:QUDP_RBisimularity} is crucial. Moreover, there are other subtleties which deserve to be  mentioned.
\begin{lemma}
All constants in a block $[X]_{\mathsf{Id}_{R}}$ are $R$-bisimilar.
\end{lemma}
The description of $\widehat{\mathcal{B}} = \{(\mathbf{Id}_R, \mathbf{Pr}_R, \mathbf{Cm}_R ,\mathbf{Dc}_R, \mathbf{Rd}_R)\}_{R}$ relies on the family of orders $\{ <_R \}_R$ defined in Section~\ref{subsec:blocks_orders}. It contains three steps:
\begin{itemize}
\item
In the first step, we determine $\mathbf{Id}_R$ for every $R$:
$\mathbf{Id}_R = \mathsf{Id}_{R}$. According to Proposition~\ref{prop:R_monotone}, Proposition~\ref{prop:R_bis_vs_IDR} and their corollaries, $\mathbf{Id}_R$ satisfies constraints~1--3 in Section~\ref{subsec:decomposition_base}. In particular, $R$ is admissible if and only if $R$ is $\widehat{\mathcal{B}}$-admissible.

\item
In the second step, we determine other constituents of $\widehat{\mathcal{B}}_R$ for every $\widehat{\mathcal{B}}$-admissible $R$:
\begin{itemize}
\item
$\mathbf{Pr}_R = \{[X]_{R} \;|\;  X \in \mathsf{Pr}_R \mbox{ and } X \not\simeq_R Y \textrm{ for} \textrm{ every } Y <_{R} X \}$.

\item
$\mathbf{Cm}_R = \{[X]_{R} \;|\; X \in \mathsf{Cm}_R \mbox{, or } X \in \mathsf{Pr}_R \textrm{ and } X \simeq_R  Y   <_{R} X  \mbox{ for some } Y\}$.

\item
If $[X]_{R} \in  \mathbf{Pr}_R$, then
$\mathbf{Rd}_R([X]_{R}) = \{Y \;|\;  Y X \simeq_R X \}$.  Be aware that $\mathbf{Rd}_R([X]_{R})$ is admissible (also $\widehat{\mathcal{B}}$-admissible) according to Lemma~\ref{lem:Rd_admissible}.

\item
If $[X]_{R} \in  \mathbf{Cm}_R$, then $\mathbf{Dc}_R([X]_{R}) = [X_r]_{R_{r}} [X_{r-1}]_{R_{r-1}} \ldots [X_2]_{R_{2}} [X_1]_{R_{1}}$, in which $X \simeq_R X_r.X_{r-1}.\ldots.X_1$, $R_1 = R$, $[X_i]_{R_{i}} \in \mathbf{Pr}_{R_{i}}$ and $R_{i+1} = \mathbf{Rd}_{R_{i}}([X_i]_{R_{i}})$ for every $1 \leq i < r$. Thanks to  the $\widehat{\mathcal{B}}$-admissibility of $\mathbf{Rd}_{R_{i}}([X_i]_{R_{i}})$ for $1 \leq i \leq r$, $\mathbf{Dc}_R([X]_{R})$ is well-defined.
\end{itemize}

\item
In the third step, for every non-$\widehat{\mathcal{B}}$-admissible $R$,  $\widehat{\mathcal{B}}_{\mathbf{Id}_R}$ is assigned to $\widehat{\mathcal{B}}_{R}$.   That is,  $\mathbf{Pr}_R \coloneqq  \mathbf{Pr}_{\mathbf{Id}_R}$, $\mathbf{Cm}_R \coloneqq  \mathbf{Cm}_{\mathbf{Id}_R}$, and so on.
\end{itemize}
Pay special attention to the descriptions of $\mathbf{Pr}_R$ and $\mathbf{Cm}_R$. They have slightly different from $\mathsf{Pr}_R$ and
$\mathsf{Cm}_R$. Semantically, if $X \in \mathsf{Pr}_R$ and $X \simeq_R Y$, then $Y \in \mathsf{Pr}_R$. In the syntactic description of $\mathbf{Pr}_R$ and $\mathbf{Cm}_R$, we need the $\widehat{\mathcal{B}}_{R}$-primes to be absolutely unique, which is accomplished via $<_{R}$.  The orders $<_{R}$ take effects in double means:
Let $R$ be admissible, then
\begin{enumerate}
\item  Among the $R$-blocks of $\simeq_R$-primes, there is exactly one distinguished $R$-block that is qualified as a $\widehat{\mathcal{B}}_R$-prime, which is the $<_{R}$-minimum one in the related $\simeq_R$-class.

\item
 Let $[X]_{R}$ be a $\widehat{\mathcal{B}}_R$-prime. If $[X]_{R} \stackrel{\tau}{\longmapsto}_R \alpha$, then $X \not \stackrel{\widehat{\mathcal{B}}}{=}_{R} \alpha$. If $X \Longrightarrow_{R} Y \not\Longrightarrow_{R} X$, then $X \not \stackrel{\widehat{\mathcal{B}}}{=}_{R} Y$.
\end{enumerate}
Every decomposition base constructed during the refinement procedure in our algorithm will satisfy these two properties.

\begin{remark}
For realtime normed BPA (or BPP), there is an even strong property. There exists a uniform order `$<$' on all the constants such that, whenever $X \stackrel{\ell}{\longrightarrow} \alpha$ is syntactically (also semantically) decreasing or $\asymp$-preserving (for some appropriate $\asymp$), all the constants in $\alpha$ will be strictly less than $X$ in order `$<$'.
In history, this property plays a significant role in the previous fast bisimilarity decision algorithms~\cite{DBLP:journals/mscs/HirshfeldJM96,DBLP:conf/fsttcs/CzerwinskiL10}.

For totally normed BPA, we have an adaptation of this strong property, in which the condition becomes $X \stackrel{\ell}{\longrightarrow} \alpha$ is syntactically (no longer semantically) decreasing, or weak-norm-preserving (no longer $\asymp$-preserving)~\cite{DBLP:journals/corr/He14a}.

For non-realtime normed BPA systems, the above requirement is definitely too strong to be satisfied, so that the decision algorithm must be developed in some other ways. This is the origination of putting semantic norms into the algorithm.
\end{remark}

Ultimately  we have the following coincidence result.
\begin{proposition}
$\alpha \simeq_R \beta$ if and only if $\alpha \stackrel{\widehat{\mathcal{B}}}{=}_{R} \beta$.
\end{proposition}

\section{Description of the Algorithm}\label{sec:naive-algorithm}

Our algorithm takes the partition refinement approach. The purpose  is to figure out the $\widehat{\mathcal{B}}$ defined in Section~\ref{subsec:representing_simeq}.  The strategy is to
start with a special initial base $\mathcal{B}_0$ satisfying $\widehat{\mathcal{B}} \subseteq \mathcal{B}_0$ and iteratively refine it.  We will use notation $\mathcal{B} \subseteq \mathcal{D}$ to  mean that ${\stackrel{\mathcal{B}}{=}_{R}}  \subseteq {\stackrel{\mathcal{D}}{=}_{R}}$ for every $R$.
The refinement operation will be denoted by $\mathsf{Ref}$. By taking $\mathcal{B}_{i+1} = \mathsf{Ref}(\mathcal{B}_i)$, we have a sequence of decomposition bases 
\[
\mathcal{B}_0,  \mathcal{B}_1 , \mathcal{B}_2 ,  \ldots
\]
such that 
\[
{\mathcal{B}_0} \supseteq  {\mathcal{B}_1} \supseteq {\mathcal{B}_2} \supseteq  \ldots.
\]
The correctness of the refinement operation adopted in this paper depends on the following requirements, which will be proved gradually:
\begin{enumerate}
\item
$\widehat{\mathcal{B}} \subseteq {\mathcal{B}_0}$.

\item
$\mathsf{Ref}(\widehat{\mathcal{B}} ) = {\widehat{\mathcal{B}} }$.

\item
If ${\widehat{\mathcal{B}} } \subsetneq {\mathcal{B}}$, then $ {\widehat{\mathcal{B}} } \subseteq \mathsf{Ref}(\mathcal{B}) \subsetneq {\mathcal{B}}$.
\end{enumerate}
According to the above three requirements, once the sequence $\{\mathcal{B}_i\}_{i\in \omega}$ becomes stable, say ${\mathcal{B}_i} =  {\mathcal{B}_{i+1}}$ for some $i$, we can affirm that ${\widehat{\mathcal{B}}} = {\mathcal{B}_i}$.

On the whole, our algorithm is an iteration:
\begin{center}
\small
 \begin{tabular}{|p{8.2cm}|}\hline \vspace{-1ex}
\begin{enumerate}
\item
Compute the initial base $\mathcal{B}_0$ and let $\mathcal{D} \coloneqq  \mathcal{B}_0$.

\item
Compute the new base $\mathcal{B}$ from the old base $\mathcal{D}$.

\item
If  $\mathcal{B}$ equals $\mathcal{D}$ then halt and return $\mathcal{B}$.

\item
$\mathcal{D} \coloneqq  \mathcal{B}$ and go to step 2. \vspace{-2ex}
\end{enumerate}
       \\ \hline
\end{tabular}
\end{center}

Apparently, the algorithm relies on the initial base  and the refinement step which computes $\mathcal{B} = \mathsf{Ref}(\mathcal{D})$ from $\mathcal{D}$.

\subsection{Relationships between Old and New Bases}\label{subsec:naive-algorithm}
Before describing the algorithm in details, we investigate the relationship between two bases $\mathcal{B}$ and $\mathcal{D}$ assume that $\mathcal{B} \subseteq \mathcal{D}$.

\begin{lemma}\label{lem:compare_ID}
If $\mathcal{B} \subseteq  \mathcal{D}$, then
$\mathbf{Id}^{\mathcal{B}}_R  \subseteq \mathbf{Id}^{\mathcal{D}}_R$ for every $R$.
\end{lemma}

\begin{remark}
An interesting consequence according to Lemma~\ref{lem:compare_ID} is that,  if $R = \mathbf{Id}^{\mathcal{D}}_R$, then
$R = \mathbf{Id}^{\mathcal{B}}_R$ must hold because $R \subseteq \mathbf{Id}^{\mathcal{B}}_R  \subseteq \mathbf{Id}^{\mathcal{D}}_R$.  This confirms the fact that, during the iteration of refinement, once a reference set $R$ becomes $\mathcal{B}$-admissible, it preserves the admissibility in the future.
\end{remark}

\begin{lemma}\label{lem:compare_PR}
If $\mathcal{B} \subseteq  \mathcal{D}$ and $\mathbf{Id}^{\mathcal{B}}_R  = \mathbf{Id}^{\mathcal{D}}_R$, then $\mathbf{Pr}^{\mathcal{D}}_R \subseteq \mathbf{Pr}^{\mathcal{B}}_R$.
\end{lemma}

\begin{lemma}\label{lem:compare_RD}
If $\mathcal{B} \subseteq  \mathcal{D}$,  $\mathbf{Id}^{\mathcal{B}}_R  = \mathbf{Id}^{\mathcal{D}}_R$, and $\mathbf{Pr}^{\mathcal{B}}_R = \mathbf{Pr}^{\mathcal{D}}_R$, then  $\mathbf{Rd}^{\mathcal{B}}_R \subseteq \mathbf{Rd}^{\mathcal{D}}_R$.
\end{lemma}

\begin{lemma}\label{lem:compare_all}
If $\mathcal{B} \subseteq  \mathcal{D}$, and moreover $\mathbf{Id}^{\mathcal{B}}_R  = \mathbf{Id}^{\mathcal{D}}_R$, $\mathbf{Pr}^{\mathcal{B}}_R = \mathbf{Pr}^{\mathcal{D}}_R$, and $\mathbf{Rd}^{\mathcal{B}}_R = \mathbf{Rd}^{\mathcal{D}}_R$ for every $R$,  then $\mathcal{B} =  \mathcal{D}$.
\end{lemma}

The purpose of Lemma~\ref{lem:compare_ID} to Lemma~\ref{lem:compare_all} is to get the following fact.
\begin{proposition}\label{prop:number_iter}
The total number of iterations (i.e.~refinement operations) in our algorithm is exponentially bounded.
\end{proposition}

\subsection{The Initial Base} \label{subsec:initial_base}

The initial base $\mathcal{B}_0 = \{\mathcal{B}_{0,R}\}_R$ is defined as follows:
\begin{itemize}
\item
$\mathbf{Id}_{R} \coloneqq \mathbf{C}_{\mathrm{G}} = \{X \in \mathbf{C} \;|\;    | X |_{\mathrm{wk}} = 0 \}$ for every $R$.
\end{itemize}
Thus $\mathbf{C}_{\mathrm{G}}$ is the only $\mathcal{B}_0$-admissible set.
\begin{itemize}
\item
$\mathbf{Pr}_{R} \coloneqq \{[P]_{\mathbf{C}_{\mathrm{G}}}\}$ where $[P]_{\mathbf{C}_{\mathrm{G}}}$ is  the $<_{\mathbf{C}_{\mathrm{G}}}$-minimum $\mathbf{C}_{\mathrm{G}}$-block satisfying $|P|_{\mathrm{wk}}= 1$.  $\mathbf{Pr}_{R} \coloneqq \emptyset$ in case $\mathbf{C}_{\mathrm{G}} = \mathbf{C}$.

\item
$\mathbf{Cm}_{R} \coloneqq \mathbf{C}_{\mathbf{C}_{\mathrm{G}}} \setminus \mathbf{Pr}_{R}$.

\item
$\mathbf{Dc}_{R}([X]_{\mathbf{C}_{\mathrm{G}}}) \coloneqq \underbrace{[P]_{\mathbf{C}_{\mathrm{G}}}  \ldots  [P]_{\mathbf{C}_{\mathrm{G}}}}_{ | X |_{\mathrm{wk}}\textrm{ times}}$ if $[X]_{\mathbf{C}_{\mathrm{G}}} \in \mathbf{Cm}_{R}$.


\item
$\mathbf{Rd}_{R}([P]_{\mathbf{Id}_{R}}) \coloneqq \mathbf{C}_{\mathrm{G}} $, if $\mathbf{Pr}_{R} = \{[P]_{\mathbf{Id}_{R}}\}$.
\end{itemize}
Now $\mathcal{B}_{0,R}$ is defined as $(\mathbf{Id}_R, \mathbf{Pr}_R, \mathbf{Cm}_R ,\mathbf{Dc}_R, \mathbf{Rd}_R)$.
Notice that  $\mathcal{B}_{0,R}$ is the same for every $R \subseteq \mathbf{C}_{\mathrm{G}}$.

\begin{lemma}
$\alpha \stackrel{\mathcal{B}_0}{=}_{R} \beta$ if and only if $ | \alpha |_{\mathrm{wk}}  = | \beta |_{\mathrm{wk}}$.
\end{lemma}
\begin{lemma}\label{lem:corrct_init}
$\widehat{\mathcal{B}} \subseteq {\mathcal{B}_0}$. Namely, ${\simeq_R} \subseteq  {\stackrel{\mathcal{B}}{=}_{R}}$ for every $R$.
\end{lemma}
One can check that all the five constraints described in Section~\ref{subsec:decomposition_base} are satisfied by $\mathcal{B}_0$.

\subsection{Expansion Conditions}\label{subsec:expansion}

We start to define new base $\mathcal{B}$ from the old base $\mathcal{D}$. This is the core of our algorithm.  The newly constructed $\stackrel{\mathcal{B}}{=}_{R}$ is made to be a decreasing bisimulation with $R$-expansion of $\stackrel{\mathcal{D}}{=}_{R}$.  Referring to Lemma~\ref{lem:expansion}, we have
${\stackrel{\mathcal{B}}{=}_{R}}  \subseteq {\stackrel{\mathcal{D}}{=}_{R}}$, and for every $\alpha, \beta$ in $R$-nf, the following conditions hold whenever $\alpha \stackrel{\mathcal{B}}{=}_{R} \beta$:
\begin{enumerate}
\item
if  $\alpha \Longrightarrow_R \epsilon$,  then $\beta \Longrightarrow_R \epsilon$;

\item
Whenever $\alpha \stackrel{\tau}{\longrightarrow}_R \alpha'$,
\begin{enumerate}
\item
if $ \alpha \stackrel{\tau}{\longrightarrow}_R \alpha'$ is $\stackrel{\mathcal{B}}{=}_{R}$-decreasing, then
$\beta  \stackrel{\stackrel{\mathcal{B}}{=}_{R}}{\Longrightarrow}_{R} \cdot \stackrel{\tau}{\longrightarrow}_{R} \beta'$ for some $\beta'$ such that $\alpha' \stackrel{\mathcal{B}}{=}_{R} \beta'$;

\item
if $ \alpha \stackrel{\tau}{\longrightarrow}_R \alpha'$ is not $\stackrel{\mathcal{B}}{=}_{R}$-decreasing and $\alpha \not\stackrel{\mathcal{D}}{=}_{R} \alpha'$, then
$\beta  \stackrel{\stackrel{\mathcal{B}}{=}_{R}}{\Longrightarrow}_{R} \cdot \stackrel{\tau}{\longrightarrow}_{R} \beta'$ for some $\beta'$ such that $\alpha' \stackrel{\mathcal{D}}{=}_{R} \beta'$;
\end{enumerate}

\item
Whenever $\alpha \stackrel{a}{\longrightarrow}_R \alpha'$,
\begin{enumerate}
\item
if $\alpha \stackrel{a}{\longrightarrow}_R \alpha'$ is $\stackrel{\mathcal{B}}{=}_{R}$-decreasing,  then
$\beta \stackrel{\stackrel{\mathcal{B}}{=}_{R}}{\Longrightarrow}_{R} \cdot \stackrel{a}{\longrightarrow}_{R} \beta'$ for some $\beta'$ such that $\alpha' \stackrel{\mathcal{B}}{=}_{R} \beta'$.

\item
if $\alpha \stackrel{a}{\longrightarrow}_R \alpha'$  is not $\stackrel{\mathcal{B}}{=}_{R}$-decreasing,  then
$\beta \stackrel{\stackrel{\mathcal{B}}{=}_{R}}{\Longrightarrow}_{R} \cdot \stackrel{a}{\longrightarrow}_{R} \beta'$ for some $\beta'$ such that $\alpha' \stackrel{\mathcal{D}}{=}_{R} \beta'$.
\end{enumerate}
\end{enumerate}
The above conditions will be called {\em expansion conditions} in the following. Our task is to construct $\mathcal{B}$ from $\mathcal{D}$ and validate these expansion conditions.
From expansion conditions we can see that, in case ${\stackrel{\mathcal{B}}{=}_{R}}  = {\stackrel{\mathcal{D}}{=}_{R}}$, $\stackrel{\mathcal{B}}{=}_{R}$ must be an $R$-bisimulation. Thus when ${\simeq_R}  \subsetneq  {\stackrel{\mathcal{D}}{=}_{R}}$, we must have ${\stackrel{\mathcal{B}}{=}_{R}}  \subsetneq {\stackrel{\mathcal{D}}{=}_{R}}$.

Basically, the construction contains  three steps:
\begin{enumerate}
\item
Determine $\mathbf{Id}^{\mathcal{B}}_R$ for every qualified $R$.  After that, we know whether a given $R$ is $\mathcal{B}$-admissible. Note that some $R$'s which are not $\mathcal{D}$-admissible can be $\mathcal{B}$-admissible.

\item
Determine other constituents of $\mathcal{B}_R$ for every $\mathcal{B}$-admissible $R$.

\item
For non-$\mathcal{B}$-admissible $R$'s,  $\mathcal{B}_{\mathbf{Id}_R}$ is copied to $\mathcal{B}_{R}$.
\end{enumerate}

The third step is relatively trivial.  Its correctness depends on the following lemma.
\begin{lemma}\label{lem:step3}
If ${\simeq_{\mathbf{Id}^{\mathcal{B}}_R}} \subseteq {\stackrel{\mathcal{B}}{=}_{\mathbf{Id}^{\mathcal{B}}_R}}$, then ${\simeq_{R}} \subseteq {\stackrel{\mathcal{B}}{=}_{R}}$.
\end{lemma}
The first and second steps of the construction are described in Section~\ref{subsec:determiningID} and Section~\ref{subsec:determining_other}.

\subsection{Determining $\mathbf{Id}^{\mathcal{B}}_R$}\label{subsec:determiningID}

First of all, we must determine what $\mathbf{Id}^{\mathcal{B}}_R$ is.
This problem asks under what circumstance  we can believe that $X \stackrel{\mathcal{B}}{=}_{R} \epsilon$ for $X\in \mathbf{C}_{\mathrm{G}}$.  Be aware that Lemma~\ref{lem:compare_ID} confirms that $\mathbf{Id}^{\mathcal{B}}_R  \subseteq \mathbf{Id}^{\mathcal{D}}_R$.  The basic idea is to make use of the expansion conditions.

\begin{definition}\label{def:ID_candidate}
Let   $S$ be a set that makes $R \subseteq S  \subseteq \mathbf{Id}^{\mathcal{D}}_R$.  We call $S$ an {\em $\mathbf{Id}^{\mathcal{B}}_R$-candidate} if the following conditions are satisfied whenever $X \in S \setminus R$:
\begin{enumerate}
\item
If $X  \stackrel{\tau}{\longrightarrow}_R \alpha$ and $\alpha \not\in (\mathbf{Id}^{\mathcal{D}}_R)^{*}$, then
$\epsilon  \stackrel{\tau}{\longrightarrow}_{R} \beta$  for some $\beta$  such that $\mathtt{dcmp}_R^{\mathcal{D}}(\alpha) = \mathtt{dcmp}_R^{\mathcal{D}}(\beta)$.

\item
If $X \stackrel{a}{\longrightarrow}_R \alpha$, then
$\epsilon \stackrel{a}{\longrightarrow}_{R} \beta$ for some $\beta$ such that $\mathtt{dcmp}_R^{\mathcal{D}}(\alpha) = \mathtt{dcmp}_R^{\mathcal{D}}(\beta)$.
\end{enumerate}
\end{definition}
According to Definition~\ref{def:ID_candidate},
\begin{enumerate}
\item
$R$ is an $\mathbf{Id}^{\mathcal{B}}_R$-candidate.

\item
$\mathbf{Id}^{\mathcal{B}}_R$-candidates are closed under union.
\end{enumerate}
$\mathbf{Id}^{\mathcal{B}}_R$ is defined  as the largest $\mathbf{Id}^{\mathcal{B}}_R$-candidate. One fast way of computing $\mathbf{Id}^{\mathcal{B}}_R$ is described as procedure $\textsc{ComputingId}(R)$ in Fig.~\ref{fig:new_base_I}.

\begin{remark}
We can also determine $\mathbf{Id}^{\mathcal{B}}_R$ in exponential time even by directly enumerating all the  $\mathbf{Id}^{\mathcal{B}}_R$-candidates.

\end{remark}

It is easy to check
the following properties.
\begin{lemma}\label{lem:idem_id_R}
\begin{enumerate}
\item
$R \subseteq \mathbf{Id}^{\mathcal{B}}_R \subseteq \mathbf{Id}^{\mathcal{D}}_R$.

\item
If $R \subseteq S$, then $\mathbf{Id}^{\mathcal{B}}_R  \subseteq \mathbf{Id}^{\mathcal{B}}_S$.   If  $R \subseteq S \subseteq \mathbf{Id}^{\mathcal{B}}_R$, then $\mathbf{Id}^{\mathcal{B}}_R  = \mathbf{Id}^{\mathcal{B}}_S$.  In particular,  $\mathbf{Id}^{\mathcal{B}}_{\mathbf{Id}^{\mathcal{B}}_R} = \mathbf{Id}^{\mathcal{B}}_R$.
\end{enumerate}
\end{lemma}
According to $\mathcal{B}_R$-reduction rules in Section~\ref{subsec:decomposition_base},  $\alpha \stackrel{\mathcal{B}}{=}_{R} \epsilon$ if and only $\alpha \in (\mathbf{Id}^{\mathcal{B}}_R)^{*}$. Thus $(\mathbf{Id}^{\mathcal{B}}_R)^{*}$ is the only class that the $\stackrel{\mathcal{B}}{=}_{R}$-norm of whose members  are zero.  The correctness of the construction of $\mathbf{Id}^{\mathcal{B}}_R$ depends on the following lemma.
\begin{lemma}\label{lem:step1}
Assume that ${\simeq_R} \subseteq {\stackrel{\mathcal{D}}{=}_{R}}$. Then
$\alpha \simeq_R \epsilon$ implies $\alpha \stackrel{\mathcal{B}}{=}_{R} \epsilon$.
\end{lemma}

\begin{figure}[p]

\begin{center} \small
 \begin{tabular}{|p{8.2cm}|}\hline \vspace{-0ex}
 {\normalsize  $\textsc{ComputingId}(R)$:}

\begin{enumerate}

\item
  $\mathbf{Id}^{\mathcal{B}}_{R} \coloneqq \mathbf{Id}^{\mathcal{D}}_{R}$.

\item
\textbf{while} there exists $X \in \mathbf{Id}^{\mathcal{B}}_{R} - R$ such that one of the followings are violated:
\begin{itemize}
\item
If $X   \stackrel{\tau}{\longrightarrow}_{R} \alpha$ and $\alpha \not\in (\mathbf{Id}^{\mathcal{D}}_{R})^{*}$, then
$\epsilon  \stackrel{\tau}{\longrightarrow}_{R} \beta$  for some $\beta$  such that $\alpha\stackrel{\mathcal{D}}{=}_R \beta$.

\item
If $X \stackrel{a}{\longrightarrow}_R \alpha$, then
$\epsilon \stackrel{a}{\longrightarrow}_{R} \beta$ for some $\beta$ such that $\alpha\stackrel{\mathcal{D}}{=}_R \beta$.
\end{itemize}

\textbf{do}  remove $Y$ from $\mathbf{Id}^{\mathcal{B}}_{R}$ for every $Y \Longrightarrow_R \zeta\Longrightarrow_R \epsilon$ and $X$ appears in $\zeta$.

\textbf{end while}
\vspace{-1ex}
\end{enumerate}
   \\    \hline\vspace{-0ex}

 {\normalsize  $\textsc{Initializing}$:}
\begin{enumerate}
\item
\textbf{for}  every $\mathcal{B}$-admissible $R$

\quad  $\mathbf{Pr}^{\mathcal{B}}_R \coloneqq \emptyset$ ; $\mathbf{Cm}^{\mathcal{B}}_R \coloneqq \emptyset$.

\quad \textbf{for} every $[X]_R \in \mathbf{C}_{R}$

\quad \quad $\mathbf{Dc}^{\mathcal{B}}_R([X]_{R}) \coloneqq \bot$ ; $\mathbf{Rd}^{\mathcal{B}}_R([X]_{R}) \coloneqq \bot$.

\quad \quad $d_R[[X]_{R}]  \coloneqq \bot$.

\quad \textbf{end for}

\textbf{end for}

\item
$\mathbf{U}  \coloneqq \{ [X]_R \;|\;   R =  \mathbf{Id}^{\mathcal{B}}_R \mbox{ and }  [X]_{R} \in \mathbf{C}_{R}\}$;

$\mathbf{V}  \coloneqq \emptyset$;

$\mathbf{T}  \coloneqq \emptyset$.

\vspace{-1ex}
\end{enumerate}
   \\    \hline\vspace{-0ex}

 {\normalsize   $\textsc{Expand}_{R}( X, [Y_k]_{R_k} \ldots [Y_1]_{R_1})$: }
\begin{enumerate}
\item
\textbf{if} $X \not \stackrel{\mathcal{D}}{=}_R Y_k \ldots Y_1$   \textbf{then}

\quad \textbf{return} \textbf{false}.

\item
\textbf{if}
the followings conditions are met:
\begin{itemize}
\item
Whenever $[X]_{R} \stackrel{\ell}{\longmapsto}_R \alpha$, then

\begin{enumerate}
\item
if $d_R(\alpha) = m-1$, then $[Y_k]_{R_k} \stackrel{\ell}{\longmapsto}_{R_k} \zeta$ such that  $\alpha \stackrel{\mathcal{B}}{=}_R \zeta.Y_{k-1}.\ldots. Y_1$.

\item
else,  either $\ell = \tau$ and  $\alpha \stackrel{\mathcal{D}}{=}_R Y_k. \ldots .Y_1$, or $[Y_k]_{R_k} \stackrel{\ell}{\longmapsto}_{R_k} \zeta$ such that $\alpha \stackrel{\mathcal{D}}{=}_R \zeta.Y_{k-1}.\ldots. Y_1$.
\end{enumerate}

\item
Either $[X]_{R} \stackrel{\tau}{\longmapsto}_R \alpha$ for some $\alpha$ such that  $\alpha \stackrel{\mathcal{B}}{=}_R Y_k. \ldots .Y_1$; or whenever $[Y_k]_{R_k} \stackrel{\ell}{\longmapsto}_{R_k} \zeta$,
\begin{enumerate}
\item
if $d_R(\gamma.Y_{k-1}. \ldots .Y_1) = m-1$, then $[X]_{R} \stackrel{\ell}{\longmapsto}_R \alpha$ for some $\alpha$ such that  $\alpha \stackrel{\mathcal{B}}{=}_R \zeta.Y_{k-1}.\ldots. Y_1$.

\item
else, $[X]_{R} \stackrel{\ell}{\longmapsto}_R \alpha$ for some $\alpha$ such that $\alpha \stackrel{\mathcal{D}}{=}_R \zeta.Y_{k-1}.\ldots. Y_1$.
\end{enumerate}
\end{itemize}

\textbf{then}

\quad \textbf{return} \textbf{true}.

\textbf{else}

\quad \textbf{return} \textbf{false}.

\vspace{-1ex}
\end{enumerate}
   \\    \hline\vspace{-0ex}

 {\normalsize  $\textsc{ComputingRd}_R([X]_R)$: }
\begin{enumerate}
\item
$T \coloneqq  \{ W \;|\;  W.X \stackrel{\mathcal{D}}{=}_R X \}$;  $\mathbf{Rd}^{\mathcal{B}}_R([X]_R) \coloneqq T$.

\item
\textbf{while} there exists $Y \in \mathbf{Rd}^{\mathcal{B}}_R([X]_R)$ such that one of the followings are violated:
\begin{itemize}
\item
If $Y  \stackrel{\tau}{\longrightarrow} \zeta$ and $\zeta \not\in T^{*}$, then
$[X]_{R}  \stackrel{\tau}{\longmapsto}_R \beta$  for some $\beta$  such that $\zeta. X \stackrel{\mathcal{D}}{=}_R \beta$.

\item
If $Y \stackrel{a}{\longrightarrow}  \zeta$, then
$[X]_{R}  \stackrel{\tau}{\longmapsto}_{R} \beta$  for some $\beta$  such that $\zeta. X \stackrel{\mathcal{D}}{=} \beta$.
\end{itemize}

\textbf{do}
remove $Y$ from $\mathbf{Rd}^{\mathcal{B}}_R([X]_R)$ for every $Y \Longrightarrow_R \zeta\Longrightarrow_R \epsilon$ and $X$ appears in $\zeta$.

\textbf{end while}

\vspace{-1ex}
\end{enumerate}
   \\    \hline
    \end{tabular}
\end{center}
\caption{Constructing New Base: Part~I}\label{fig:new_base_I}
\end{figure}

\begin{figure}[t]

\begin{center}
 \begin{tabular}{|p{8.2cm}|}\hline\vspace{-0ex}
 \textsc{\normalsize  ConstructingNewBase:}

\begin{enumerate}
\item
\textbf{for} every $R$

\quad\textsc{ComputingId}$(R)$.

\textbf{end for}

\item
\textsc{Initializing}.

\item
$m  \coloneqq 1$.

\item \label{line:repeat}
\textbf{repeat}

\item \label{line:first_while}
\quad
\textbf{while} their exists  $[X]_R \in \mathbf{U}$ such that

\quad\quad\quad $[X]_R \stackrel{\ell}{\longmapsto}\gamma$ and $d_R[\gamma] = m - 1$,

\quad
\textbf{do}

\quad
\quad select one of such $[X]_R$ which is $<_R$-minimum.

\quad
\quad
$d_R[[X]_R] := m$.

\quad
\quad
\textbf{if} there exists  $\delta$ such that

\quad
\quad \quad
\quad  $\textsc{Expand}_{R}( X, \mathtt{dcmp}^{\mathcal{B}}_R(\delta))$,
\textbf{then}

\quad
\quad\quad
 put $[X]_R$ into  $\mathbf{Cm}^{\mathcal{B}}_R$.

 \quad\quad\quad$\mathbf{Dc}^{\mathcal{B}}_R([X]_R) \coloneqq \mathtt{dcmp}^{\mathcal{B}}_R(\delta)$.

\quad
\quad
\textbf{else}

\quad\quad\quad put $[X]_R$ into  $\mathbf{Pr}^{\mathcal{B}}_R$.

\quad
\quad\quad  $\textsc{ComputingRd}_R([X]_R)$ .

\quad
\quad
\textbf{end if}

\quad
\quad move $[X]_R$ from $\mathbf{U}$ to $\mathbf{V}$.

\quad \textbf{end while}

\item\label{line:second_while}
\quad
\textbf{while}  their exists  $[X]_R \in \mathbf{U}$ such that

\quad\quad\quad$[X]_R \stackrel{\tau}{\longmapsto}_R \gamma$ and $d_R(\gamma) = m$,

\quad
\textbf{do}

\quad
\quad
\textbf{if} $\textsc{Expand}_{R}( X, \mathtt{dcmp}^{\mathcal{B}}_R(\gamma))$, \textbf{then}

\quad\quad
\quad
 put $[X]_R$ into  $\mathbf{Cm}^{\mathcal{B}}_R$.

\quad
\quad\quad $d_R[[X]_R] := m$.

 \quad
\quad\quad$\mathbf{Dc}^{\mathcal{B}}_R([X]_R) \coloneqq \mathtt{dcmp}^{\mathcal{B}}_R(\alpha)$.

\quad
\quad \quad move $[X]_R$ from $\mathbf{U}$ to $\mathbf{V}$.

\quad
\quad
\textbf{else}

\quad\quad\quad  move $[X]_R$ from $\mathbf{U}$ to $\mathbf{T}$.

\quad
\quad
\textbf{end if}

\quad \textbf{end while}

\item
\quad put every block in $\mathbf{T}$ into $\mathbf{U}$.

\item
\quad
$m  \coloneqq m+1$.

\textbf{until} $\mathbf{U} = \emptyset$

\item
\textbf{for}  every non-$\mathcal{B}$-admissible  $R$

\quad $\mathcal{B}_R' \coloneqq  \mathcal{B}_{\mathbf{Id}^{\mathcal{B}}_R}'$.

\textbf{end for}
\vspace{-1ex}
\end{enumerate}
\\ \hline
    \end{tabular}
\end{center}
\caption{Constructing New Base: Part~II}\label{fig:new_base_II}
\end{figure}

\subsection{Determining Other Constituents of $\mathcal{B}_R$}\label{subsec:determining_other}

When  $\mathbf{Id}^{\mathcal{B}}_R$ has been determined for every $R$, we can construct the whole $\mathcal{B}$ via a greedy strategy.
Since $\mathcal{B}_R$ is completely determined by $\mathcal{B}_{\mathbf{Id}^{\mathcal{B}}_R}$, we only need to construct those $\mathcal{B}_R$'s in which $R$ is $\mathcal{B}$-admissible.

The algorithm in Fig.~\ref{fig:new_base_II} constructs $\mathcal{B}_{R}$ and compute the $\stackrel{\mathcal{B}}{=}_{R}$-norms of $[X]_R$'s for every $\mathcal{B}$-admissible $R$ at the same time via the greedy strategy. When the program starts an iteration of \textbf{repeat}-block at line~\ref{line:repeat}, it attempts to find
all the blocks $[X]_R$'s such that $\|X\|_{\stackrel{\mathcal{B}}{=}_{R}} = m$.  In the algorithm $d_R[[X]_R]$ is used to indicate $\|X\|_{\stackrel{\mathcal{B}}{=}_{R}}$.  We also write $d_R(\alpha)$ for $\| \alpha \|_{\stackrel{\mathcal{B}}{=}_{R}}$. Precisely,
\begin{equation} \label{eqn:d_R}
d_R(\alpha) \stackrel{\mathrm{def}}{=} \sum_{i=1}^{r} d_R[[X_i]_{R_i}],
\end{equation}
if  $\mathtt{dcmp}^{\mathcal{B}}_{R}(\alpha) = [X_r]_{R_{r}} [X_{r-1}]_{R_{r-1}}\ldots [X_1]_{R_{1}}$.
The algorithm maintains two sets $\mathbf{U}$ and $\mathbf{V}$. They forms a partition of all the $\mathcal{B}$-admissible blocks. (A block $[X]_R$ is $\mathcal{B}$-admissible if $R$ is $\mathcal{B}$-admissible. ) During the execution of the algorithm, we can move a certain $[X]_R$ from $\mathbf{U}$ to $\mathbf{V}$.  At that time, we  define the related information for $[X]_R$: determine whether $[X]_R$ is a $\mathcal{B}_{R}$-prime or a $\mathcal{B}_{R}$-composite; compute $\mathbf{Rd}^{\mathcal{B}}_{R}([X]_R)$ if it is a $\mathcal{B}_{R}$-prime; compute $\mathbf{Dc}^{\mathcal{B}}_{R}([X]_R)$ if it is a $\mathcal{B}_{R}$-composite.

In the following, we say that a block $[X]_R$ is {\em treated} if  $[X]_R \in \mathbf{V}$. When $[X]_R$ is selected during the execution of the algorithm, it is called {\em under treating}.
Every time $[X]_R$ is under treating, we confirm the following fact:
\begin{itemize}
\item
If $[X]_{R} \stackrel{\ell}{\longmapsto}_{R} \alpha$ with $d_R(\alpha) = m - 1$,  ($\stackrel{\mathcal{B}}{=}_R$-decreasing)

 \item
or if $[X]_{R} \stackrel{\tau}{\longmapsto}_{R} \alpha$ with $d_R(\alpha) = m$, (possibly $\stackrel{\mathcal{B}}{=}_R$-preserving. )
\end{itemize}
then all the blocks in the $\mathcal{B}_R$-decomposition of $\alpha$ have been treated, thus the related information for $\alpha$ is already known.
The first case is guaranteed by the non-decrease of $m$. The second case is by the aid of the order $<_R$.
These are two cases which correspond to two different possibilities that the $\mathcal{B}_R$-norm of $[X]_{R}$ can be declared as $m$.

\begin{remark}
If $R$ is not $\mathcal{B}$-admissible, the second part of the above fact cannot always be satisfied.  This is one of the reasons why we must construct $\mathcal{B}_R$ only for $\mathcal{B}$-admissible $R$'s at first  and then copy back to all other $R$'s.
\end{remark}

\subsubsection{Treating $[X]_{R}$: The First Possibility. }

There exists a witness path of $[X]_R$ starting with a ${\stackrel{\mathcal{B}}{=}_{R}}$-decreasing transition.
That is,
$[X]_R \stackrel{\ell}{\longmapsto}_R \gamma$ for some $\gamma$ such that $\|\gamma\|_{\stackrel{\mathcal{B}}{=}_{R}} = m-1$.
This possibility is treated  via the \textbf{while}-block at line~\ref{line:first_while}.

At the time we have known $\mathtt{dcmp}^{\mathcal{B}}_R(\gamma)$.
The first problem is to decide whether $[X]_R$ is a prime or a composite. To this end, we try to guess a candidate for  decomposition of $[X]_R$, say $[Y_k]_{R_k}[Y_{k-1}]_{R_{k-1}} \ldots [Y_1]_{R_1}$ with $R_1 = R$ and $R_{i+1} = \mathbf{Rd}^{\mathcal{B}}_{R_{i}}(Y_i)$ for $1 \leq i < k$. If this decomposition is `right', we will have $X  \stackrel{\mathcal{B}}{=}_{R} Y_k\ldots Y_1$. Since $\stackrel{\mathcal{B}}{=}_{R}$ will be ensured to be a decreasing bisimulation.
We must have a matching of $[X]_R \stackrel{\ell}{\longmapsto}_R \gamma$ from $[Y_k]_{R_k} \ldots [Y_1]_{R_1}$, which must be induced by $[Y_k]_{R_k}$.  From the above investigation, we can require that $[Y_{k-1}]_{R_{k-1}} \ldots [Y_1]_{R_1}$ is a suffix of $\mathtt{dcmp}^{\mathcal{B}}_R(\gamma)$. In summary, in order to guess a candidate for decomposition of $[X]_R$, we need to:
\begin{itemize}
\item
Guess $k$.  Thus $[Y_{k-1}]_{R_{k-1}} \ldots [Y_1]_{R_1}$ is obtained from $\mathtt{dcmp}^{\mathcal{B}}_R(\gamma)$.

\item
Guess $[Y_{k}]_{R_{k}}$,  which ensures that $\| Y_k \ldots Y_1\|_{\stackrel{\mathcal{B}}{=}_{R}} = m$.
\end{itemize}
If $k > 1$, then every $\| Y_i \|_{\stackrel{\mathcal{B}}{=}_{R_i}}  < m$ thus $[Y_i]_{R_i} \in \mathbf{V}$ for every $1 \leq i \leq k$.
If $k = 1$, then we guess $\mathbf{Dc}^{\mathcal{B}}_{R}([X]_{R}) = [Y_1]_{R}$ for  $Y_1 <_R X$. If every time we pick out the $<_R$-minimum such $[X]_R$, then we can ensure that $[Y_1]_{R} \in \mathbf{V}$.

\begin{remark}
It is probable  that $[X_1]_R <_R [X_2]_R$ and  $[X_2]_R$ is treated before $[X_1]_R$. For example, it is probable that
$\| X_2 \|_{\stackrel{\mathcal{B}}{=}_{R}} <_R \| X_1 \|_{\stackrel{\mathcal{B}}{=}_{R}}$.  But taking our way, we can ensure that $[X_1]_R$ is treated before $[X_2]_R$ whenever  $[X_1]_R <_R [X_2]_R$ and  $[X_1]_R \stackrel{\mathcal{B}}{=}_{R} [X_2]_R$.
\end{remark}

After one candidate  $[Y_k]_{R_k} \ldots [Y_1]_{R_1}$ is  found, we will make use of the expansion conditions (Section~\ref{subsec:expansion}) to decide whether $\mathbf{Dc}^{\mathcal{B}}_{R}([X]_{R})$ can be defined as $[Y_k]_{R_k} \ldots [Y_1]_{R_1}$. This is done by $\textsc{Expand}_{R}(X, [Y_k]_{R_k} \ldots [Y_1]_{R_1})$ defined in Fig.~\ref{fig:new_base_I}.

\subsubsection{Treating $[X]_{R}$: The Second Possibility. }

Every witness path of $[X]_R$ starts with a ${\stackrel{\mathcal{B}'}{=}_{R}}$-preserving silent transition. That is,
$\|\gamma\|_{\stackrel{\mathcal{B}'}{=}_{R}} \geq m$ for every $\gamma$ such that $[X]_R \stackrel{\ell}{\longmapsto}_R \gamma$, but $[X]_R \stackrel{\tau}{\longmapsto}_R \gamma$ for some $\gamma$ such that $X  \stackrel{\mathcal{B}}{=}_{R} \gamma$ (which needs to be confirmed) and $\|\gamma\|_{\stackrel{\mathcal{B}}{=}_{R}} = m$.
This possibility is treated  via the \textbf{while}-blocks at line~\ref{line:second_while}.

This possibility is relatively easy because there is no need to guess the candidates for decomposition of $[X]_{R}$. If  $X  \stackrel{\mathcal{B}}{=}_{R} \gamma$, we must have $\mathbf{Dc}^{\mathcal{B}}_{R}([X]_{R}) = \mathtt{dcmp}^{\mathcal{B}}_R(\gamma)$.  Let $[Y_k]_{R_k} \ldots [Y_1]_{R_1}$ be $\mathtt{dcmp}^{\mathcal{B}}_R(\gamma)$. Then we will check $\textsc{Expand}_{R}(X, [Y_k]_{R_k} \ldots [Y_1]_{R_1})$.
Note that it is unnessesary to check the second half of expansion conditions, because $[X]_R \stackrel{\tau}{\longmapsto}_R \cdot  \stackrel{\mathcal{B}}{=}_R \gamma$ always holds.

\subsubsection{Determining $\mathbf{Rd}^{\mathcal{B}}_R([X]_R)$}

When $[X]_R$ is declared as a $\mathcal{B}_{R}$-prime. There is an extra work: define $\mathbf{Rd}^{\mathcal{B}}_R([X]_R)$. Intuitively  $\mathbf{Rd}^{\mathcal{B}}_R([X]_R)$ contains all the constants $Y$ which make $Y.X \stackrel{\mathcal{B}}{=}_{R} X$.  It is necessary that $Y \in \mathbf{C}_{\mathrm{G}}$.
We can use the same way of determining $\mathbf{Id}^{\mathcal{B}}_R$ to determine  $\mathbf{Rd}^{\mathcal{B}}_R([X]_R)$.

\begin{definition}\label{def:RD_candidate}
Let $[X]_R$ be a $\mathcal{B}_R$-prime, let  $T$ be the set $\{ W \;|\;  W.X \stackrel{\mathcal{D}}{=}_R X \}$, and let $S \subseteq T$.  We call $S$ an {\em $\mathbf{Rd}^{\mathcal{B}}_R([X]_R)$-candidate} if the following conditions are satisfied whenever $Y \in S$:
\begin{enumerate}
\item
If $Y  \stackrel{\tau}{\longrightarrow} \zeta$ and $\zeta \not\in T^{*}$, then
$[X]_{R}  \stackrel{\tau}{\longmapsto} \beta$  for some $\beta$  such that $\zeta. X \stackrel{\mathcal{D}}{=}_R \beta$.

\item
If $Y \stackrel{a}{\longrightarrow}  \zeta$, then
$[X]_{R}  \stackrel{\tau}{\longmapsto}_{R} \beta$  for some $\beta$  such that $\zeta. X \stackrel{\mathcal{D}}{=}_R \beta$.
\end{enumerate}
\end{definition}
According to Definition~\ref{def:RD_candidate},
\begin{enumerate}
\item
$\emptyset$ is an $\mathbf{Rd}^{\mathcal{B}}_R([X]_R)$-candidate.

\item
$\mathbf{Rd}^{\mathcal{B}}_R([X]_R)$-candidates are closed under union.
\end{enumerate}
$\mathbf{Rd}^{\mathcal{B}}_R([X]_R)$ is defined as the largest $\mathbf{Rd}^{\mathcal{B}}_R([X]_R)$-candidate. One fast way of computing $\mathbf{Rd}^{\mathcal{B}}_R([X]_R)$ is described as procedure $\textsc{ComputingRd}_R([X]_R)$ in Fig.~\ref{fig:new_base_I}.

\subsubsection{Basic Properties of the Construction}

We point out the following important properties.
\begin{lemma}\label{lem:admissible_Rd}
$\mathbf{Rd}^{\mathcal{B}}_R([X]_R)$ constructed above is $\mathcal{B}$-admissible.
\end{lemma}
\begin{lemma}\label{lem:corrct_norm}
$d_R[[X]_R]$ computed in our algorithm is equal to $\| X \|_{\stackrel{\mathcal{B}}{=}_R}$. As an inference , $d_R(\alpha) = \| \alpha \|_{\stackrel{\mathcal{B}}{=}_R}$.
\end{lemma}
\begin{lemma}\label{lem:correct_contain_1}
${\mathcal{B}} \subseteq {\mathcal{D}}$. Moreover, if ${\widehat{\mathcal{B}} } \subsetneq {\mathcal{D}}$, then ${\mathcal{B}}  \subsetneq {\mathcal{D}}$.
\end{lemma}

\subsection{The Correctness of the Refinement Operation}
Remember Lemma~\ref{lem:corrct_init} and Lemma~\ref{lem:correct_contain_1}. The remain thing is to confirm the following fact.
\begin{theorem}\label{thm:correctness}
Suppose that ${\widehat{\mathcal{B}} } \subseteq {\mathcal{D}}$, then ${\widehat{\mathcal{B}} } \subseteq {\mathcal{B}}$. Namely, $\alpha \simeq_R \beta$ implies $\alpha \stackrel{\mathcal{B}}{=}_R \beta$ for every $R$.
\end{theorem}
It is enough to prove Theorem~\ref{thm:correctness} under the assumption that $R$'s are $\mathcal{B}$-admissible.  If  $ {\simeq_R} \subseteq {\stackrel{\mathcal{B}}{=}_R}$ for every $\mathcal{B}$-admissible $R$'s, then for non-$\mathcal{B}$-admissible $R$ we have ${\simeq_R} \subseteq {\simeq_{\mathbf{Id}^{\mathcal{B}}_R}} \subseteq {\stackrel{\mathcal{B}}{=}_{\mathbf{Id}^{\mathcal{B}}_R}} = {\stackrel{\mathcal{B}}{=}_R}$.

The correctness of Theorem~\ref{thm:correctness} relies on some important observations. The following one is crucial.
\begin{lemma}\label{lem:unique_expansion}
\sloppy
Let $R$ be $\mathcal{B}$-admissible. Let $\gamma = Y_k \ldots Y_1$ and $\delta = Z_l \ldots Z_1$ such that
$[Y_k]_{R_k} [Y_{k-1}]_{R_{k-1}} \ldots [Y_1]_{R_1}$ and
$[Z_l]_{S_l} [Z_{l-1}]_{S_{l-1}} \ldots [Z_1]_{S_1}$ are two $\mathcal{B}_R$-decompositions, in which
$R_1, S_1 = R$ and $R_{i+1} = \mathbf{Rd}_{R_{i}}(Y_i)$ for $1 \leq i < k$ and $S_{j +1} = \mathbf{Rd}_{S_{j}}(Z_j)$ for $1 \leq j < l$.  If $\gamma$ and $\delta$ satisfy the expansion conditions for $\mathcal{B}_R$, then we have $k = l$, $R_i = S_i$ and $[Y_i]_{R_i} = [Z_i]_{S_i}$ for $1 \leq i \leq k$.
\end{lemma}

Lemma~\ref{lem:unique_expansion} confirms that, when $[X]_R$ is being treated, at most one decomposition candidate $[Y_k]_{R_k} \ldots [Y_1]_{R_1}$ can make  $\textsc{Expand}_{R}(X, [Y_k]_{R_k} \ldots [Y_1]_{R_1})$ return \textbf{true}. If such a candidate exists, it is declared as $\mathbf{Dc}^{\mathcal{B}}_R([X]_R)$ and
and $[X]_R$ is declared as a $\mathcal{B}_R$-composite; otherwise $[X]_R$ is declared as a $\mathcal{B}_R$-prime.
Decreasing bisimulation property is crucial to validate Lemma~\ref{lem:unique_expansion}. This is why $\stackrel{\mathcal{B}}{=}_R$ must be constructed as a decreasing bisimulation with $R$-expansion of $\stackrel{\mathcal{D}}{=}_{R}$ (Section~\ref{subsec:expansion}), rather than simply defined as $R$-expansion of $\stackrel{\mathcal{D}}{=}_{R}$.

Apparently, the proof of Theorem~\ref{thm:correctness} should be done by induction. Remember that our algorithm maintains a set $\mathbf{V}$, containing all the blocks which have been treated. We will suppose that $R$ is $\mathcal{B}$-admissible and let $[X]_R$ be a block which is about to be put into $\mathbf{V}$. We try to prove Proposition~\ref{prop:correctness}.
A process $\alpha$ is called {\em $\mathcal{B}_{R, \mathbf{V}}$-applicable} if the derivation of $\alpha \stackrel{\mathcal{B}}{\rightarrow}_R \mathtt{dcmp}^{\mathcal{B}}_{R}(\alpha)$  (refer to Section~\ref{subsec:decomposition_base}) only depends on the blocks in $\mathbf{V}$.
The following statement is used as induction hypothesis:
\begin{enumerate}
\item[\textbf{IH}.]
Let $S$ be an arbitrary $\mathcal{B}$-admissible set. Suppose that $\gamma$ is $\mathcal{B}_{S,\mathbf{V}}$-applicable, and $\mathtt{dcmp}^{\widehat{\mathcal{B}}}_{S}(\gamma) = [W_u]_{S_u} \ldots[W_1]_{S_1}$. Then $W_u\ldots W_1$ is $\mathcal{B}_{S,\mathbf{V}}$-applicable, and $\mathtt{dcmp}^{\mathcal{B}}_{S,\mathbf{V}}(\gamma) = \mathtt{dcmp}^{\mathcal{B}}_{S,\mathbf{V}}(W_u \ldots W_1)$. That is, $\gamma \stackrel{\mathcal{B}}{=}_S W_u \ldots W_1$.
\end{enumerate}
When $\mathbf{V}$ contains all blocks, we can get Theorem~\ref{thm:correctness}.

Making use of \textbf{IH} we can establish the following.
\begin{lemma}\label{lem:pre_correctness}
Suppose $S$ is $\mathcal{B}$-admissible and $\mathtt{dcmp}^{\widehat{\mathcal{B}}}_{S}(W) = [W_u]_{S_u} \ldots[W_1]_{S_1}$, and $W_u \ldots W_1$ is $\mathcal{B}_{S,\mathbf{V}}$-applicable.
\begin{enumerate}
\item
If $\|W_u \ldots W_1\|_{\stackrel{\mathcal{B}}{=}_S} < m$, then $[W]_S \in \mathbf{V}$ and $W \stackrel{\mathcal{B}}{=}_S W_u \ldots W_1$.

\item
If $\|W_u \ldots W_1\|_{\stackrel{\mathcal{B}}{=}_S} = m$ and $W <_S X$, then $[W]_S \in \mathbf{V}$ and $W \stackrel{\mathcal{B}}{=}_S W_u \ldots W_1$.
\end{enumerate}
\end{lemma}
With Lemma~\ref{lem:pre_correctness} and \textbf{IH}, we can show the following auxiliary lemma.
\begin{lemma}\label{lem:applicable}
Suppose $[X]_R$ be a $\widehat{\mathcal{B}}_{R}$-composite, and assume that $\mathbf{Dc}^{\widehat{\mathcal{B}}}_{R}([X]_R) = [Z_t]_{R_t} \ldots[Z_1]_{R_1}$.  Then $Z_t \ldots Z_1$ is $\mathcal{B}_{R,\mathbf{V}}$-applicable.
\end{lemma}
Now, the following Proposition~\ref{prop:correctness} is obtained by Lemma~\ref{lem:applicable} and Lemma~\ref{lem:unique_expansion}, and finally Theorem~\ref{thm:correctness} is proved.
\begin{proposition}\label{prop:correctness}
Suppose $[X]_R$ be a $\widehat{\mathcal{B}}_{R}$-composite, and assume that $\mathbf{Dc}^{\widehat{\mathcal{B}}}_{R}([X]_R) = [Z_t]_{R_t} \ldots[Z_1]_{R_1}$.  Then $X \stackrel{\mathcal{B}}{=}_R Z_t \ldots Z_1$.
\end{proposition}

\subsection{Remark}
\begin{remark}
The refinement steps defined in previous works~\cite{DBLP:journals/tcs/HirshfeldJM96,DBLP:conf/fsttcs/CzerwinskiL10,DBLP:journals/corr/He14a} have the following interesting property, which says that $\stackrel{\mathcal{B}}{=}$ is the {\em largest}  decreasing bisimulation with expansion of $\stackrel{\mathcal{D}}{=}$, in which
\begin{itemize}
\item
$\mathcal{B}$ and $\mathcal{D}$ are the new and the old base corresponding to the related works, and $\stackrel{\mathcal{B}}{=}$ is the equivalence relation generated by $\mathcal{B}$;

\item
`decreasing' is syntectic;

\item
`decreasing bisimulation with expansion' is a simplified version of the one in Definition~\ref{def:_decreasing_bisimulation_with_expansion}.
\end{itemize}
This fact is also pointed out in~\cite{DBLP:journals/corr/He14a}. Before this, a step of refinement is divided into two stages, and in~\cite{DBLP:journals/corr/He14a}, it is pointed out that this two-stage understanding does not fit well for branching bisimilarity, thus the notion of {\em decreasing bisimulation with expansion} is invented accordingly.

In this paper, however, we do not claim that $\stackrel{\mathcal{B}}{=}_R$ is the {\em largest} decreasing bisimulation with expansion of $\stackrel{\mathcal{D}}{=}_R$.   We surmise that the `largest' does not always make sense, because semantic norms take the place of syntactic ones.  Fortunately, the correctness of the algorithm does not rely on the `largest'. It only relies on the three requirements stated at the beginning of Section~\ref{sec:naive-algorithm}.
\end{remark}

\subsection{The Time Complexity}
The running time of our algorithm is exponentially bounded, according to its description, together with
Lemma~\ref{exponential_bound_st}, Lemma~\ref{lem:norm_size}, Lemma~\ref{lem:bound_dc}, and Proposition~\ref{prop:number_iter}. Finally, we can conclude.
\begin{theorem}
Branching bisimilarity on normed BPA is EXPTIME-complete.
\end{theorem}

\section{Examples}

\subsection{Example One}

Let us illustrate the algorithm for the following normed BPA system $\Gamma = (\mathbf{C}, \mathcal{A}, \Delta)$ in which
\begin{itemize}
\item
 $\mathbf{C} = \{A_0, A_1, B , C \}$;

\item
$\mathcal{A} = \{a, b,  \tau\}$;

\item
$\Delta$ contains the following rules:
\begin{center}
$A_0 \stackrel{a}{\longrightarrow} A_1$, \; $A_1 \stackrel{a}{\longrightarrow} A_0$, \;
 $A_0 \stackrel{b}{\longrightarrow} \epsilon$, \; $A_1 \stackrel{b}{\longrightarrow} B$,

 $B \stackrel{a}{\longrightarrow} \epsilon$, \; $B \stackrel{\tau}{\longrightarrow} \epsilon$, \;
$C \stackrel{a}{\longrightarrow} C$, \;  $C \stackrel{\tau}{\longrightarrow} \epsilon$.
\end{center}
\end{itemize}

By direct observation, we have $A_0.C \simeq A_1.C$ but $A_0 \not\simeq A_1$,
this observation tells us that $A_0 \simeq_{\{C\}} A_1$.

Another observation is that, $A_0.A_0.C \simeq A_1.A_0.C$ but
$A_0.A_0 \not\simeq A_1.A_0$, which tells us that even if $A_0 \not\in \mathsf{Rd}(C)$, we still cannot
cancel the rightmost $C$.

Below we will demonstrate the behaviour of the algorithm on this system $\Gamma$ to get more valuable facts.

\subsubsection{Preprocessing}
The ground constants $\mathbf{C}_{\mathrm{G}} = \{B,C\}$. Thus the reference set can be $\emptyset$, $\{B\}$, $\{C\}$, and $\{B,C\}$.  All these sets are qualified. We have  $[X]_R = \{X\}$ for every $R \subseteq \mathbf{C}_{\mathrm{G}}$ and $X \in \mathbf{C} \setminus R$.  We can also find that the $R$-propagating of $X$ are always the empty set for every $R \subseteq \mathbf{C}_{\mathrm{G}}$ and $X \in \mathbf{C} \setminus R$. Thus $[X]_R \longmapsto_R \alpha$ if and only if  $X \longrightarrow_R \alpha$. Therefore, we will simply write $X$  for $[X]_R$.

We know that $\emptyset$-transitions $\longrightarrow_{\emptyset}$ is exactly $\longrightarrow$.
For future use, we list the $\{B,C\}$-transitions $\longrightarrow_{\{B,C\}}$:
\begin{center}
$A_0 \stackrel{a}{\longrightarrow}_{\{B,C\}} A_1$, \; $A_1 \stackrel{a}{\longrightarrow}_{\{B,C\}} A_0$, \;
 $A_0 \stackrel{b}{\longrightarrow}_{\{B,C\}} \epsilon$,

 $A_1 \stackrel{b}{\longrightarrow}_{\{B,C\}}  \epsilon$, \;
 $\epsilon \stackrel{a}{\longrightarrow}_{\{B,C\}} \epsilon$, \; $\epsilon \stackrel{\tau}{\longrightarrow}_{\{B,C\}} \epsilon$.
\end{center}

The order of $R$-blocks does not matter in this example. We choose the following orders:
\begin{itemize}
\item
$A_0 <_{\emptyset} A_1 <_{\emptyset} B <_{\emptyset} C$.

\item
$A_0 <_{\{B\}} A_1 <_{\{B\}}  C$.

\item
$A_0 <_{\{C\}} A_1 <_{\{C\}}  B$.

\item
$A_0 <_{\{B, C\}} A_1$.

\end{itemize}

Finally,  $|B|_{\mathrm{wk}} = |C|_{\mathrm{wk}} = 0$; $|A_0|_{\mathrm{wk}} = |A_1|_{\mathrm{wk}} = 1$.

\subsubsection{The Initial Base}
Now we define the initial base $\mathcal{B}_0$ according to Section~\ref{subsec:initial_base}. We know that
$\mathbf{Id}_{R}=\{B,C\}$ for every $R \subseteq \{B, C\}$. The only $\mathcal{B}_0$-admissible set is $\{B, C\}$. Therefore $\mathcal{B}_{0,R} = \mathcal{B}_{0,\{B, C\}}$ for every $R\subseteq \{B, C\}$,  and $\mathcal{B}_{0,\{B, C\}}$ is defined as follows:
\begin{itemize}
\item
$\mathbf{Pr}_{\{B,C\}}=\{A_0\}$.

\item
$\mathbf{Cm}_{\{B,C\}}=\{A_1\}$.

\item
$\mathbf{Dc}_{\{B,C\}}(A_1)=\{A_0\}$

\item
$\mathbf{Rd}_{\{B,C\}}(A_0)= \{B,C\}$.
\end{itemize}

The date of the initial base is summarized as follows:
{\small
\begin{tabular}{|c||c|c|c|c|c|c|} \hline
    blocks          & ord & norm & $\mathbf{Pr}$  & $\mathbf{Cm}$  & $\mathbf{Rd}$   & $\mathbf{Dc}$  \\
\hline\hline
$[A_0]_{\{B,C\}}$   & 1  & 1   & \checkmark   &              &  $\{B,C\}$   &      \\ \hline
$[A_1]_{\{B,C\}}$   & 2  & 1   &              &  \checkmark  &              &  $[A_0]_{\{B,C\}}$    \\
\hline
\end{tabular}
}

Finally, $\mathcal{B}_0$ is assigned to $\mathcal{D}$.

\subsubsection{The 1st Iteration}

First, we calculate $\mathbf{Id}^{\mathcal{B}}_{R}$ for every $R \subseteq \mathbf{C}_{\mathrm{G}}$ via $\textsc{ComputingId}(R)$ . We obtain:
\begin{itemize}
\item
$\mathbf{Id}^{\mathcal{B}}_{\emptyset}=\emptyset$, and

\item
$\mathbf{Id}^{\mathcal{B}}_{\{B\}}=\mathbf{Id}^{\mathcal{B}}_{\{C\}} = \mathbf{Id}^{\mathcal{B}}_{\{B,C\}} = \{B,C\}$.
\end{itemize}
Thus only $\emptyset$ and $\{B,C\}$ are $\mathcal{B}$-admissible.

Now we go into the main part. At first,
\begin{itemize}
\item
$\mathbf{U} = \{[A_0]_\emptyset, [A_1]_\emptyset, [B]_\emptyset, [C]_\emptyset, [A_0]_{\{B,C \}}, [A_1]_{\{B, C\}}\}$.

\item

$\mathbf{V} = \emptyset$.
\end{itemize}

Now let $m=1$ and explore the \textbf{repeat}-block.   Since $A_0 \stackrel{b}{\longrightarrow}_{\empty} \epsilon$ and $d_{\emptyset}(\epsilon) = 0$, and moreover $[A_0]_\emptyset$  is $<_\emptyset$ minimum,   $[A_0]_\emptyset$ is selected from $\mathbf{U}$. $[A_0]_\emptyset$ is deemed to be a $\mathcal{B}_\emptyset$-prime because there is no candidate of decomposition of $[A_0]_\emptyset$. Thus $[A_0]_\emptyset$ is put into $\mathbf{Pr}^{\mathcal{B}}_{\emptyset}$. Then we compute $\mathbf{Rd}^{\mathcal{B}}_{\emptyset}([A_0]_\emptyset)$ via $\textsc{ComputingRd}_R([A_0]_{\emptyset})$, and the result is $\mathbf{Rd}^{\mathcal{B}}_{\emptyset}([A_0]_\emptyset) = \{B,C\}$. After that $[A_0]_\emptyset$ is put into $\mathbf{V}$.  Next, we can select $[B]_\emptyset$ from $\mathbf{U}$. The only candidate for decomposition of $[B]_\emptyset$ is $[A_0]_\emptyset$. One can check that $\textsc{Expand}_{\emptyset}(B,  [A_0]_\emptyset)$ returns \textbf{false}, thus we can affirm that $[B]_\emptyset$ is a $\mathcal{B}_{\emptyset}$-prime, and $\mathbf{Rd}^{\mathcal{B}}_{\emptyset}([B]_\emptyset) = \{B,C\}$ can be computed in the same way. Next, $[C]_\emptyset$,
$[A_0]_{\{B,C\}}$, and $[A_1]_{\{B,C\}}$ are treated successively. $[C]_\emptyset$ and
$[A_0]_{\{B,C\}}$ are $\mathcal{B}$-primes.   $[A_1]_{\{B,C\}}$ is, however, a $\mathcal{B}$-composite because one can check $\textsc{Expand}_{\{B,C\}}(A_1,  [A_0]_{\{B,C\}})$ return \textbf{true}. Now $\mathbf{Dc}^{\mathcal{B}}_{\{B,C\}}([A_1]_{\{B,C\}}) = [A_0]_{\{B,C\}}$.

Now $m =2$.  Because  $A_1 \stackrel{a}{\longrightarrow}_\emptyset A_0$ and $d_{\emptyset}(A_0) = 1$,  we can select $[A_1]_\emptyset$ from $\mathbf{U}$ and define $d_{\emptyset}[A_1] = 2$.
We can check that $[A_1]_{\emptyset}$ is prime too, and $\mathbf{Rd}^{\mathcal{B}}_{\{B,C\}}([A_1]_{\emptyset})=\{B,C\}$.

The date computed in this iteration is summarized as follows:
{\center \small
\begin{tabular}{|c||c|c|c|c|c|c|} \hline
    blocks          & ord & norm & $\mathbf{Pr}$  & $\mathbf{Cm}$  & $\mathbf{Rd}$   & $\mathbf{Dc}$  \\
\hline\hline
$[A_0]_\emptyset$   & 1  & 1   & \checkmark   &              &  $\{B,C\}$   &       \\ \hline
$[A_1]_\emptyset$   & 6  & 2   & \checkmark   &              &  $\{B,C\}$   &     \\ \hline
$[B]_\emptyset$     & 2  & 1   & \checkmark   &              &  $\{B,C\}$   &     \\ \hline
$[C]_\emptyset$     & 3  & 1   & \checkmark   &              &  $\{B,C\}$   &      \\ \hline
$[A_0]_{\{B,C\}}$   & 4  & 1   & \checkmark   &              &  $\{B,C\}$   &      \\ \hline
$[A_1]_{\{B,C\}}$   & 5  & 1   &              &  \checkmark  &              &  $[A_0]_{\{B,C\}}$    \\
\hline
\end{tabular}
}

Finally, $\mathcal{B}$ is assigned to $\mathcal{D}$.

\subsubsection{The 2nd Iteration}
We calculate $\mathbf{Id}^{\mathcal{B}}_{R}$ for every $R \subseteq \mathbf{C}_{\mathrm{G}}$ via $\textsc{ComputingId}(R)$  at first. Once again, $\emptyset$ and $\{B,C\}$ are $\mathcal{B}$-admissible.

When $m=1$, we find $[A_0]_\emptyset, [B]_\emptyset, [C]_\emptyset \in \mathbf{Pr}^{\mathcal{B}}_{\emptyset}$, in which $\mathbf{Rd}^{\mathcal{B}}_{\emptyset}([A_0]_\emptyset) = \mathbf{Rd}^{\mathcal{B}}_{\emptyset}([B]_\emptyset) =\emptyset$ and  $\mathbf{Rd}^{\mathcal{B}}_{\emptyset}([C]_\emptyset) =\{B,C\}$.  Then we find $[A_0]_{\{B,C\}} \in   \mathbf{Pr}^{\mathcal{B}}_{\{B,C\}}$ and $[A_1]_{\{B,C\}} \in \mathbf{Cm}^{\mathcal{B}}_{\{B,C\}}$.
When $m=2$, we find that $[A_1]_\emptyset \in \mathbf{Pr}^{\mathcal{B}}_{\emptyset}$ and $\mathbf{Rd}^{\mathcal{B}}_{\emptyset}([A_1]_\emptyset)=\emptyset$.

The date computed in this iteration is summarized as follows:
{\center \small
\begin{tabular}{|c||c|c|c|c|c|c|} \hline
    blocks          & ord & norm & $\mathbf{Pr}$  & $\mathbf{Cm}$  & $\mathbf{Rd}$   & $\mathbf{Dc}$  \\
\hline\hline
$[A_0]_\emptyset$   & 1  & 1   & \checkmark   &              &  $\emptyset$   &       \\ \hline
$[A_1]_\emptyset$   & 6  & 2   & \checkmark   &              &  $\emptyset$  &     \\ \hline
$[B]_\emptyset$     & 2  & 1   & \checkmark   &              &  $\emptyset$   &     \\ \hline
$[C]_\emptyset$     & 3  & 1   & \checkmark   &              &  $\{B,C\}$   &      \\ \hline
$[A_0]_{\{B,C\}}$   & 4  & 1   & \checkmark   &              &  $\{B,C\}$   &      \\ \hline
$[A_1]_{\{B,C\}}$   & 5  & 1   &              &  \checkmark  &              &  $[A_0]_{\{B,C\}}$    \\
\hline
\end{tabular}
}

We find an interesting fact that the only difference between $\mathcal{B}$ and $\mathcal{D}$ is $\mathbf{Rd}^{\mathcal{B}}_{\emptyset} \subsetneq \mathbf{Rd}^{\mathcal{D}}_{\emptyset}$. Compare this fact with Lemma~\ref{lem:compare_RD}.

Finally, $\mathcal{B}$ is assigned to $\mathcal{D}$.

\subsubsection{The 3rd Iteration}

In this iteration,  we can obtain that $\mathcal{B} = \mathcal{D}$. Therefore the algorithm stops here. We can draw the conclusion that $\widehat{\mathcal{B}} = \mathcal{D}$. In other words, ${\simeq_{R}} = {\stackrel{\mathcal{D}}{=}_R}$ for every $R \subseteq \mathbf{C}_{\mathrm{G}}$.

\subsubsection{Conclusion}
We confirm that $A_0 \not\simeq A_1$ by showing $\mathtt{dcmp}^{\widehat{\mathcal{B}}}_{\emptyset}(A_0) \neq \mathtt{dcmp}^{\widehat{\mathcal{B}}}_{\emptyset}(A_1)$.  In fact,  $\mathtt{dcmp}^{\widehat{\mathcal{B}}}_{\emptyset}(A_0) = [A_0]_{\emptyset}$ and $\mathtt{dcmp}^{\widehat{\mathcal{B}}}_{\emptyset}(A_1) = [A_1]_{\emptyset}$.

We can show $A_0.C \simeq A_1.C$.  Using the $\widehat{\mathcal{B}}$-reduction rules defined in Section~\ref{subsec:decomposition_base}, we have
\begin{itemize}
\item
 $C \stackrel{\widehat{\mathcal{B}}}{\rightarrow}_{\emptyset} [C]_{\emptyset} $;

\item
$\mathbf{Rd}^{\widehat{\mathcal{B}}}_{\emptyset}([C]_{\emptyset}) = \{B,C\}$;

\item
 $A_0 \stackrel{\widehat{\mathcal{B}}}{\rightarrow}_{\{B,C\}} [A_0]_{\{B,C\}} $;

\item
 $A_1 \stackrel{\widehat{\mathcal{B}}}{\rightarrow}_{\{B,C\}} [A_1]_{\{B,C\}} \stackrel{\widehat{\mathcal{B}}}{\rightarrow}_{\{B,C\}} [A_0]_{\{B,C\}}$.
\end{itemize}
Therefore,  we have
\[
\mathtt{dcmp}_{\emptyset}(A_0.C) = \mathtt{dcmp}_{\emptyset}(A_1.C) = [A_0]_{\{B,C\}}[C]_{\emptyset},
\]
hence $A_0.C \simeq A_1.C$.

Actually, we can show
\begin{enumerate}
\item
For every $i_1,j_1, \ldots, i_t,j_t \in \{0,1\}$, we have
\[
A_{i_t}.\ldots.A_{i_1}.C \simeq A_{j_t}.\ldots.A_{j_1}.C.
\]

\item
For every $i_1,j_1, \ldots, i_t,j_t \in \{0,1\}$,
\[
A_{i_t}.\ldots.A_{i_1} \simeq A_{j_t}.\ldots.A_{j_1}
\]
if and only if $i_1 = j_1, \ldots, i_t = j_t$.

\end{enumerate}

\section{Remark}
The algorithm described in Section~\ref{sec:naive-algorithm} can be further improved. For example, in the \textbf{repeat}-block at line~\ref{line:repeat}, although $m$ can be exponentially large, there is no need to enumerate every $m$. We can compute the next candidate of $m$ based on the right-hand-sides of $\Delta$. Thus only polynomial number of candidates of $m$ is available.  In addition, we notice that, although the length of the decomposition of $[X]_R$ can be exponentially large, the technique of string compression can be used such that the representation and manipulation of strings can be implemented in polynomial time. This is done in all the previous works on polynomial-time algorithms for checking bisimilarity on realtime BPA. Ultimately,
the number of ground constants is essentially the only factor of the exponential time.  Therefore, we claim that branching bisimilarity on normed BPA is in fact fixed parameter tractable.

\section*{Acknowledgment}
This research is supported by the National Nature Science Foundation of China (61472240, 91318301, 61261130589).



\bibliographystyle{IEEEtran}
%

\bibliographystyle{plain}
\bibliography{pa}

\newpage
\appendices

\section{Proofs in Section~\ref{sec:Relativized_bisimilarity}}

\subsection{Proof of the Computation Lemma}

\subsubsection{Proof of Lemma~\ref{lem:computation_lemma}}

We present a complete proof of Lemma~\ref{lem:computation_lemma} here which depends only on Definition~\ref{def:beq}. Though  Lemma~\ref{lem:computation_lemma} is well-known, the proof here has some subtleties.
Importantly, the proof framework will be used to show Lemma~\ref{lem:computation_lemma_R}.

Let
\[
\alpha = \alpha_0 \stackrel{\tau}{\longrightarrow} \alpha_1 \stackrel{\tau}{\longrightarrow} \alpha_2 \stackrel{\tau}{\longrightarrow} \ldots \stackrel{\tau}{\longrightarrow} \alpha_k  \simeq \alpha.
\]
We show that $\alpha_i \simeq \alpha_j$ for every $0 \leq i, j \leq k$.  To this end,
let $\mathcal{S} \stackrel{\mathrm{def}}{=} \{(\alpha_i, \alpha_j) \;|\;  0 \leq i, j \leq k \}$, and
construct the equivalence relation ${\asymp} \stackrel{\mathrm{def}}{=} ( \mathcal{S}  \cup {\simeq} )^{*}$. We emphasize that $\mathcal{S}$ can be viewed as a single `equivalence class',  because $\mathcal{S}$ is both symmetric and transitive, and connective.
We confirm that ${\asymp}$ is a bi\-simulation.  The crux is to show the following {\em key property}:  For every $i,j \in \{0,1, \ldots, k\}$,
\begin{enumerate}
\item
If $\alpha_i \stackrel{a}{\longrightarrow} \gamma$, then
    $\alpha_j \stackrel{\asymp} \Longrightarrow \cdot \stackrel{a}{\longrightarrow} \delta$ for some $\delta$ such that  $\gamma \simeq \delta$.

\item
If $\alpha_i \stackrel{\not\asymp}{\longrightarrow} \gamma$, then
    $\alpha_j \stackrel{\asymp} \Longrightarrow \cdot \stackrel{\not\asymp}{\longrightarrow} \delta$ for some $\delta$ such that  $\gamma \simeq \delta$.
\end{enumerate}

To show this property, we study the following  two cases:
\begin{itemize}
\item
$i \geq  j$.  In this case, we have $\alpha_j \stackrel{\asymp} \Longrightarrow \alpha_i$. By letting $\delta = \gamma$, we get the key property.

\item
$i < j$.  In this case, we will make use of the fact $\alpha \simeq \alpha_k$.  Consider the transition $\alpha \stackrel{\tau}{\longrightarrow} \alpha_1$. Now either $\alpha_1 \simeq \alpha_k$, or $\alpha_k \stackrel{\simeq}{\Longrightarrow} \cdot \stackrel{\tau}{\longrightarrow} \alpha_1'$ such that $\alpha_1 \simeq \alpha_1'$. In view of $(\alpha, \alpha_1) \in \mathcal{S}$, we conclude that in either case, $\alpha_k \stackrel{\asymp}{\Longrightarrow} \alpha_1'$ for some $\alpha_1'$ such that $\alpha_1 \simeq \alpha_1'$.  By repeatedly applying the above argument. We can show that $\alpha_k \stackrel{\asymp}{\Longrightarrow} \alpha_i'$ for some $\alpha_i'$ such that $\alpha_i \simeq \alpha_i'$. Since $\alpha_j \stackrel{\asymp}{\Longrightarrow} \alpha_k$, we now have $\alpha_j \stackrel{\asymp}{\Longrightarrow} \alpha_i'$ such that $\alpha_i \simeq \alpha_i'$.  Now it is a routine work to justify the key property. For example, suppose that $\alpha_i \stackrel{a}{\longrightarrow} \gamma$, then
we have $\alpha_j \stackrel{\asymp} \Longrightarrow \alpha_i' \stackrel{\asymp} \Longrightarrow \cdot \stackrel{a}{\longrightarrow} \delta$ for some $\delta$ such that  $\gamma \simeq \delta$.
\end{itemize}

\begin{remark}
We can show the bisimulation property of $\asymp$ by repeatedly using the {\em key property} and the bisimulation property of $\simeq$. Be very careful that it would be a mistake if the `key property' was modified slightly as follows: For every $i,j \in \{0,1, \ldots, k\}$,
\begin{enumerate}
\item
If $\alpha_i \stackrel{a}{\longrightarrow} \gamma$, then
    $\alpha_j \stackrel{\asymp} \Longrightarrow \cdot \stackrel{a}{\longrightarrow} \delta$ for some $\delta$ such that  $\gamma \asymp \delta$.

\item
If $\alpha_i \stackrel{\not\asymp}{\longrightarrow} \gamma$, then
    $\alpha_j \stackrel{\asymp} \Longrightarrow \cdot \stackrel{\not\asymp}{\longrightarrow} \delta$ for some $\delta$ such that  $\gamma \asymp \delta$.
\end{enumerate}
This mistake is essentially the same as the well-known mistake of `weak bisimulation up-to weak bisimilarity'.  To get a better understanding of the mistake, readers are referred to Chapter~5 of~\cite{Milner1989}, especially Section~5.7.
\end{remark}

\subsubsection{Proof of Lemma~\ref{lem:computation_lemma_R}}

This proof is an adaptation of the proof of Lemma~\ref{lem:computation_lemma}.

Let $R \subseteq \mathbf{C}_{\mathrm{G}}$ and let $\alpha$ be in $R$-nf. Suppose that
\[
\alpha = \alpha_0 \stackrel{\tau}{\longrightarrow}_R \alpha_1 \stackrel{\tau}{\longrightarrow}_R \alpha_2 \stackrel{\tau}{\longrightarrow}_R \ldots \stackrel{\tau}{\longrightarrow}_R \alpha_k  \simeq_R \alpha.
\]
We show that $\alpha_i \simeq_R \alpha_j$ for every $0 \leq i, j \leq k$.  To this end,
let $\mathcal{S} \stackrel{\mathrm{def}}{=} \{(\alpha_i, \alpha_j) \;|\;  0 \leq i, j \leq k \}$, and
construct ${\asymp} \stackrel{\mathrm{def}}{=} ( \mathcal{S}  \cup {\simeq_R} )^{*}$.  We confirm that ${\asymp}$ is an $R$-bisimulation (via Proposition~\ref{prop:R_bisimilarity}).  First of all, we point out the following basic facts:
\begin{itemize}
\item
$\asymp$ is an equivalence relation.

\item
$\mathcal{S}$ is symmetric, transitive, and connective.

\item
${=_R} \subseteq {\asymp}$. (Because ${=_R} \subseteq {\simeq_R}$ and ${\simeq_R} \subseteq {\asymp}$. )

\item
If $\alpha_i \Longrightarrow \epsilon$ for some $0 \leq i \leq k$, then $\alpha_j \Longrightarrow \epsilon$ for every $0 \leq j \leq k$.

\item
All $\alpha_i$'s are in $R$-nf for $0 \leq i \leq k$.

\end{itemize}
The crux of the proof  is to show the following {\em key property}: For every $i,j \in \{0,1, \ldots, k\}$,
\begin{enumerate}
\item
If $\alpha_i \stackrel{a}{\longrightarrow}_R \gamma$, then
    $\alpha_j \stackrel{\asymp} \Longrightarrow_R \cdot \stackrel{a}{\longrightarrow}_R \delta$ for some $\delta$ such that  $\gamma \simeq_R \delta$.

\item
If $\alpha_i \stackrel{\not\asymp}{\longrightarrow}_R \gamma$, then
    $\alpha_j \stackrel{\asymp} \Longrightarrow_R \cdot \stackrel{\not\asymp}{\longrightarrow}_R \delta$ for some $\delta$ such that  $\gamma \simeq_R \delta$.
\end{enumerate}

To show this property, we study the following  two cases:
\begin{itemize}
\item
$i \geq  j$.  In this case, we have $\alpha_j \stackrel{\asymp} \Longrightarrow_R \alpha_i$. By letting $\delta = \gamma$, we get the key property.

\item
$i < j$.  In this case, we will make use of the fact $\alpha \simeq_R \alpha_k$.
Consider the transition $\alpha \stackrel{\tau}{\longrightarrow}_R \alpha_1$.
Now either $\alpha_1 \simeq_R \alpha_k$, or $\alpha_k \stackrel{\simeq}{\Longrightarrow}_R \cdot \stackrel{\tau}{\longrightarrow}_R \alpha_1'$ such that $\alpha_1 \simeq_R \alpha_1'$.
In view of $(\alpha, \alpha_1) \in \mathcal{S}$, we conclude that in either case, $\alpha_k \stackrel{\asymp}{\Longrightarrow}_R \alpha_1'$ for some $\alpha_1'$ such that $\alpha_1 \simeq_R \alpha_1'$.
By repeatedly applying the above argument. We can show that $\alpha_k \stackrel{\asymp}{\Longrightarrow}_R \alpha_i'$ for some $\alpha_i'$ such that $\alpha_i \simeq_R \alpha_i'$.
Since $\alpha_j \stackrel{\asymp}{\Longrightarrow}_R \alpha_k$, we now have $\alpha_j \stackrel{\asymp}{\Longrightarrow}_R \alpha_i'$ such that $\alpha_i \simeq_R \alpha_i'$.  Now it is a routine work to justify the key property. For example, suppose that $\alpha_i \stackrel{a}{\longrightarrow}_R \gamma$, then
we have $\alpha_j \stackrel{\asymp} \Longrightarrow_R \alpha_i' \stackrel{\asymp} \Longrightarrow_R \cdot \stackrel{a}{\longrightarrow}_R \delta$ for some $\delta$ such that  $\gamma \simeq_R \delta$.
\end{itemize}



\subsection{Proof of Proposition~\ref{prop:R_monotone}}
Suppose $R_1 \subseteq R_2$. We show that
\[
{\asymp} \stackrel{\mathrm{def}}{=} (\simeq_{R_1} \cup  \simeq_{R_2})^{*}
\]
is an $R_2$-bisimulation.

We emphasize the following basic facts:
\begin{itemize}
\item
${\simeq_{R_1}} \subseteq {\asymp}$ and ${\simeq_{R_2}} \subseteq {\asymp}$.

\item
${=_{R_2}} \subseteq {\asymp}$.

\item
$\simeq_{R_1}$, $\simeq_{R_2}$, ${\asymp}$ are all equivalence relations.
\end{itemize}

In order to show that ${\asymp}$ is an $R_2$-bisimulation, we design the following property~\textbf{I} and \textbf{II}:
\begin{itemize}
\item[\textbf{I}.]
Suppose there are $\alpha$ and $\beta$ satisfying $\alpha \simeq_{R_1} \beta$, then
\begin{enumerate}
\item
if  $\alpha\Longrightarrow \epsilon$,  then $\beta \Longrightarrow \epsilon$;

\item
if $\alpha \stackrel{\not\asymp}{\longrightarrow} \alpha'$, then
$\beta_{R_2}  \stackrel{\asymp}{\Longrightarrow}_{R_2} \cdot \stackrel{\not\asymp}{\longrightarrow}_{R_2} \beta'$ for some $\beta'$ such that $\alpha' \simeq_{R_1} \cdot \simeq_{R_2} \beta'$;

\item
if $\alpha \stackrel{a}{\longrightarrow} \alpha'$, then
$\beta_{R_2}  \stackrel{\asymp}{\Longrightarrow}_{R_2} \cdot \stackrel{a}{\longrightarrow}_{R_2} \beta'$ for some $\beta'$ such that $\alpha' \simeq_{R_1} \cdot \simeq_{R_2} \beta'$.
\end{enumerate}

\item[\textbf{II}.]
Suppose there are $\alpha$ and $\beta$ satisfying $\alpha \simeq_{R_2} \beta$, then
\begin{enumerate}
\item
if  $\alpha\Longrightarrow \epsilon$,  then $\beta \Longrightarrow \epsilon$;

\item
if $\alpha \stackrel{\not\asymp}{\longrightarrow} \alpha'$, then
$\beta_{R_2}  \stackrel{\asymp}{\Longrightarrow}_{R_2} \cdot \stackrel{\not\asymp}{\longrightarrow}_{R_2} \beta'$ for some $\beta'$ such that $\alpha' \simeq_{R_2} \beta'$;

\item
if $\alpha \stackrel{a}{\longrightarrow} \alpha'$, then
$\beta_{R_2}  \stackrel{\asymp}{\Longrightarrow}_{R_2} \cdot \stackrel{a}{\longrightarrow}_{R_2} \beta'$ for some $\beta'$ such that $\alpha' \simeq_{R_2} \beta'$.
\end{enumerate}
\end{itemize}
Note that, if $\alpha \asymp \beta$, then we must have
\[
\alpha \simeq_{R_{i_1}} \cdot \simeq_{R_{i_2}} \cdot \ldots \cdot \simeq_{R_{i_n}}  \beta
\]
for some $n \in \mathbb{N}$ and $i_k \in \{1,2\}$ for every $1 \leq k \leq n$.
Therefore, by repeatedly using the  property~\textbf{I} and \textbf{II}, we can obtain that ${\asymp}$ is a bisimulation.

We can observe that the property~\textbf{II} is an direct inference of the bisimulation property of $\simeq_{R_2}$.  Thus it suffices to  prove property~\textbf{I}.

Condition~1 (i.e.~ground preservation) is trivial. Other two conditions have the same structures and Condition~3 cannot be more difficult than Condition~2. Thus we choose to prove Condition~2.

Suppose there are $\alpha$ and $\beta$ satisfying $\alpha \simeq_{R_1} \beta$. According to Definition~\ref{def:R_beq}, we have:
\begin{itemize}
\item
if $\alpha \stackrel{\not\simeq_{R_1}}{\longrightarrow} \alpha'$, then
$\beta_{R_1}  \stackrel{\simeq_{R_1}}{\Longrightarrow}_{R_1} \cdot \stackrel{\not\simeq_{R_1}}{\longrightarrow}_{R_1} \beta'$ for some $\beta'$ such that $\alpha' \simeq_{R_1} \beta'$.
\end{itemize}

Assume $\alpha \stackrel{\not\asymp}{\longrightarrow} \alpha'$. Because ${\simeq_{R_1}} \subseteq {\asymp}$, we must have $\alpha \stackrel{\not\simeq_{R_1}}{\longrightarrow} \alpha'$.
We can find $\beta'$ such that $\beta_{R_1}  \stackrel{\simeq_{R_1}}{\Longrightarrow}_{R_1} \cdot \stackrel{\not\simeq_{R_1}}{\longrightarrow}_{R_1} \beta'$ and $\alpha' \simeq_{R_1} \beta'$.
In view of Lemma~\ref{lem:R_transition} and Lemma~\ref{lem:char_R_arrow}, we have
\[
\beta =_{R_1} \cdot \stackrel{\simeq_{R_1}}{\Longrightarrow} \cdot \stackrel{\not\simeq_{R_1}}{\longrightarrow} \cdot  =_{R_1} \beta'
\]
for some $\beta'$ such that $\alpha' \simeq_{R_1} \beta'$.
Since $R_1 \subseteq R_2$, we have ${=_{R_1}} \subseteq  {=_{R_2}}$.
Also note that ${\simeq_{R_1}} \subseteq {\asymp}$.
Thus
\[
\beta =_{R_2} \cdot \stackrel{\asymp}{\Longrightarrow}  \beta'' \stackrel{\tau}{\longrightarrow} \cdot =_{R_2} \beta'
\]
for some $\beta', \beta''$ such that $\alpha' \simeq_{R_1} \beta'$.
We can ensure that $\beta'' \not\asymp \beta'$, because $\beta'' \asymp \beta \asymp \alpha \not\asymp \alpha' \asymp \beta'$.
Therefore,
\[
\beta =_{R_2} \cdot \stackrel{\asymp}{\Longrightarrow} \cdot  \stackrel{\not\asymp}{\longrightarrow} \cdot  =_{R_2} \beta'
\]
for some $\beta'$ such that $\alpha' \simeq_{R_1} \beta'$. Applying  Lemma~\ref{lem:R_transition} and Lemma~\ref{lem:char_R_arrow} repeatedly, and remembering $=_{R_2} \subseteq \simeq_{R_2}$, we have
\[
\beta_{R_2}  \stackrel{\asymp}{\Longrightarrow}_{R_2} \cdot \stackrel{\not\asymp}{\longrightarrow}_{R_2} \beta'_{R_2}
\]
for some $\beta'_{R_2}$ such that $\alpha' \simeq_{R_1} \cdot  \simeq_{R_2} \beta'_{R_2}$.





\begin{remark}
We make a mistake in the previous version, because we take the following slightly different variant of property~\textbf{I} and \textbf{II}:
\begin{itemize}
\item[\textbf{I'}.]
Suppose there are $\alpha$ and $\beta$ satisfying $\alpha \simeq_{R_1} \beta$, then
\begin{enumerate}
\item
if  $\alpha\Longrightarrow \epsilon$,  then $\beta \Longrightarrow \epsilon$;

\item
if $\alpha \stackrel{\not\asymp}{\longrightarrow} \alpha'$, then
$\beta_{R_2}  \stackrel{\asymp}{\Longrightarrow}_{R_2} \cdot \stackrel{\not\asymp}{\longrightarrow}_{R_2} \beta'$ for some $\beta'$ such that $\alpha' \asymp \beta'$;

\item
if $\alpha \stackrel{a}{\longrightarrow} \alpha'$, then
$\beta_{R_2}  \stackrel{\asymp}{\Longrightarrow}_{R_2} \cdot \stackrel{a}{\longrightarrow}_{R_2} \beta'$ for some $\beta'$ such that $\alpha' \asymp \beta'$.
\end{enumerate}

\item[\textbf{II'}.]
Suppose there are $\alpha$ and $\beta$ satisfying $\alpha \simeq_{R_2} \beta$, then
\begin{enumerate}
\item
if  $\alpha\Longrightarrow \epsilon$,  then $\beta \Longrightarrow \epsilon$;

\item
if $\alpha \stackrel{\not\asymp}{\longrightarrow} \alpha'$, then
$\beta_{R_2}  \stackrel{\asymp}{\Longrightarrow}_{R_2} \cdot \stackrel{\not\asymp}{\longrightarrow}_{R_2} \beta'$ for some $\beta'$ such that $\alpha' \asymp \beta'$;

\item
if $\alpha \stackrel{a}{\longrightarrow} \alpha'$, then
$\beta_{R_2}  \stackrel{\asymp}{\Longrightarrow}_{R_2} \cdot \stackrel{a}{\longrightarrow}_{R_2} \beta'$ for some $\beta'$ such that $\alpha' \asymp \beta'$.
\end{enumerate}
\end{itemize}
It will run into trouble when using property~\textbf{I'} and \textbf{II'} to prove bisimulation property of $\asymp$. The reason is the same as in the situation of `weak bisimulation up to weak bisimilarity'~\cite{Milner1989}.  To get a better understanding of this mistake, readers are referred to Chapter~5 of~\cite{Milner1989}, especially Section~5.7.
\end{remark}

\subsection{Proof of Proposition~\ref{prop:R_bis_vs_IDR}}

Because $R \subseteq \mathsf{Id}_R$, by Proposition~\ref{prop:R_monotone}, we know that ${\simeq_{R}} \subseteq {\simeq_{\mathsf{Id}_R}}$. Thus it suffices to show  ${\simeq_{\mathsf{Id}_R}}  \subseteq {\simeq_{R}}$.
Let $S = \mathsf{Id}_R$, it suffices to show that $\simeq_S$ is an $R$-bisimulation.

Clearly, $\simeq_S$ is an equivalence relation, and ${=_R} \subseteq {\simeq_S}$. Thus it suffices to prove the following property: Whenever $\alpha \simeq_{S} \beta$,
\begin{enumerate}
\item
if  $\alpha\Longrightarrow \epsilon$,  then $\beta \Longrightarrow \epsilon$;

\item
if $\alpha \stackrel{\not\simeq_{S}}{\longrightarrow} \alpha'$, then
$\beta_{R}  \stackrel{\simeq_{S}}{\Longrightarrow}_{R} \cdot \stackrel{\not\simeq_{S}}{\longrightarrow}_{R} \beta'$ for some $\beta'$ such that $\alpha' \simeq_{S} \beta'$;

\item
if $\alpha \stackrel{a}{\longrightarrow} \alpha'$, then
$\beta_{R}  \stackrel{\simeq_{S}}{\Longrightarrow}_{R} \cdot \stackrel{a}{\longrightarrow}_{R} \beta'$ for some $\beta'$ such that $\alpha' \simeq_{S} \beta'$.
\end{enumerate}
Condition~1 is trivial. Other two conditions have the same structures. As usual we choose to prove Condition~2.

Suppose that $\alpha \simeq_{S} \beta$ and $\alpha \stackrel{\not\simeq_{S}}{\longrightarrow} \alpha'$. By Definition~\ref{def:R_beq}, we have:
\begin{quote}
$\beta_{S}  \stackrel{\simeq_{S}}{\Longrightarrow}_{S} \cdot \stackrel{\not\simeq_{S}}{\longrightarrow}_{S} \beta'$ for some $\beta'$ such that $\alpha' \simeq_{S} \beta'$.
\end{quote}

Now there are two cases:
\begin{itemize}
\item
$\beta_{S} = \epsilon$.  In this case, we have $\beta \in S^{*} = (\mathsf{Id}_R)^{*}$ thus $\beta \simeq_R \epsilon$ by Lemma~\ref{lem:Id_def}. Since ${\simeq_{R}} \subseteq {\simeq_{S}}$, we have $\beta_{S} \simeq_{S} \epsilon$.
Consider any $\beta''$ such that $\beta_{S}  \stackrel{\simeq_{S}}{\longrightarrow}_{S} \beta''$. We have the following properties.
\begin{itemize}
\item
$\beta'' \Longrightarrow \epsilon$ (because $\beta'' \simeq_{S} \epsilon$).

\item
There exists $X \in S$ such that $X \stackrel{\simeq_{S}}{\longrightarrow} \gamma$ and $\gamma_S = \beta''$. (By Lemma~\ref{lem:char_R_arrow})
\end{itemize}
Now we have $\epsilon \simeq_R X \stackrel{\tau}{\longrightarrow} \gamma \Longrightarrow \epsilon$ and by Lemma~\ref{lem:R_transition},
\[
\epsilon \simeq_R X_R \stackrel{\tau}{\longrightarrow}_R \gamma_R \Longrightarrow_R \epsilon.
\]
According to Computation Lemma (Lemma~\ref{lem:computation_lemma_R}), $\gamma \simeq_R \epsilon$, and thus by Lemma~\ref{lem:Id_def}, $\gamma\in (\mathsf{Id}_R)^{*} = S^{*}$. Therefore $\beta'' = \gamma_S = \epsilon$. This crucial fact leads to the following assertion:
\begin{quote}
Whenever $\beta_{S}  \stackrel{\simeq_{S}}{\Longrightarrow}_{S} \beta'' \stackrel{\not\simeq_{S}}{\longrightarrow}_{S} \beta'$, $\beta''$ must be $\epsilon$.
\end{quote}
Therefore, we confirm that
\begin{quote}
$\beta_{S} \stackrel{\not\simeq_{S}}{\longrightarrow}_{S} \beta'$ for some $\beta'$ such that $\alpha' \simeq_{S} \beta'$.
\end{quote}

Now according to Lemma~\ref{lem:char_R_arrow},  there exists $X \in S$ such that $X \stackrel{\not\simeq_{S}}{\longrightarrow} \widehat{\beta}'$ for some $ \widehat{\beta}'$ and ${\widehat{\beta}'}_{S} = \beta'$.
Knowing ${\simeq_R} \subseteq {\simeq_{S}}$, we have $X \stackrel{\not\simeq_R}{\longrightarrow} \widehat{\beta}'$.
Because $X \in S = \mathsf{Id}_R$, we have $X \simeq_R \epsilon$. Now according to Definition~\ref{def:R_beq},
\[
\epsilon \stackrel{\simeq_R }{\Longrightarrow}_{R} \cdot \stackrel{\not\simeq_R }{\longrightarrow}_{R} \gamma' \simeq_R \widehat{\beta}' \simeq_R \beta'
\]
for some $\gamma'$, which implies that
\[
\epsilon \stackrel{\simeq_{S} }{\Longrightarrow}_{R} \gamma'' \stackrel{\tau}{\longrightarrow}_{R} \gamma' \simeq_R \alpha'
\]
for some $\gamma',\gamma''$.
Finally we can observe $\gamma'' \not\simeq_{S} \gamma'$, for  $\gamma'' \simeq_{S} \epsilon  \simeq_{S} \beta \simeq_{S} \alpha \not\simeq_{S} \alpha' \simeq_{S} \gamma'$.
Above all, we find such $\gamma$ that $\beta_{R}  \stackrel{\simeq_{S} }{\Longrightarrow}_{R}  \cdot \stackrel{\not\simeq_{S}}{\longrightarrow}_R \gamma'$ and  $\alpha' \simeq_{S} \gamma'$.

\item
$\beta_{S} \neq \epsilon$.  In this case, there are several subcases,  depending on the path $\beta_{S}  \stackrel{\simeq_{S}}{\Longrightarrow}_{S}  \cdot \stackrel{\not\simeq_{S}}{\longrightarrow}_{S} \beta'$:
\begin{itemize}
\item
$\beta_{S} \stackrel{\not\simeq_{S}}{\longrightarrow}_{S} \beta'$.
We can show by applying Lemma~\ref{lem:char_R_arrow} that
$\beta \stackrel{\tau}{\longrightarrow} \widehat{\beta}'$ for some $\widehat{\beta}'$ such that $\alpha' \simeq_{S} \widehat{\beta}'$.  %
By Lemma~\ref{lem:R_transition}, we have  $\beta_{R}  \stackrel{\tau}{\longrightarrow}_R \widehat{\beta}'_{R}$. Knowing the fact that $\widehat{\beta}'_{R} \simeq_{R} \widehat{\beta}'$ and ${\simeq_{R}} \subseteq {\simeq_{S}}$, we have $\widehat{\beta}'_{R}\simeq_{S} \widehat{\beta}'$, and thus $\alpha' \simeq_{S} \widehat{\beta}'_{R}$. Finally it is a routine work to observe $\beta_{R}  \not\simeq_{S} \widehat{\beta}'_{R}$.
In summary, we have $\beta_{R}  \stackrel{\not\simeq_{S}}{\longrightarrow}_R \widehat{\beta}'_{R}$ and  $\alpha' \simeq_{S} \widehat{\beta}'_{R}$.

\item
$\beta_{S}  \stackrel{\simeq_{S}}{\longrightarrow}_{S} \eta \stackrel{\simeq_{S}}{\Longrightarrow}_{S}  \cdot \stackrel{\not\simeq_{S}}{\longrightarrow}_{S} \beta'$.
We can show by applying Lemma~\ref{lem:char_R_arrow} that
$\beta  \stackrel{\simeq_{S}}{\longrightarrow} \widehat{\eta}$ for some $\widehat{\eta}$ such that $\widehat{\eta}_S = \eta$.  In view of ${\simeq_R} \subseteq {\simeq_S}$, and by Lemma~\ref{lem:R_transition}, we have $\beta_R  \stackrel{\simeq_{S}}{\longrightarrow}_R \widehat{\eta}_R$ and $\alpha \simeq_S \widehat{\eta}_R$ (because $\alpha \simeq_S \beta \simeq_S \eta \simeq_S  \widehat{\eta} \simeq_R \widehat{\eta}_R$).   Remember the fact $\widehat{\eta}_S = \eta$, we can now use induction to confirm that
\[
\widehat{\eta}_{R}  \stackrel{\simeq_{S}}{\Longrightarrow}_{R} \cdot \stackrel{\not\simeq_{S}}{\longrightarrow}_{R} \beta'
\]
for some $\beta'$ such that $\alpha' \simeq_{S} \beta'$.  Put them together, we get
\[
\beta_R  \stackrel{\simeq_{S}}{\longrightarrow}_R \widehat{\eta}_R
 \stackrel{\simeq_{S}}{\Longrightarrow}_{R} \cdot \stackrel{\not\simeq_{S}}{\longrightarrow}_{R} \beta'
\]
for some $\beta'$ such that $\alpha' \simeq_{S} \beta'$.
\end{itemize}

\end{itemize}

\subsection{Proof of Theorem~\ref{thm:relative_bis} and Theorem~\ref{thm:relative_bis_str}}

Since Theorem~\ref{thm:relative_bis} is a special case of Theorem~\ref{thm:relative_bis_str}.  We only prove Theorem~\ref{thm:relative_bis_str}.
The proof  is divided into the following two lemmas: Lemma~\ref{lem:relative_bis_onlyif} and Lemma~\ref{lem:relative_bis_if}.

\begin{lemma}\label{lem:relative_bis_onlyif}
If $S = \mathsf{Rd}_{R}(\gamma)$,  then $\alpha \simeq_{S} \beta$ implies $\alpha \gamma \simeq_{R}  \beta\gamma$.
\end{lemma}

\begin{proof}
We can assume that $\gamma \not\in R^{*}$.  If not, we will have $S =  \mathsf{Id}_{R}$ and thus according to Proposition~\ref{prop:R_monotone}, ${\simeq_{S}} =  {\simeq_{R}}$. The result of this lemma holds accordingly.

Let ${\bumpeq}$ be the relation $\{(\alpha\gamma, \beta\gamma) \,|\, \alpha \simeq_{S} \beta\}$.
Define the relation
\[
{\asymp}  \stackrel{\mathrm{def}}{=} ({\bumpeq} \cup {\simeq_{R}})^{*}.
\]
We show that $\asymp$ is an $R$-bisimulation.

We point out the following basic facts:
\begin{itemize}
\item
${\bumpeq} \subseteq {\asymp}$,  ${=_R} \subseteq {\simeq_{R}} \subseteq {\asymp}$, and ${\bumpeq} \circ  {\simeq_R} \subseteq {\asymp}$.

\item
$\mathcal{I} \subseteq {\simeq_R}$, in which  $\mathcal{I} = \{(\zeta, \zeta) \;|\;  \zeta \in \mathbf{C}^{*} \}$, which is the identical relation on $\mathbf{C}^{*}$.

\item
$\bumpeq$ is symmetric and transitive. 

\item
$\simeq_{R}$, ${\asymp}$ are all equivalence relations.
\end{itemize}

Now consider an arbitrary pair $(\zeta, \eta)$ such that  $\zeta \asymp \eta$.
We will prove
\begin{enumerate}
\item
if  $\zeta \Longrightarrow \epsilon$,  then $\eta \Longrightarrow \epsilon$; (trivial)

\item
if $\zeta \stackrel{\not\asymp}{\longrightarrow} \zeta'$, then
$\eta_R  \stackrel{\asymp}{\Longrightarrow}_R \cdot \stackrel{\not\asymp}{\longrightarrow}_R \eta'$ for some $\eta'$ such that $\zeta' \asymp \eta'$;

\item
if $\zeta \stackrel{a}{\longrightarrow} \zeta'$, then
$\eta_R  \stackrel{\asymp}{\Longrightarrow}_R \cdot \stackrel{a}{\longrightarrow}_R \eta'$ for some $\eta'$ such that $\zeta' \asymp \eta'$.
\end{enumerate}

As usual (similar to the proof of Proposition~\ref{prop:R_monotone}), we show the following property~\textbf{I} and \textbf{II}:
\begin{itemize}
\item[\textbf{I}.]
Suppose there are $\zeta$ and $\eta$ satisfying $\zeta \bumpeq \eta$, then
\begin{enumerate}
\item
if  $\zeta \Longrightarrow \epsilon$,  then $\eta \Longrightarrow \epsilon$;

\item
if $\zeta \stackrel{\not\asymp}{\longrightarrow} \zeta'$, then
$\eta_{R}  \stackrel{\asymp}{\Longrightarrow}_{R} \cdot \stackrel{\not\asymp}{\longrightarrow}_{R} \eta'$ for some $\eta'$ such that $(\zeta', \eta') \in  (\bumpeq \cdot \simeq_{R}) \cup  \mathcal{I}$;

\item
if $\zeta \stackrel{a}{\longrightarrow} \zeta'$, then
$\eta_{R}  \stackrel{\asymp}{\Longrightarrow}_{R} \cdot \stackrel{a}{\longrightarrow}_{R} \eta'$ for some $\eta'$ such that $(\zeta', \eta') \in  (\bumpeq \cdot \simeq_{R}) \cup  \mathcal{I}$.
\end{enumerate}

\item[\textbf{II}.]
Suppose there are $\zeta$ and $\eta$ satisfying $\zeta \simeq_{R} \eta$, then
\begin{enumerate}
\item
if  $\zeta \Longrightarrow \epsilon$,  then $\eta \Longrightarrow \epsilon$;

\item
if $\zeta \stackrel{\not\asymp}{\longrightarrow} \zeta'$, then
$\eta_{R}  \stackrel{\asymp}{\Longrightarrow}_{R} \cdot \stackrel{\not\asymp}{\longrightarrow}_{R} \eta'$ for some $\eta'$ such that $\zeta' \simeq_{R} \eta'$;

\item
if $\zeta \stackrel{a}{\longrightarrow} \zeta'$, then
$\eta_{R}  \stackrel{\asymp}{\Longrightarrow}_{R} \cdot \stackrel{a}{\longrightarrow}_{R} \eta'$ for some $\eta'$ such that $\zeta' \simeq_{R} \eta'$.
\end{enumerate}
\end{itemize}

Now assume that $\zeta \asymp \eta$, we must have $\zeta \simeq_{R} \cdot \bumpeq \cdot \simeq_{R} \cdot \ldots \cdot \bumpeq \cdot \simeq_{R} \beta$. (Think why. ) Thus we can show that  $\asymp$ is an $R$-bisimulation by applying  property~\textbf{I} and \textbf{II} finitely many times.

Since property~\textbf{II} is trivial,  it suffices to  prove property~\textbf{I}.

If $(\zeta, \eta) \in {\bumpeq}$. Now we must have $\zeta = \alpha\gamma$ and $\eta =\beta \gamma$ such that $\alpha  \simeq_{S} \beta$.
There are two cases:


\begin{enumerate}
\item
$\alpha \neq \epsilon$. In this case  $\zeta \stackrel{\ell}{\longrightarrow} \zeta'$ is induced by $\alpha\gamma \stackrel{\ell}{\longrightarrow} \alpha'\gamma$.
\begin{itemize}

\item
If  $\ell = \tau$ and $\alpha\gamma \not\asymp \alpha'\gamma$.  In this case, we must have  $\alpha'  \not\simeq_{S} \alpha$. According to the fact $\alpha \simeq_{S} \beta$ and Definition~\ref{def:R_beq},  we have
\[
\beta_{S}  \stackrel{\simeq_{S}}{\Longrightarrow}_{S}  \cdot \stackrel{\not\simeq_{S}}{\longrightarrow}_{S} \widehat{\beta}
\]
for some $\widehat{\beta}$ such that $\alpha' \simeq_{S} \widehat{\beta}$.
The above  path from $\beta_S$ to $\widehat{\beta}$ can be written as follows:
\[
\beta_S = \beta_0  \stackrel{\simeq_{S}}{\longrightarrow}_S \beta_1 \stackrel{\simeq_{S}}{\longrightarrow}_S \ldots \stackrel{\simeq_{S}}{\longrightarrow}_S  \beta_k  \stackrel{\not\simeq_{S}}{\longrightarrow}_S \widehat{\beta}.
\]
Consider this path.  We have two possibilities:
\begin{itemize}
\item
$\beta_i \neq \epsilon$ for every $0 \leq i \leq k$. If so, according to Lemma~\ref{lem:char_R_arrow}, we have
\[
\beta  \stackrel{\simeq_{S}}{\Longrightarrow} \cdot \stackrel{\not\simeq_{S}}{\longrightarrow} \beta' \simeq_S \widehat{\beta}
 \]
for some $\beta'$.  Actually  we have  $\beta \stackrel{\simeq_{S}}{\Longrightarrow}  \beta'' \stackrel{\tau}{\longrightarrow} \beta'$ with  $\alpha' \simeq_{S} \beta' $.
Therefore we have $\beta\gamma \stackrel{\bumpeq }{\Longrightarrow}  \beta'' \gamma \stackrel{\not\bumpeq}{\longrightarrow} \beta'\gamma$, and   $\alpha'\gamma \bumpeq \beta' \gamma$.
Furthermore,  because ${=_R} \subseteq {\simeq_R}$, we have 
\[
\beta\gamma_R \stackrel{\asymp}{\Longrightarrow}_R  \beta'' \gamma_R \stackrel{\not\asymp}{\longrightarrow}_R \beta'\gamma_R
\] 
with  $\alpha'\gamma \bumpeq \cdot \simeq_R \beta' \gamma_R$.  Note that the reason for $\beta'' \gamma_R \not\asymp \beta'\gamma_R$ is that $\beta'' \gamma_R \asymp \beta\gamma_R \asymp \alpha\gamma_R \not\asymp \alpha'\gamma_R \asymp \beta' \gamma_R$.

\item
$\beta_i = \epsilon$ for some $0 \leq i \leq k$.  Choose the largest $i$ such that $\beta_i = \epsilon$, Then according to Lemma~\ref{lem:char_R_arrow}, we have
\[
\beta_S \stackrel{\simeq_{S}}{\Longrightarrow}_S  \epsilon =_{S} X \stackrel{\simeq_{S}}{\Longrightarrow} \beta'' \stackrel{\not\simeq_{S}}{\longrightarrow} \beta'  \simeq_S \widehat{\beta}
\]
for some  $\beta'$.
Then we have the following facts.
\begin{enumerate}
\item
Because $\beta_S \stackrel{\simeq_{S}}{\Longrightarrow}_S \epsilon$, by Lemma~\ref{lem:groud_preserve} we have $\beta \stackrel{\simeq_{S}}{\Longrightarrow} \epsilon$. Then we have $\alpha \simeq_S \beta \simeq_S \epsilon$ and  $\beta\gamma \stackrel{\asymp}{\Longrightarrow} \gamma$.

\item
We know $X_S = \epsilon$, or equivalently $X \in \mathsf{Rd}_R(\gamma)$,  which means that $X\gamma \simeq_R \gamma$. Then, because $\alpha \simeq_S \beta \simeq_S \epsilon$, we have $\alpha \simeq_S X$, thus $\alpha\gamma \bumpeq X\gamma$.

\item
According to $\alpha\gamma \bumpeq X\gamma$ and $X \stackrel{\simeq_{S}}{\Longrightarrow} \beta'' \stackrel{\not\simeq_{S}}{\longrightarrow}  \beta'$, we can now take the way in the first possibility to obtain the following fact:
 $X \gamma \stackrel{\bumpeq}{\Longrightarrow}  \beta'' \gamma \stackrel{\not\bumpeq}{\longrightarrow} \beta'\gamma$ with $\alpha'\gamma \bumpeq \beta' \gamma$.

\item
Now remember  $X\gamma \simeq_R \gamma$, we have $\gamma_R  \stackrel{\simeq_R}{\Longrightarrow}_R \gamma_R'' \stackrel{\tau}{\longrightarrow}_R \gamma_R'$ for some $\gamma''$ and $\gamma'$ such that $ \beta'' \gamma \simeq_R \gamma''$ and $\beta' \gamma \simeq_R \gamma'$.
\end{enumerate}
In all, we have 
\[
\beta\gamma_R \stackrel{\asymp}{\Longrightarrow}_R \gamma_R  \stackrel{\asymp}{\Longrightarrow}_R \gamma_R'' \stackrel{\tau}{\longrightarrow}_R \gamma_R'
\] 
such that $\alpha\gamma \asymp \gamma''$ and $\alpha' \gamma \bumpeq \cdot \simeq_R  \gamma'$.  To see  $\gamma''  \not\asymp \gamma'$, we notice that $\gamma'' \asymp \beta\gamma \asymp \alpha\gamma \not\asymp \alpha'\gamma \asymp \gamma'$.
\end{itemize}

\item
$\ell \neq \tau$.  This case can be proved in the same way as the case $\ell = \tau$ and $\alpha\gamma \not\asymp \alpha'\gamma$.
\end{itemize}

\item
$\alpha = \epsilon$.  In this case  $\zeta \stackrel{\ell}{\longrightarrow} \zeta'$ is induced by $\gamma \stackrel{\ell}{\longrightarrow} \gamma'$.  Now we have $\beta \stackrel{\simeq_S}{\Longrightarrow} \epsilon$.  Thus $\beta\gamma_R \stackrel{\asymp}{\Longrightarrow}_R \gamma_R \stackrel{\ell}{\longrightarrow} \gamma_R'$, with $\alpha \gamma = \gamma \asymp \gamma$ and $(\gamma', \gamma') \in \mathcal{I}$.
\end{enumerate}

\end{proof}

\begin{lemma}\label{lem:relative_bis_if}
Suppose  $S = \mathsf{Rd}_{R}(\gamma)$,  then $\alpha \gamma \simeq_{R} \beta\gamma$ implies $\alpha \simeq_{S} \beta$.
\end{lemma}

\begin{proof}
Define the set
\[
{\asymp} \stackrel{\mathrm{def}}{=} \{(\alpha, \beta) \,|\, \alpha\gamma \simeq_R \beta\gamma\}
\]
As before we can assume $\gamma \not\in R^{*}$. Otherwise the conclusion of the lemma is relatively trivial.

We show that $\asymp$ is an $S$-bisimulation.  It is easy to see that $\asymp$ is an equivalence relation indeed.

Now we check the properties in Definition~\ref{def:R_beq}.
\begin{enumerate}
\item
We show ${=_{S}} \subseteq {\asymp}$.  Let $\alpha =_{S} \beta$. According to Definition~\ref{def:R_equal}, there exist  $\alpha', \beta' \in S^{*} = (\mathsf{Rd}_{R}(\gamma))^{*}$ and a process $\zeta$ such that $\alpha = \zeta\alpha'$ and $\beta = \zeta\beta'$. Now $\alpha\gamma = \zeta\alpha'\gamma \simeq_R \zeta\gamma \simeq_R \zeta\beta'\gamma = \beta\gamma$. Therefore $\alpha \asymp \beta$ by the definition of $\asymp$.

\item
If $\alpha \asymp \beta$ and  $\alpha \Longrightarrow \epsilon$,  we show $\beta \Longrightarrow \epsilon$. According to the definition of $\asymp$, $\alpha\gamma \simeq_R \beta\gamma$. Now $\alpha \gamma_R \Longrightarrow_R \gamma_R$ must be matched by $\beta \gamma_R$. Let us suppose that the matching is $\beta \gamma_R \Longrightarrow_R \beta'\gamma_R \simeq_R \gamma_R$, which is induced by $\beta \Longrightarrow \beta'$. Otherwise we will have $\beta \Longrightarrow \epsilon$ immediately.  Now we have $\beta'\gamma \simeq_R \gamma$, which implies $\beta'\Longrightarrow \epsilon$ and consequently $\beta \Longrightarrow \epsilon$.

\item
If $\alpha  \asymp  \beta$ and  $\alpha \stackrel{\not\asymp}{\longrightarrow} \alpha'$, then we show that
$\beta_S  \stackrel{\asymp}{\Longrightarrow}_{S} \cdot \stackrel{\not\asymp}{\longrightarrow}_{S} \beta'$ for some $\beta'$ such that $\alpha' \asymp \beta'$.
According to definition of $\asymp$, we have $\alpha\gamma \simeq_R \beta\gamma$. Moreover,  $\alpha \stackrel{\not\asymp}{\longrightarrow} \alpha'$ is equivalent to $\alpha \gamma \stackrel{\not\simeq_R}{\longrightarrow} \alpha'\gamma$. There are two cases:
\begin{itemize}
\item
$\alpha \not\in (\mathsf{Rd}_R(\gamma))^{*}$. In this case we also have $\beta  \not\in (\mathsf{Rd}_R(\gamma))^{*}$.   Thus the action $\alpha \gamma \stackrel{\not\simeq_R}{\longrightarrow} \alpha'\gamma$ must be matched by $\beta \gamma_R \stackrel{\simeq_R}{\Longrightarrow}_R  \beta''\gamma_R \stackrel{\not\simeq_R}{\longrightarrow}_R \beta'\gamma_R$ for some $\beta''$ and $\beta'$ with $\alpha'\gamma \simeq_R \beta'\gamma$.  This is equal to $\beta \stackrel{\asymp}{\Longrightarrow} \beta'' \stackrel{\not\asymp}{\longrightarrow}\beta'$ with $\alpha' \asymp \beta'$, which implies  $\beta_S \stackrel{\asymp}{\Longrightarrow}_S \beta_S'' \stackrel{\not\asymp}{\longrightarrow}_S\beta_S'$ with $\alpha' \asymp \beta'$ according to the fact ${=_{S}} \subseteq {\asymp}$.

\item
$\alpha \in (\mathsf{Rd}_R(\gamma))^{*}$.  In this case we have $\alpha \gamma \simeq_R \gamma \simeq_R  \beta \gamma$ hence $\beta \in (\mathsf{Rd}_R(\gamma))^{*}$.  In other words, $\alpha, \beta \in S^{*}$.  Now we have $\beta \gamma  \stackrel{\simeq}{\Longrightarrow}_R \gamma$, or equivalently $\beta \stackrel{\asymp}{\Longrightarrow} \epsilon$.  Now remember $\epsilon =_S \alpha$ and $\alpha \stackrel{\not\asymp}{\longrightarrow} \alpha'$.  Combine these transitions we have $\beta \stackrel{\asymp}{\Longrightarrow} \epsilon =_S \alpha \stackrel{\not\asymp}{\longrightarrow} \alpha'$, and  thus
$\beta_S \stackrel{\asymp}{\Longrightarrow}_S \epsilon = \alpha_S \stackrel{\not\asymp}{\longrightarrow}_S \alpha_S'$ and trivially $\alpha' \asymp \alpha_S'$.  We are done.
\end{itemize}

\item
If $\alpha  \asymp  \beta$ and  $\alpha \stackrel{a}{\longrightarrow} \alpha'$, then we show
$\beta_S \stackrel{\asymp}{\Longrightarrow}_{S} \cdot \stackrel{a}{\longrightarrow}_{S} \beta'$ for some $\beta'$ such that $\alpha' \asymp \beta'$. This case can be treated in the same way as the previous one.
\end{enumerate}
\end{proof}

\subsection{Proofs Concerning $R$-redundancy}
\subsubsection*{Proof of Lemma~\ref{lem:redundant_chain}}
Suppose $X \in \mathsf{Rd}_{\mathsf{Rd}_{R}(\delta)}(\gamma)$. By Definition~\ref{def:Id_Rd_R}, this is equivalent to $X\gamma  \simeq_{\mathsf{Rd}_{R}(\delta)} \gamma$. According to Theorem~\ref{thm:relative_bis_str}, $X\gamma\delta  \simeq_R \gamma\delta$, which means $X \in \mathsf{Rd}_R(\gamma \delta)$.  The proof of  the other direction is from the fact that  the above reasoning steps  are reversible.

\subsubsection*{Proof of Lemma~\ref{lem:Rd_admissible}}
We show that $\mathsf{Id}_{R} = R$. That is, $X \simeq_R \epsilon$ if and only if $X \in R$.  It is trivial that $X \simeq_R \epsilon$ whenever $X \in R$. So it suffices to show $X \in R$ whenever $X \simeq_R \epsilon$.

Assume that $X \simeq_R \epsilon$.  Since $R =\mathsf{Rd}_{R'}(\gamma)$, by  Theorem~\ref{thm:relative_bis_str}, $X\gamma \simeq_{R'} \gamma$ which implies $X \in \mathsf{Rd}_{R'}(\gamma) = R$ according to Definition~\ref{def:Id_Rd_R}.

\subsection{Proofs of Lemma~\ref{lem:R_prime_property} and Theorem~\ref{thm:QUDP_RBisimularity}}
\subsubsection*{Proof of Lemma~\ref{lem:R_prime_property}}
There are two cases.
\begin{enumerate}
\item
$\alpha X \simeq_{R} X$. In this case, we have $X \simeq_{R} \beta Y$. Since $X$ and $Y$ are both $\simeq_{R}$-primes, we must have $\beta Y \simeq_{R} Y$. Therefore $X \simeq_{R} Y$.

\item
$\alpha X \not\simeq_{R} X$.  In this case, we also have $\alpha X \not\simeq_{R} Y$, $\beta Y \not\simeq_{R} Y$ and $\beta Y \not\simeq_{R} X$. Consider the following sequence of transitions
\[
\alpha X \stackrel{\ell_1}\longrightarrow \alpha_1 X \stackrel{\ell_2}\longrightarrow \ldots \stackrel{\ell_k}{\longrightarrow} \alpha_k X \stackrel{\ell}{\longrightarrow}\alpha' X
\]
with $\alpha_i X \not\simeq_{R} X$ for $1 \leq i \leq k$ and $\alpha' X \simeq_{R} X$. This sequence must be matched by $\beta Y$ via
\[ \beta Y  \stackrel{\ell_1}\Longrightarrow \beta_1 Y \stackrel{\ell_2}\Longrightarrow  \ldots  \stackrel{\ell_k}\Longrightarrow \beta_k Y \stackrel{\ell}{\longrightarrow} \beta'Y,
\]
such that
\begin{itemize}
\item
$\beta_i Y \not\simeq_{R} Y$ for $1 \leq i \leq k$ and $\beta' Y \simeq_{R} Y$;  and

\item
$\alpha_i X  \simeq_{R} \beta_i Y $ for $1 \leq i \leq k$ and $\alpha' X  \simeq_{R} \beta' Y$.
\end{itemize}
Accordingly $X \simeq_{R} \alpha' X  \simeq_{R} \beta' Y \simeq_{R} Y$.
\end{enumerate}

\subsubsection*{Proof of  Theorem~\ref{thm:QUDP_RBisimularity}}

By Induction on $r$ or $s$.  Remember from
Definition~\ref{def:R_prime_decomposition}, $P_i$ is a $\simeq_{R_i}$-prime and $Q_j$ is a $\simeq_{S_j}$-prime.

Because $P_r  .  P_{r-1}  .  \ldots .  P_{2} .  P_{1} \simeq_{R_1} Q_s  .  Q_{s-1}  .  \ldots  .  Q_{2} .  Q_{1}$, by  Lemma~\ref{lem:R_prime_property}, we have $P_1 \simeq_{R_1} Q_1$, which implies  $\mathsf{Rd}_{R_1}(P_1) = \mathsf{Rd}_{R_1}(Q_1)$ by Lemma~\ref{lem:Rd_R_basic_facts}. In other words, $R_2 = S_2$.  According to Lemma~\ref{prop:Rd_congruence}, we have $P_r  .  P_{r-1}  .  \ldots .  P_{2} \simeq_{R_2} Q_s  .  Q_{s-1}  .  \ldots  .  Q_{2}$, and now the proof is accomplished by  using induction hypothesis.

\section{Proofs in Section~\ref{sec:naive-algorithm}}

\subsection{Proofs Concerning Relationships between Bases}

\subsubsection*{Proof of Lemma~\ref{lem:compare_ID}}

If $X \in \mathbf{Id}^{\mathcal{B}}_R$,  then $X  \stackrel{\mathcal{B}}{=}_{R} \epsilon$. Since ${\stackrel{\mathcal{B}}{=}_{R}} \subseteq  {\stackrel{\mathcal{D}}{=}_{R}}$,  we have $X  \stackrel{\mathcal{D}}{=}_{R} \epsilon$, hence $X \in \mathbf{Id}^{\mathcal{D}}_R$.

\subsubsection*{Proof of Lemma~\ref{lem:compare_PR}}

First note that, because $\mathbf{Id}^{\mathcal{B}}_R  = \mathbf{Id}^{\mathcal{D}}_R$, a reference set $R$ is $\mathcal{B}$-admissible if and only if it is $\mathcal{D}$-admissible.
The members in $\mathbf{Pr}_R$ and $\mathbf{Pr}_R'$ are both the $\mathbf{Id}_R$-blocks. Thus comparing $\mathbf{Pr}_R'$ with $\mathbf{Pr}_R$ makes sense.

We only need to prove the conclusion for $\mathcal{B}$-admissible (also $\mathcal{D}$-admissible) $R$'s. Keep in mind that $R = \mathbf{Id}^{\mathcal{B}}_R  = \mathbf{Id}^{\mathcal{D}}_R$.

Now suppose $[X]_{R} \in \mathbf{Cm}^{\mathcal{B}}_R$,  we will show $[X]_{R}  \in \mathbf{Cm}^{\mathcal{D}}_R$. Since $[X]_{R} \in \mathbf{Cm}^{\mathcal{B}}_R$, we can suppose that $X  \stackrel{\mathcal{B}}{=}_{R} \alpha . P$  in which $[P]_{R} \in \mathbf{Pr}^{\mathcal{B}}_R$ and $[X]_{R} \neq [P]_{R}$. (Caution: we cannot assert that $P <_R X$. )  Because  ${\stackrel{\mathcal{B}}{=}_{R}}  \subseteq {\stackrel{\mathcal{D}}{=}_{R}}$, we also have $X \stackrel{\mathcal{D}}{=}_{R} \alpha. P$.  Since $P \not\in R$, either  $[P]_{R} \in   \mathbf{Pr}^{\mathcal{D}}_{R}$ or $[P]_{R} \in   \mathbf{Cm}^{\mathcal{D}}_{R}$ will happen.

If $[P]_{R} \in   \mathbf{Pr}^{\mathcal{D}}_{R}$, it is done. Otherwise  we must  have $P  \stackrel{\mathcal{D}}{=}_{R} \alpha'. P'$ such that $[P']_{R}  \in \mathbf{Pr}^{\mathcal{D}}_{R}$.  Thus $X \stackrel{\mathcal{D}}{=}_{R} \alpha. \alpha'.P'$ with $[P']_{R}  \in \mathbf{Pr}^{\mathcal{D}}_{R}$. Therefore, in either case $[X]_{R}  \not\in \mathbf{Pr}^{\mathcal{D}}_{R}$. The only possibility is $[X]_{R} \in \mathbf{Cm}^{\mathcal{D}}_{R}$.

\subsubsection*{Proof of Lemma~\ref{lem:compare_RD}}

It suffices to prove $\mathbf{Rd}^{\mathcal{B}}_R \subseteq \mathbf{Rd}^{\mathcal{D}}_R$ for $\mathcal{B}_R$-admissible (also $\mathcal{B}_R$-admissible) $R$'s. Let $[P]_R$ be a $\mathcal{B}_R$-prime (also $\mathcal{D}_R$-prime).
According to the rules of $\mathcal{B}_R$-reduction defined in Section~\ref{subsec:decomposition_base}, $X \in \mathbf{Rd}^{\mathcal{B}}_R([P]_R)$ if and only if $X P \stackrel{\mathcal{B}}{=}_{R} P$. Thus  $X P\stackrel{\mathcal{D}}{=}_{R} P$ hence $ X \in \mathbf{Rd}^{\mathcal{D}}_R([P]_R)$, since  ${\stackrel{\mathcal{B}}{=}_{R}} \subseteq {\stackrel{\mathcal{D}}{=}_{R}}$.

\subsubsection*{Proof of Lemma~\ref{lem:compare_all}}

First we point out that $\mathbf{Cm}^{\mathcal{B}}_R = \mathbf{Cm}^{\mathcal{D}}_R$ whenever $\mathbf{Id}^{\mathcal{B}}_R  = \mathbf{Id}^{\mathcal{D}}_R$ and $\mathbf{Pr}^{\mathcal{B}}_R = \mathbf{Pr}^{\mathcal{D}}_R$.

Now we prove $\mathcal{B} = \mathcal{D}$.  Suppose on the contrary that $\mathcal{B} \subsetneq \mathcal{D}$. There is some  $\alpha$ and some  $\mathcal{B}_R$-admissible (also $\mathcal{B}_R$-admissible) $R$ such that $\mathtt{dcmp}_R^{\mathcal{B}}(\alpha) \neq \mathtt{dcmp}_R^{\mathcal{D}}(\alpha)$.  By examine the $\mathcal{B}$-reduction rules, The only possibility to make $\mathtt{dcmp}_R^{\mathcal{B}}(\alpha) \neq \mathtt{dcmp}_R^{\mathcal{D}}(\alpha)$  is the existence of $\mathcal{B}_R$-admissible (also $\mathcal{B}_R$-admissible) reference set $S$ together with  $[X]_S \in \mathbf{Cm}^{\mathcal{B}}_S  = \mathbf{Cm}^{\mathcal{D}}_S$ such that $\mathbf{Dc}^{\mathcal{B}}_S([X]_S) \neq \mathbf{Dc}^{\mathcal{D}}_S([X]_S)$.  Let us say $\mathbf{Dc}^{\mathcal{B}}_S([X]_S) = [X_r]_{S_{r}} [X_{r-1}]_{S_{r-1}}\ldots [X_2]_{S_{2}} [X_1]_{S_{1}}$, in which  $S_1 = S$, $[X_i]_{S_{i}} \in \mathbf{Pr}^{\mathcal{B}}_{S_{i}}$ and $S_{i+1} = \mathbf{Rd}^{\mathcal{B}}_{S_{i}}([X_i]_{S_{i}})$ for every $1 \leq i \leq r$.  Consequently, $[X_i]_{S_{i}} \in \mathbf{Pr}^{\mathcal{D}}_{S_{i}}$ and $S_{i+1} = \mathbf{Rd}^{\mathcal{D}}_{S_{i}}([X_i]_{S_{i}})$ for every $1 \leq i \leq r$ according to the condition of the lemma.
Now, because ${\stackrel{\mathcal{B}}{=}_S}   \subseteq {\stackrel{\mathcal{D}}{=}_S}$, we have  $X \stackrel{\mathcal{D}}{=} X_r.X_{r-1}.\ldots. X_1$. Thus $\mathtt{dcmp}_R^{\mathcal{D}}(X) = \mathtt{dcmp}_R^{\mathcal{D}}(X_r.X_{r-1}.\ldots. X_1)$. Amazingly, $\mathtt{dcmp}_R^{\mathcal{D}}(X_r.X_{r-1}.\ldots. X_1)
$ turns out to be $[X_r]_{S_{r}} [X_{r-1}]_{S_{r-1}}\ldots [X_2]_{S_{2}} [X_1]_{S_{1}}$, which implies $\mathbf{Dc}^{\mathcal{D}}_S([X]_S) = [X_r]_{S_{r}} [X_{r-1}]_{S_{r-1}}\ldots [X_2]_{S_{2}} [X_1]_{S_{1}} = \mathbf{Dc}^{\mathcal{B}}_S([X]_S)$.  This is a contradiction.

\subsection{Proofs of  the Properties of Construction}

\subsubsection*{Proofs of Lemma~\ref{lem:step3}}

Clearly  $R \subseteq \mathbf{Id}^{\mathcal{B}}_R$. According to  Proposition~\ref{prop:R_monotone}, ${\simeq_{R}} \subseteq {\simeq_{\mathbf{Id}^{\mathcal{B}}_R}}$.  On the other hand, ${\stackrel{\mathcal{B}}{=}_{\mathbf{Id}^{\mathcal{B}}_R}} = {\stackrel{\mathcal{B}}{=}_{R}}$.  Therefore, ${\simeq_{R}} \subseteq {\simeq_{\mathbf{Id}^{\mathcal{B}}_R}} \subseteq {\stackrel{\mathcal{B}}{=}_{\mathbf{Id}^{\mathcal{B}}_R}} = {\stackrel{\mathcal{B}}{=}_{R}}$.

\subsubsection*{Proofs of Lemma~\ref{lem:step1}}
It is a routine work to check that $\{ X \;|\; X \simeq_R \epsilon\}$ is an  $\mathbf{Id}^{\mathcal{B}}_R$-candidate.

\subsubsection*{Proofs of Lemma~\ref{lem:admissible_Rd}}

We need to show that, $\mathbf{Id}^{\mathcal{B}}_{\mathbf{Rd}^{\mathcal{B}}_R([X]_R)} = {\mathbf{Rd}^{\mathcal{B}}_R([X]_R)} $. To this end, we show $\mathbf{Id}^{\mathcal{B}}_{\mathbf{Rd}^{\mathcal{B}}_R([X]_R)}$ is an $\mathbf{Rd}^{\mathcal{B}}_R([X]_R)$-candidate.

As before we let  $T \stackrel{\mathrm{def}}{=} \{ W \;|\;  W.X \stackrel{\mathcal{D}}{=}_R X \}$.
First we confirm $\mathbf{Id}^{\mathcal{B}}_{\mathbf{Rd}^{\mathcal{B}}_R([X]_R)} \subseteq T$.
 By induction $T$ is a  $\mathcal{D}$-admissible set thus it is also $\mathcal{B}$-admissible set, hence  $\mathbf{Id}^{\mathcal{B}}_{T} = T$.  Because $\mathbf{Rd}^{\mathcal{B}}_R([X]_R) \subseteq T$, we have $ \mathbf{Id}^{\mathcal{B}}_{\mathbf{Rd}^{\mathcal{B}}_R([X]_R)} \subseteq \mathbf{Id}^{\mathcal{B}}_{T} = T$.

Now we check the conditions of $\mathbf{Rd}^{\mathcal{B}}_R([X]_R)$-candidate. Let
$Y \in \mathbf{Id}^{\mathcal{B}}_{\mathbf{Rd}^{\mathcal{B}}_R([X]_R)}$. If $Y \in \mathbf{Rd}^{\mathcal{B}}_R([X]_R)$, then nothing need to to. Now we suppose that $Y \not\in \mathbf{Rd}^{\mathcal{B}}_R([X]_R)$.

\begin{enumerate}
\item
If $Y  \stackrel{\tau}{\longrightarrow} \zeta$ and $\zeta \not\in T^{*}$.  In this case, because $ \mathbf{Id}^{\mathcal{B}}_{\mathbf{Rd}^{\mathcal{B}}_R([X]_R)} \subseteq T$, we have $\zeta \not\in (\mathbf{Id}^{\mathcal{B}}_{\mathbf{Rd}^{\mathcal{B}}_R([X]_R)})^{*}$.  Thus  $Y  \stackrel{\tau}{\longrightarrow}_{\mathbf{Rd}^{\mathcal{B}}_R([X]_R)} \widehat{\zeta} =_{{\mathbf{Rd}^{\mathcal{B}}_R([X]_R)}} \zeta$.
Now according to the definition of $\mathbf{Id}^{\mathcal{B}}_{\mathbf{Rd}^{\mathcal{B}}_R([X]_R)}$, we have $\epsilon  \stackrel{\tau}{\longrightarrow}_{\mathbf{Rd}^{\mathcal{B}}_R([X]_R)} \widehat{\eta}$ for some $ \widehat{\eta}$  such that $\mathtt{dcmp}_{\mathbf{Rd}^{\mathcal{B}}_R([X]_R)}^{\mathcal{D}}(\widehat{\zeta}) = \mathtt{dcmp}_{\mathbf{Rd}^{\mathcal{B}}_R([X]_R)}^{\mathcal{D}}(\widehat{\eta})$.
In other words, there is
$Z \in \mathbf{Rd}^{\mathcal{B}}_R([X]_R)$ and $Z  \stackrel{\tau}{\longrightarrow} \eta$ for some $\eta$  such that $\mathtt{dcmp}_{\mathbf{Rd}^{\mathcal{B}}_R([X]_R)}^{\mathcal{D}}(\zeta) = \mathtt{dcmp}_{\mathbf{Rd}^{\mathcal{B}}_R([X]_R)}^{\mathcal{D}}(\eta)$. This makes $\zeta X \stackrel{\mathcal{D}}{=}_R \eta X$.
Now, we use the fact that $\mathbf{Rd}^{\mathcal{B}}_R([X]_R)$ itself is an $\mathbf{Rd}^{\mathcal{B}}_R([X]_R)$-candidate. Thus it satisfies the relevant conditions. That is,
$[X]_{\mathbf{Id}_{R}}  \stackrel{\tau}{\longmapsto} \beta$  for some $\beta$  such that $\eta. X \stackrel{\mathcal{D}}{=}_R \beta$.  And finally we have $\zeta X \stackrel{\mathcal{D}}{=}_R \beta$.

\item
If $Y \stackrel{a}{\longrightarrow}  \zeta$.  The proof is complete the same as the first case.
\end{enumerate}

\subsubsection*{Proof of Lemma~\ref{lem:corrct_norm}}

Apparently we can only prove the proposition for  $\mathcal{B}$-admissible $R$'s.

The proof is by induction.  Assume at some time in the the execution of the algorithm, we have the set $\mathbf{V}$ which contains all the treated blocks, and we have a current value of $m$, and current block $[X]_R$. The induction hypothesis is the following:
\begin{quote}
If $[X]_R \in \mathbf{V}$, then $d_R[[X]_R] = \| X \|_{\stackrel{\mathcal{B}}{=}_R}$;  if $\alpha$ satisfies $\mathtt{dcmp}_{R}(\alpha) \in \mathbf{V}^{*}$, then $d_R(\alpha) = \| \alpha \|_{\stackrel{\mathcal{B}}{=}_R}$.
\end{quote}
According to the algorithm, it is clear that $d_R[[X]_R] \geq \| X \|_{\stackrel{\mathcal{B}}{=}_R}$. The reason is elaborated as follows.
\begin{enumerate}
\item
If $d_R[[X]_R] = m$ via the fact $[X]_R \stackrel{\ell}{\longmapsto} \gamma$ and $d_R(\gamma) = m - 1$. Then there is a path from $X$ to $\gamma$ with length $1$ and there is a path from $\gamma$ to $\epsilon$ with length $m-1$. Thus totally we have a path from $X$ to $\epsilon$ with length $m$.

\item
If $d_R[[X]_R] = m$ via the fact $[X]_R \stackrel{\ell}{\longmapsto} \gamma$, $d_R(\gamma) = m$ and $X \stackrel{\mathcal{B}}{=}_R \gamma$. Then there is a path from $X$ to $\gamma$ with length $0$ and there is a path from $\gamma$ to $\epsilon$ with length $m$. Thus totally we have a path from $X$ to $\epsilon$ with length $m$.
\end{enumerate}
Thus in both case we have  $d_R[[X]_R] \geq \| X \|_{\stackrel{\mathcal{B}}{=}_R}$.

Now assume, for contradiction, that $d_R[[X]_R] > \| X \|_{\stackrel{\mathcal{B}}{=}_R}$.  In other words, $m > \| X \|_{\stackrel{\mathcal{B}}{=}_R}$.  Then according to induction hypothesis, for every $[Y]_S \in \mathbf{V}$,  $d_{R}[[Y]_S] = \| Y \|_{\stackrel{\mathcal{B}}{=}_S}$.
Now consider the time when $m$ is assigned to $d_{R}[[X]_R]$. There are two possibilities:
\begin{enumerate}
\item
$[X]_R \stackrel{\ell}{\longmapsto} \gamma$ and $d_R(\gamma) = m - 1$.  In this case, by induction $d_{R}(\gamma) = \| \gamma \|_{\stackrel{\mathcal{B}}{=}_R} = m-1$.    There can not be other transition of $[X]_R$ such as $[X]_R \stackrel{\ell}{\longmapsto}  \zeta$ that
\begin{enumerate}
\item
either $d_{R}(\zeta) < m-1$,
\item
or $d_{R}(\zeta) = m-1$ and $\zeta  \stackrel{\mathcal{B}}{=}_R X$.
\end{enumerate}
If so, $m-1$ would be assigned to $d_{R}[[X]_R]$ and  the block $[X]_R$ should have already been put into $\mathbf{V}$.   This is a contradiction.

\item
$ [X]_R \stackrel{\ell}{\longmapsto} \gamma$, $d_R(\gamma) = m$ and $X \stackrel{\mathcal{B}}{=}_R \gamma$.  In this case, by induction $d_{R}(\gamma) = \| \gamma \|_{\stackrel{\mathcal{B}}{=}_R} = m$, thus we must have $ \|X  \|_{\stackrel{\mathcal{B}}{=}_R} = \| \gamma \|_{\stackrel{\mathcal{B}}{=}_R} = m$. This is a contradiction.

\end{enumerate}

\section{Proof of Theorem~\ref{thm:correctness}}

\subsection{Proof of Lemma~\ref{lem:unique_expansion}}
There are several cases according to the values of $k$ and $l$:
\begin{itemize}
\item
If $k, l > 1$.  In this case we can assume that $\|\gamma\|_{\stackrel{\mathcal{B}}{=}_R} =\|\delta\|_{\stackrel{\mathcal{B}}{=}_R}$, which is already known,  and we can tell whether a given action is $\stackrel{\mathcal{B}}{=}_R$-decreasing.  Suppose we have a decreasing transition of $\gamma \Longleftrightarrow \cdot \stackrel{\ell}{\longrightarrow}_{R} \eta. Y_{k-1}\ldots Y_1$  which is induced by $[Y_k]_{R_k} \stackrel{\ell}{\longmapsto}_{R_k} \eta$.  Now we have the matching transition  $\delta\Longleftrightarrow \cdot \stackrel{\ell}{\longrightarrow}_{R} \zeta. Z_{l-1}\ldots Z_1$ which is induced by $[Z_l]_{R_l} \stackrel{\ell}{\longmapsto}_{R_l} \zeta$.  Moreover, we have
\[
\eta. Y_{k-1}.\ldots. Y_1  \stackrel{\mathcal{B}}{=}_R \zeta.Z_{l-1}.\ldots. Z_1
\]
Now we must have $Y_1 = Z_1$. And the process $Y_k \ldots Y_2$ and $Z_l \ldots Z_2$ also satisfy the expansion property for $\mathcal{B}_R$. The result of the lemma can be obtained by induction.

\item
If $k = 1$ and $l > 1$. In this case, $[Y_1]_R$ must not be a prime.  According to our algorithm, one of the candidates will be defined as $\mathbf{Dc}^{\mathcal{B}}_R([Y_1]_R)$ if there exist some candidates which can pass the expansion testing.

\item
If $k = l = 1$.  If $[Z_1]_R <_R [Y_1]_R$, then this is the same as the above case. Otherwise we can change the role of $Y_1$ and $Z_1$.
\end{itemize}

\subsection{Preparations for the Proof of Theorem~\ref{thm:correctness}}

To make things clear, we introduce some new terminologies.  Note that the program in Fig.~\ref{fig:new_base_II} maintains a set $\mathbf{V}$ of the blocks which have already been treated.  During the execution of the algorithm, $\mathbf{V}$ start from $\emptyset$ and get larger and larger.  Intuitively these blocks in $\mathbf{V}$  contain part of information of ${\mathcal{B}}$.  Formally, we can define the {\em partial decomposition base} $\mathcal{B}_{\mathbf{V}} = \{ \mathcal{B}_{R,\mathbf{V}}\}_{R\subseteq \mathbf{C}_{G}}$ in which $\mathcal{B}_{R,\mathbf{V}} = (\mathbf{Id}^{\mathcal{B}}_R, \mathbf{Pr}^{\mathcal{B}}_{R,\mathbf{V}}, \mathbf{Cm}^{\mathcal{B}}_{R,\mathbf{V}} ,\mathbf{Dc}^{\mathcal{B}}_{R,\mathbf{V}}, \mathbf{Rd}^{\mathcal{B}}_{R,\mathbf{V}})$ where
\begin{itemize}
\item
$\mathbf{Pr}^{\mathcal{B}}_{R,\mathbf{V}} = \mathbf{Pr}^{\mathcal{B}}_{R} \cap \mathbf{V}$.

\item
$\mathbf{Cm}^{\mathcal{B}}_{R,\mathbf{V}} = \mathbf{Cm}^{\mathcal{B}}_{R} \cap \mathbf{V}$.

\item
$\mathbf{Dc}^{\mathcal{B}}_{R,\mathbf{V}}([X]_R) =\left\{
 \begin{array}{ll}
  \mathbf{Dc}^{\mathcal{B}}_{R}([X]_R)  & \;  \textrm{if  $[X]_R \in \mathbf{Cm}^{\mathcal{B}}_{R,\mathbf{V}}$} \\
  \mbox{undefined}                      & \; \textrm{otherwise}
 \end{array}
 \right.$

\item
$\mathbf{Rd}^{\mathcal{B}}_{R,\mathbf{V}}([X]_R) = \left\{
 \begin{array}{ll}
  \mathbf{Rd}^{\mathcal{B}}_{R}([X]_R)  & \;  \textrm{if  $[X]_R \in \mathbf{Pr}^{\mathcal{B}}_{R,\mathbf{V}}$} \\
  \mbox{undefined}                      & \; \textrm{otherwise}
 \end{array}
 \right.$
\end{itemize}
$\mathcal{B}_{\mathbf{V}}$ is called {\em partial} in the sense that $\mathbf{Pr}^{\mathcal{B}}_{R,\mathbf{V}} \cup \mathbf{Cm}^{\mathcal{B}}_{R,\mathbf{V}} = \mathbf{C}_{R} \cap \mathbf{V} \subseteq \mathbf{C}_{R}$. Comparatively, $\mathbf{Pr}^{\mathcal{B}}_{R} \cup \mathbf{Cm}^{\mathcal{B}}_{R} = \mathbf{C}_{R}$.

At some time in the execution of the algorithm,
we get a specific value of ${\mathbf{V}}$, then  $\mathcal{B}_{\mathbf{V}}$ is already known at that time. Now we can define  $\mathtt{dcmp}^{\mathcal{B}}_{R, \mathbf{V}}(\alpha)$  for any process $\alpha$.   $\mathtt{dcmp}^{\mathcal{B}}_{R, \mathbf{V}}(\alpha) = \mathtt{dcmp}^{\mathcal{B}}_{R}(\alpha)$ if the derivation of $\alpha \stackrel{\mathcal{B}}{\rightarrow}_R \mathtt{dcmp}^{\mathcal{B}}_{R}(\alpha)$  (refer to Section~\ref{subsec:decomposition_base})  only relies on  the information provided in $\mathcal{B}_{\mathbf{V}}$. Otherwise
$\mathtt{dcmp}^{\mathcal{B}}_{R, \mathbf{V}}(\alpha)$ is undefined. In the following, a process $\alpha$ is called {\em $\mathcal{B}_{R, \mathbf{V}}$-applicable} if $\mathtt{dcmp}^{\mathcal{B}}_{R, \mathbf{V}}(\alpha)$ is defined.

Now we prepare to confirm the important result:
 $\alpha \simeq_R \beta$ implies $\alpha \stackrel{\mathcal{B}}{=}_R \beta$.

First of all, we find that it is enough to prove the result under the assumption that $R$ is $\mathcal{B}$-admissible.  Note that If  $ {\simeq_R} \subseteq {\stackrel{\mathcal{B}}{=}_R}$ for every $\mathcal{B}$-admissible $R$'s, then for every $R$ we have ${\simeq_R} \subseteq {\simeq_{\mathbf{Id}^{\mathcal{B}}_R}} \subseteq {\stackrel{\mathcal{B}}{=}_{\mathbf{Id}^{\mathcal{B}}_R}} = {\stackrel{\mathcal{B}}{=}_R}$. Thus in the rest of this section we assume $R$ to be $\mathcal{B}$-admissible.

With the help of Lemma~\ref{lem:unique_expansion}, we are able to establish  Theorem~\ref{thm:correctness}.

We take the following approach to prove Theorem~\ref{thm:correctness}. Remember that our algorithm maintains a set $\mathbf{V}$, containing all the blocks which have been treated. We will
suppose that $R$ is $\mathcal{B}$-admissible. Let $[X]_R$ be a block which is about to be put into $\mathbf{V}$. We try to prove that if $[X]_R$ is a $\widehat{\mathcal{B}}_R$-composite, and let $\mathbf{Dc}^{\widehat{\mathcal{B}}}_{R}([X]_R) = [Z_t]_{R_t}\ldots [Z_1]_{R_1}$, then $X \stackrel{\mathcal{B}}{=}_R  Z_t\ldots Z_1$.

Apparently, the proof must be done by induction. However, this is not an easy task.  We will choose the following statements as our induction hypotheses:
\begin{enumerate}
\item[\textbf{I}.]
Let $R$ be an arbitrary $\mathcal{B}$-admissible set. Suppose that $\gamma$ is $\mathcal{B}_{R,\mathbf{V}}$-applicable, and $\mathtt{dcmp}^{\widehat{\mathcal{B}}}_{R}(\gamma) = [W_u]_{R_u} \ldots[W_1]_{R_1}$. Then $W_u\ldots W_1$ is $\mathcal{B}_{R,\mathbf{V}}$-applicable, and $\mathtt{dcmp}^{\mathcal{B}}_{R,\mathbf{V}}(\gamma) = \mathtt{dcmp}^{\mathcal{B}}_{R,\mathbf{V}}(W_u \ldots W_1)$. That is, $\gamma \stackrel{\mathcal{B}}{=}_R W_u \ldots W_1$.
\end{enumerate}
We remark that at the time the algorithm terminates when $\mathbf{V}$ contains every blocks, the statement~I implies Theorem~\ref{thm:correctness}.

The readers are suggested to imagine the following picture in mind. Although $R_1, \ldots, R_t$ are all $\widehat{\mathcal{B}}$-admissible according to the definition of decomposition base, we cannot draw the conclusion that $R_1, \ldots, R_t$ are all $\mathcal{B}$-admissible. It may indeed happen that $Z_i \in \mathbf{Id}^{\mathcal{B}}_{R_i}$, which means that $Z_i$ is $\mathcal{B}_{R_i}$-redundant. However, since $R_1 = R$ is $\mathcal{B}$-admissible, we do have $Z_1 \not\in \mathbf{Id}^{\mathcal{B}}_{R_1}$. This fact will be used in the proof.

\subsection{Proof of Lemma~\ref{lem:pre_correctness}}

This fact can be proved by induction on $d = \|W_u \ldots W_1\|_{\stackrel{\mathcal{B}}{=}_R}$, and by studying a witness path of $\simeq_R$-norm for $W$. Then using hypothesis I and by inspecting the algorithm we can show that when $W \simeq_R W_u \ldots W_1$, $W$ should have already been put into $\mathbf{V}$.

\subsection{Proof of Lemma~\ref{lem:applicable}}

We use $\delta$ to indicate $Z_t \ldots Z_1$.
The lemma  confirms that, when $[X]_R$ is about to be put into $\mathbf{V}$,  $\mathtt{dcmp}^{\mathcal{B}}_{R,\mathbf{V}}(Z_t\ldots Z_1)$ has already been defined.  This fact is proved by induction, using the induction hypotheses.
The way is to choose a process $\gamma$ such that $[X]_R \stackrel{\ell}{\longmapsto}_R \gamma$ is on a $\stackrel{\mathcal{B}}{=}_R$-witness path of $X$. Now $\gamma$ is $\mathcal{B}_{R,\mathbf{V}}$-applicable, thus we try to use induction hypotheses on $\gamma$.  There are two cases:
\begin{itemize}
\item
\sloppy
If there exist $\gamma$ such that $[X]_R \stackrel{\ell}{\longmapsto}_R \gamma$ and $\|\gamma \|_{\stackrel{\mathcal{B}}{=}_R}  = m-1$.
In this case, the algorithm is running in the first \textbf{while}-loop in Fig.~\ref{fig:new_base_II}.   Since $X \simeq_R \delta$, there is a matching of $[X]_R \stackrel{\ell}{\longmapsto}\gamma$ from $\delta$, say $\delta \stackrel{\simeq_R }{\Longrightarrow} \cdot \stackrel{\ell}{\longrightarrow} \zeta$ for some $\zeta$ that $\gamma \simeq_R \zeta$.
Because  $\delta$ is itself a $\simeq_{R}$-prime-decomposition, we must have $\delta \Longleftrightarrow \cdot \stackrel{\ell}{\longrightarrow} \zeta$, which is induced by $[Z_t]_{R_t} \stackrel{\ell}{\longmapsto}_{R_t} \eta$ for some $\eta$, and $\zeta = \eta.Z_{t-1}\ldots Z_1$. Suppose $\mathtt{dcmp}^{\widehat{\mathcal{B}}}_{R_t}(\eta) = [Y_s]\ldots[Y_1]$ ($s \geq 0$), then $\mathtt{dcmp}^{\widehat{\mathcal{B}}}_{R}(\zeta)$ must be in the form $[Y_s]\ldots[Y_1]. [Z_{t-1}]_{R_{t-1}}.\ldots [Z_1]_{R_1}$.
Because $\gamma \simeq_R \zeta$, we have $\mathtt{dcmp}^{\widehat{\mathcal{B}}}_{R}(\gamma) =[Y_s]\ldots[Y_1]. [Z_{t-1}]_{R_{t-1}}.\ldots [Z_1]_{R_1}$.
According to induction hypothesis I, $\gamma \stackrel{\mathcal{B}}{=}_R  Y_s\ldots Y_1.Z_{t-1}\ldots Z_1$, thus $\| Y_s\ldots Y_1.Z_{t-1}\ldots Z_1 \|_{\stackrel{\mathcal{B}}{=}_R} = m-1$. Now since $\mathtt{dcmp}^{\widehat{\mathcal{B}}}_{R}(\zeta) = [Y_s]\ldots[Y_1]. [Z_{t-1}]_{R_{t-1}}.\ldots [Z_1]_{R_1}$, by Lemma~\ref{lem:pre_correctness}, $\zeta \stackrel{\mathcal{B}}{=}_R  Y_s\ldots Y_1.Z_{t-1}\ldots Z_1$ and thus $\| \zeta \|_{\stackrel{\mathcal{B}}{=}_R} = m-1$. Since $\delta \Longleftrightarrow \cdot \stackrel{\ell}{\longrightarrow} \zeta$, we have $\| \delta \|_{\stackrel{\mathcal{B}}{=}_R} \leq m$.  In summary, we have:
\begin{itemize}
\item
$\|\zeta \|_{\stackrel{\mathcal{B}}{=}_R} = \|Y_s \ldots Y_1.Z_{t-1} \ldots Z_1 \|_{\stackrel{\mathcal{B}}{=}_R} =\|\eta.Z_{t-1} \ldots Z_1 \|_{\stackrel{\mathcal{B}}{=}_R} = m-1$.

\item
$\| Z_{t-1} \ldots Z_1 \|_{\stackrel{\mathcal{B}}{=}_R} \leq m-1$.

\item
$\| \delta \|_{\stackrel{\mathcal{B}}{=}_R} \leq m$.

\item
$\| Z_1 \|_{\stackrel{\mathcal{B}}{=}_R} > 0$.
\end{itemize}
There are two possibilities:
\begin{itemize}
\item
\textsc{Either} $\| Z_i \|_{\stackrel{\mathcal{B}}{=}_{R_i}} < m$ for every $1\leq i \leq t$. In this case  we have $\delta =  Z_t \ldots Z_1 $ is $\mathcal{B}_{R,\mathbf{V}}$-applicable trivially.

\item
\textsc{Or} $t=1$ and $\| Z_1 \|_{\stackrel{\mathcal{B}}{=}_R} = m$.  In this case, we have $[Z_1]_R <_R [X]_R$ and $\| Z_1 \|_{\stackrel{\mathcal{B}}{=}_R} = \| X \|_{\stackrel{\mathcal{B}}{=}_R}$.  Thus $[Z_1]_R \in \mathbf{V}$ and thus $\delta = Z_1$ is $\mathcal{B}_{R,\mathbf{V}}$-applicable.
\end{itemize}

\item
If for every $\gamma$ such that $[X]_R \stackrel{\ell}{\longmapsto}_R \gamma$, we do not have $\|\gamma \|_{\stackrel{\mathcal{B}}{=}_R}  < m$.
In this case, the algorithm is running in the second \textbf{while}-loop in Fig.~\ref{fig:new_base_II}.  We are able to find a $\gamma$ which is $\mathcal{B}_{R,\mathbf{V}}$-applicable such that $[X]_R \stackrel{\tau}{\longmapsto}_R \gamma$ and $X \stackrel{\mathcal{B}}{=}_R \gamma$, and $\|\gamma \|_{\stackrel{\mathcal{B}}{=}_R} = \| X \|_{\stackrel{\mathcal{B}}{=}_R} = m$.
If it happens that we can find such $\gamma$ satisfying $X \simeq_R \gamma$, we can use the induction hypothesis I to confirm immediately that $\delta$ is $\mathcal{B}_{R,\mathbf{V}}$-applicable and $\gamma \stackrel{\mathcal{B}}{=}_R \delta$ (hence $X \stackrel{\mathcal{B}}{=}_R \delta$).
Thus in the following we will assume that $\gamma  \not\simeq_R X$. That is, $[X]_R \stackrel{\tau}{\longmapsto}_R \gamma  \not\simeq_R X$.
Since $X \simeq_R \delta$, there is a matching of $[X]_R \stackrel{\tau}{\longmapsto}\gamma$ from $\delta$, say $\delta \stackrel{\simeq_R }{\Longrightarrow} \cdot \stackrel{\not\simeq_R}{\longrightarrow} \zeta$ for some $\zeta$ such that $\gamma \simeq_R \zeta$. Because  $\delta$ is itself a $\simeq_{R}$-prime-decomposition, we must have $\delta \Longleftrightarrow \cdot \stackrel{\tau}{\longrightarrow} \zeta \simeq_R \gamma$, which is induced by  $[Z_t]_{R_t} \stackrel{\tau}{\longmapsto}_{R_t} \eta$ for some $\eta$ , and $\zeta = \eta.Z_{t-1}\ldots Z_1$.
Suppose $\mathtt{dcmp}^{\widehat{\mathcal{B}}}_{R_t}(\eta) = [Y_s]\ldots[Y_1]$ ($s \geq 0$), then $\mathtt{dcmp}^{\widehat{\mathcal{B}}}_{R}(\zeta)$ must be in the form $[Y_s]\ldots[Y_1]. [Z_{t-1}]_{R_{t-1}}.\ldots [Z_1]_{R_1}$.
Because $\gamma \simeq_R \zeta$, we have $\mathtt{dcmp}^{\widehat{\mathcal{B}}}_{R}(\gamma) =[Y_s]\ldots[Y_1]. [Z_{t-1}]_{R_{t-1}}.\ldots [Z_1]_{R_1}$. According to induction hypothesis I, $\gamma \stackrel{\mathcal{B}}{=}_R  Y_s\ldots Y_1.Z_{t-1}\ldots Z_1$, thus $\mathtt{dcmp}^{\widehat{\mathcal{B}}}_{R}(\zeta)$ is $\mathcal{B}_{R,\mathbf{V}}$-applicable, and  $\| \mathtt{dcmp}^{\widehat{\mathcal{B}}}_{R}(\zeta) \|_{\stackrel{\mathcal{B}}{=}_R} = \| Y_s\ldots Y_1.Z_{t-1}\ldots Z_1 \|_{\stackrel{\mathcal{B}}{=}_R} = m$. In summary:
\begin{itemize}
\item
$\|Y_s \ldots Y_1. Z_{t-1} \ldots Z_1 \|_{\stackrel{\mathcal{B}}{=}_R} = m$.

\item
$\| Z_{t-1} \ldots Z_1 \|_{\stackrel{\mathcal{B}}{=}_R} \leq m$.

\item
$\| Z_1 \|_{\stackrel{\mathcal{B}}{=}_R} > 0$.
\end{itemize}
Now we have two possibilities:
\begin{itemize}
\item
If $t\geq 2$.  Because $\| Z_1 \|_{\stackrel{\mathcal{B}}{=}_R} > 0$, thus $\| Z_{t-1 } \ldots Z_1 \|_{\stackrel{\mathcal{B}}{=}_R} > 0$. Let $S = \mathbf{Rd}^{\mathcal{B}}_{R}(Z_{t-1 } \ldots Z_1)$.  By induction $R_t \subseteq S$.  Then
 we have $\|Y_s \ldots Y_1\|_{\stackrel{\mathcal{B}}{=}_{S}} < m$. Since $\eta \simeq_{R_t} Y_s \ldots Y_1$, it is clear $\eta \simeq_{S} Y_s \ldots Y_1$, by Lemma~\ref{lem:pre_correctness}, $\eta$ is $\mathcal{B}_{S,\mathbf{V}}$-applicable, $\eta \stackrel{\mathcal{B}}{=}_{S} Y_s \ldots Y_1$, and  $\|\eta \|_{\stackrel{\mathcal{B}}{=}_{S}} = \|Y_s \ldots Y_1\|_{\stackrel{\mathcal{B}}{=}_{S}}  < m$.  Now let us investigate $Z_t$. If $Z_t \in S$, $Z_t \ldots Z_1$ is trivially $\mathcal{B}_{R,\mathbf{V}}$-applicable. If  $Z_t \not\in S$, we have got the fact that $[Z_t]_{S} \stackrel{\tau}{\longmapsto} \eta$ and $\|\eta \|_{\stackrel{\mathcal{B}}{=}_{S}} < m$.  This fact tells us that $[Z_t]_{S}$ should have been treated before $[X]_R$. In other words, $[Z_t]_{S} \in \mathbf{V}$, which means that $\delta = Z_t \ldots Z_1$ is $\mathcal{B}_{R,\mathbf{V}}$-applicable.

\item
If $t = 1$. In this case, we have the following facts: $[Z_1]_{R} \stackrel{\tau}{\longmapsto}_{R} \eta$,   $\mathtt{dcmp}^{\widehat{\mathcal{B}}}_{R}(\eta)= [Y_s]\ldots[Y_1]$, and $\|Y_s \ldots Y_1 \|_{\stackrel{\mathcal{B}}{=}_R} = m$.  We can show $\| Y_1 \|_{\stackrel{\mathcal{B}}{=}_R} > 0$, using the same argument for proving $\| Z_1 \|_{\stackrel{\mathcal{B}}{=}_R} > 0$ before.
Let us say $\eta = \eta'. W$ ($\eta'$ can be $\epsilon$), and let $S = \mathsf{Rd}_{R}(W)$.
Thus $\mathtt{dcmp}^{\widehat{\mathcal{B}}}_{R}(W)= [Y_i]\ldots[Y_1]$ and $\mathtt{dcmp}^{\widehat{\mathcal{B}}}_{S}(\eta')= [Y_s]\ldots[Y_{i+1}]$ for some $1 \leq i \leq s$. and $\eta'$ is $\mathcal{B}_{R,\mathbf{S}}$-applicable by induction.
\begin{enumerate}
\item
If  $\|Y_i \ldots Y_1 \|_{\stackrel{\mathcal{B}}{=}_R} < m$. we can use Lemma~\ref{lem:pre_correctness} to prove:
  \begin{itemize}
   \item
   $W$  is  $\mathcal{B}_{R,\mathbf{V}}$-applicable and $W \stackrel{\mathcal{B}}{=}_R Y_i\ldots Y_1$, and
   \item
   $S \subseteq \mathbf{Rd}^{\mathcal{B}}_{R}(W)$.
   \end{itemize}
Since we know $\eta' \simeq_S Y_s\ldots Y_{i+1}$,  then $\eta' \simeq_{\mathbf{Rd}^{\mathcal{B}}_{R}(W)} Y_s\ldots Y_{i+1}$. Because  $\|Y_s \ldots Y_{i+1} \|_{\stackrel{\mathcal{B}}{=}_{\mathbf{Rd}^{\mathcal{B}}_{R}(W)}} < m$, we can use induction to prove $\eta'$ is  $\mathcal{B}_{\mathbf{Rd}^{\mathcal{B}}_{R}(W),\mathbf{V}}$-applicable and $\eta' \stackrel{\mathcal{B}}{=}_{\mathbf{Rd}^{\mathcal{B}}_{R}(W)} Y_s \ldots Y_{i+1}$.  In summary, $\eta$ is  $\mathcal{B}_{R,\mathbf{V}}$-applicable and $\eta \stackrel{\mathcal{B}}{=}_{R} Y_s \ldots Y_{1}$.

\item
If $\|Y_i \ldots Y_1 \|_{\stackrel{\mathcal{B}}{=}_R} = m$.  In this case, $Y_s \ldots Y_{i+1} \in \mathbf{Rd}^{\mathcal{B}}_{R}(Y_i\ldots Y_1)$, this implies that $Y_s \ldots Y_{i+1} \Longrightarrow \epsilon$, and therefore
\[
 \eta  = \eta'.W \simeq_R Y_s \ldots Y_1 \Longrightarrow Y_i \ldots Y_1 \simeq_R  W
\]
which implies $\eta \Longrightarrow_R W$.   Now we have $Z_1 \Longrightarrow_R \eta \Longrightarrow_R W$, thus $[W]_R <_R [Z_1]_R$.  On the other hand, by $[Z_1]_R = \mathtt{dcmp}^{\widehat{\mathcal{B}}}_{R}([X]_R)$, we have $[Z_1]_R < [X]_R$. Therefore $W <_R Z_1 <_R X$.  By Lemma~\ref{lem:pre_correctness},  $W$  is  $\mathcal{B}_{R,\mathbf{V}}$-applicable and $W \stackrel{\mathcal{B}}{=}_R Y_i\ldots Y_1$.  Now in the same way of case~1, we can prove
$\eta'$ is  $\mathcal{B}_{\mathbf{Rd}^{\mathcal{B}}_{R}(W),\mathbf{V}}$-applicable and $\eta' \stackrel{\mathcal{B}}{=}_{\mathbf{Rd}^{\mathcal{B}}_{R}(W)} Y_s \ldots Y_{i+1} \stackrel{\mathcal{B}}{=}_{\mathbf{Rd}^{\mathcal{B}}_{R}(W)} \epsilon$. In summary, $\eta$ is  $\mathcal{B}_{R,\mathbf{V}}$-applicable and $\eta \stackrel{\mathcal{B}}{=}_{R} Y_s \ldots Y_{1}$.
\end{enumerate}
Up to now,  we have shown  that $[Z_1]_R <_R [X]_R$ and $[Z_1]_R \stackrel{\tau}{\longmapsto}_R \eta$ such that
$X \stackrel{\mathcal{B}}{=}_R \gamma \stackrel{\mathcal{B}}{=}_R \eta$ with $\| \eta\|_{\stackrel{\mathcal{B}}{=}_R} = m$.  This fact means that $[Z_1]_R$ should be chosen to test the expansion condition in the \textbf{while}-loop before $[X]_R$.  Now we can do without difficulty to check expansion conditions to ensure that $[Z_1]_R$ is put into $\mathbf{V}$ before $[X]_R$.
\end{itemize}
\end{itemize}

\subsection{Proof of Proposition~\ref{prop:correctness}}
We use $\delta$ to indicate $Z_t \ldots Z_1$.  By Lemma~\ref{lem:applicable}, $\delta$ is $\mathcal{B}_{R,\mathbf{V}}$-applicable. In other words, $\mathtt{dcmp}^{\mathcal{B}}_{R, \mathbf{V}}(\delta)$ is known.

The proof goes by directly exploring the expansion conditions. Only to remember the following fact:
\begin{enumerate}
\item
If $\alpha \simeq_R \beta$, then $\alpha \stackrel{\mathcal{D}}{=}_R \beta$.

\item
If $\alpha \simeq_R \beta$, and $\alpha, \beta$ are $\mathcal{B}_{R,\mathbf{V}}$-applicable, then $\alpha \stackrel{\mathcal{B}}{=}_R \beta$.
\end{enumerate}
By studying the the expansion conditions, we can confirm that $\mathtt{dcmp}^{\mathcal{B}}_{R, \mathbf{V}}(\delta)$ can successfully pass this testing.  Now it is important to take notice of Lemma~\ref{lem:unique_expansion}. It ensures that at most one decomposition candidate can pass the testing. Thus we can confirm that $\mathbf{Dc}^{\mathcal{B}}_R([X]_R) = \mathtt{dcmp}^{\mathcal{B}}_{R, \mathbf{V}}(\delta)$, which implies $X \stackrel{\mathcal{B}}{=}_R Z_t \ldots Z_1$.

\end{document}